\documentclass[11pt,letter,thmsa]{article}
\usepackage{amsmath,amssymb,amsxtra,amsfonts,amsthm,array,color,dcolumn,enumerate,ifpdf,euler,paralist,tikz,fancyhdr,float,comment,placeins}\synctex=1
\usepackage{rotating,booktabs}
\usepackage{adjustbox}
\usepackage[authoryear,round]{natbib}
\usepackage{marginnote}
\usepackage{booktabs}
\usepackage{graphicx}
\usepackage{soul}
\usepackage{subfig}
\usepackage{caption}
\usepackage{pdflscape}
\defcitealias{unicirdap1992}{UNICIRDAP}
\usepackage{anyfontsize}
\usepackage{multirow}
\usepackage{makecell}

\definecolor{myred}{RGB}{195,25,0}
\usepackage{threeparttable}
\usepackage{tikz}
\usepackage[hidelinks,colorlinks=false,citecolor=black,urlcolor=blue]{hyperref}
\usepackage[normalem]{ulem}
\usepackage{geometry}
\usepackage{tcolorbox}
\usepackage{wrapfig}

\usepackage[english]{babel}
\usepackage[utf8x]{inputenc}
\usepackage[colorinlistoftodos]{todonotes}
\usepackage{setspace}
\newtheorem{assumption}{Assumption}

\newtheorem{remark}{Remark}
\newtheorem{lemma}{Lemma}

\newtheorem{proposition}{Proposition}

\usepackage[outerbars,color]{changebar}
\definecolor{grey}{rgb}{0.75,0.75,0.75}
\definecolor{orange}{rgb}{1.0,0.5,0.5}
\definecolor{brown}{rgb}{0.5,0.25,0.0}
\definecolor{pink}{rgb}{1.0,0.5,0.5}
\begin{document}
\renewcommand{\baselinestretch}{1.2}\normalsize
\setlength{\textfloatsep}{8pt plus 2pt minus 2pt} 
\renewcommand{\arraystretch}{0.95}
\newcommand{\sym}[1]{\ifmmode^{#1}\else\(^{#1}\)\fi}


\title{Microfinance Competition in the Presence of Moneylenders: Theory and Evidence\thanks{
Author affiliations: Shyamal Chowdhury (Crawford School of Public Policy, The Australian National University, ACT 2600, Australia, corresponding author. Email: Shyamal.Chowdhury@anu.edu.au); Prabal Roy Chowdhury (Indian Statistical Institute); Joeri Smits (Department of Economics, Tufts University). We thank but do not implicate Samuele Centorrino, Hugh Harley, Isaac Loh, Fazlul Kader, Cynthia Kinnan, Sultan Orazbayev, Xuetao Shi, Sajjad Zohir, and audiences at the University of Sydney, BIDS, Netrokona University, and Australian National University for helpful comments. This experiment was conducted from 2012 to 2015, before the widespread adoption of pre-registration norms, and the study was registered in 2019 with the AEA RCT registry as AEARCTR-00042459. While the specific hypotheses were refined following data collection, the study design—including randomized variation in MFI competition stratified by baseline MFI presence—provides credible identifying variation to test theory-driven predictions. IRB approval was obtained from the University of Sydney’s Human Research Ethics Committee (2012/1240). Financial support from the Australian Agency for International Development (Agreement 53782) is gratefully acknowledged.
}}

\author{Shyamal Chowdhury, Prabal Roy Chowdhury and Joeri Smits}

\date{\today}

\begin{titlepage}

\maketitle

\abstract 

After decades of microfinance expansion and despite charging higher interest rates, moneylenders continue to exist alongside microfinance institutions (MFIs). We develop a dynamic model with ex post moral hazard, in which borrowers use MFI loans for production and can rely on moneylenders for bridge loans after shocks to repay MFIs and preserve future access to MFI credit. We then study how MFI competition shapes the coexistence between MFIs and moneylenders. The effect of competition on MFI credit is ambiguous: it may weaken repayment incentives or replace moneylenders as bridge lenders. By contrast, competition unambiguously reduces moneylender credit. In a village-level randomized expansion in Bangladesh, stratified by baseline MFI presence, one additional MFI reduces the number of moneylender loans by roughly 30–40\% two to three years later, with no detectable effect on MFI or overall borrowing. This suggests that, in markets where MFIs rely primarily on the threat of exclusion to induce repayment and moneylenders provide emergency finance, gains from additional MFI entry may be greater in less saturated areas.



 \vskip 0.5em

\thispagestyle{empty}

\end{titlepage}

\section{Introduction}

This research investigates a paradox in global development: the  resilience of the informal moneylender.    Despite the rise of Microfinance Institutions (MFIs)—whose borrower base grew from 9 million in 1997 to 120 million by 2017—moneylending remains a fixture in rural credit markets (see \citep{meyer2002demand}).\footnote{
 In Bangladesh -- our primary setting -- microfinance loans accounted for over 50\% of rural credit by 2011, with roughly one-third of rural households borrowing from MFIs \citep{hies2010}. Globally, the number of MFIs grew from 618 to 3,133 between 1997 and 2005 \citep{hermes2007empirics}.}  This persistence  contradicts the fact that the foundational mission of many MFIs -- most famously the Grameen Bank -- was the  ``liberation" of the poor from the  grip of predatory moneylenders (see, e.g., \citet{meyer2002demand}).   This tension motivates our enquiry: How does intensifying competition among socially motivated MFIs shape the continued relevance of moneylenders? We address this issue through an  integrated theoretical and empirical framework.

We model an infinite horizon framework featuring a single borrower, local moneylenders, and one or more MFIs. Lacking initial savings or assets, the borrower is entirely credit-dependent. While moneylenders are available, their interest rates are  prohibitively high, leaving MFI loans as the only viable path to production. However, MFI lending is plagued by \emph{ex post } moral hazard. Because project returns are unobservable and limited liability protects the borrower,  a borrower can strategically default by claiming project failure. The only credible enforcement mechanism is the threat of exclusion—cutting off access to future credit.\footnote{Use of credit exclusion as an incentive device is widespread in reality. \iffalse \textcolor{red}{Some examples that this is widespread would be good.} \fi}  This  creates a functional necessity for bridge loans to maintain repayment schedules during project failures, sourced either from moneylenders, or from other MFIs.

Our model incorporates two key empirical realities,  \emph{social motivation}, and \emph{double-dipping}.   First, we model MFIs as entities that maximize borrower utility rather than profit, see \citep{armendariz2010economics, besley2005competition}. This aligns with the mission of major NGOs like BRAC and Grameen, which frequently forgo profits to serve the poorest regions  (see \citep{armendariz2011mission}, and \citet{salim2013revealed}).\footnote{Roughly half of MFIs worldwide are NGOs, commonly defined as socially   motivated (\citealp{jia2023mission} UNICIRDAP, \citeyear{unicirdap1992}).} Second, we account for borrowers taking concurrent loans from multiple MFIs—a practice reaching approximately 40\% in Bangladesh (PKSF, 2006), and observed widely across the world.\footnote{Similar and even higher incidence is observed in other markets \citep{green2021dynamic}, including Mexico \citep{rozas2014microfinance}, Thailand and Vietnam   \citep{chichaibelu2017borrowing}, and  Ghana \citep{ewusie2021density}. }

We then analyse the effect of increased MFI competition  through three primary scenarios.  First, consider MFI entry  into  a previously unserved area. In the unique stationary  subgame perfect equilibrium (henceforth SSPE)  the MFI provides productive loans. However,  to avoid default during project failure, the borrower must turn to moneylenders for bridge loans. Thus paradoxically, the initial entry of an MFI creates a new demand for high-interest bridge financing to protec the borrower's MFI access.

 We then examine a scenario where there are two MFIs. Consistent with the facts that there is double-dipping, and that money is fungible, we  examine  two  SSPE, both of which involve double-dipping, but the MFI loans are used for different purposes. 
 Under \emph{ symmetric} SSPE where both MFIs provide an identical amount of productive loans, and  moneylenders provide bridge-loans, we find that  competition becomes self-defeating. As credit options proliferate, the `threat of exclusion'—the MFI’s only enforcement tool—loses its potency. Borrowers, anticipating that they can simply pivot to a new lender after a default, lose the incentive to maintain a clean record, leading to a systemic contraction of credit. This result highlights a functional complementarity between MFI and informal credit, echoing \citet{jain2003little}, who show that informal borrowing can finance MFI repayments. In our setting, however, this complementarity depends on dynamic exclusion incentives and can therefore unravel as MFI competition intensifies. Under an  \emph{asymmetric} SSPE, by contrast,  where one MFI specializes in productive loans,  while another provides lower-cost bridge  financing. In this case, we demonstrate that MFIs effectively crowd out high-interest moneylenders without eroding repayment discipline.  While this introduces another layer of moral hazard, in that the borrower may now default on both MFIs,  aggregate lending under this SSPE exceeds that under a single MFI whenever moneylender interest rates are large.

Our theoretical analysis provides   a nuanced resolution to the microfinance puzzle. While MFI entry initially bolsters moneylenders by creating a need for bridge financing, MFI competition can either reduce moneylending by eroding borrower discipline (under the symmetric SSPE), or displace it entirely through product specialization (under the asymmetric  SSPE). Ultimately, the impact of competition on total credit volume remains ambiguous, depending  on the functional roles filled by MFIs and moneylenders. compete for the same productive space,  or fill distinct functional roles.  The impact of MFI competition on the informal moneylender is however  unambiguously negative.

To test these competing theoretical predictions, we implemented a randomized controlled trial (RCT) in northern Bangladesh. 
 Partnering with RDRS, an MFI expanding into seven subdistricts, we randomized its entry across 146 villages (97 treatment, 49 control). Nearly all villages had at least one MFI at baseline. The randomized entry increased the number of MFIs in treatment villages by one, without affecting control villages, thereby generating exogenous village-level variation in MFI competition at the margin.

The increase in MFI competition had no statistically significant effect on the number of loans held by households or total amounts borrowed two to three years post-intervention. However, borrowing from moneylenders declined in both number and amount, supporting our central theoretical prediction. Specifically, the addition of one MFI reduced the number of moneylender loans per household by 28–40\%, from a baseline average of 0.25 loans. Estimated effects on moneylender interest rates were generally not significant. Although some village-level estimates were negative, they should be interpreted cautiously, as interest rates are only observed for households with outstanding moneylender loans.
This paper contributes to the literature in several important ways. First, it builds on theoretical work examining lender competition under information asymmetries.  Prior studies, such as  \citet{hoff1998moneylenders} and  \citet{mcintosh2005competition} focus on adverse selection, while   \citet{guha2013micro} and \citet{madestam2014informal} examine \emph{ex ante} moral hazard.\footnote{For models of competition among profit-maximizing lenders under adverse selection, see \citet{broecker1990credit,marquez2002competition,lester2019screening}, who show that increasing competition can worsen selection, distort screening, and reduce welfare.} In contrast, our model centres on strategic default, i.e. \emph{ex-post} moral hazard. This focus is supported by \citet{karlan2009observing}, who find in a rare experimental decomposition that ex-post moral hazard dominates adverse selection and ex-ante moral hazard in explaining default. Further,  we allow for two classes of lenders, socially motivated MFIs and moneylenders,  with the latter providing bridge loans in response to production shocks.

A second strand of literature, in particular \citet{ghosh2016information} and \citet{de2018commercialization},  analyzes the role of credit market competition on strategic default, when there is incomplete information sharing among the lenders. \citet{de2018commercialization} examine the effect of MFI commercialization (mission drift), as well as increased MFI competition  on borrower welfare in the presence of social capital and endogenously determined lending mechanisms.  \citet{ghosh2016information}, in contrast, analyse the issue of borrower heterogeneity, and how information may be revealed over time. The present paper differs from both in allowing for double-dipping by borrowers, two types of lenders (MFIs and moneylenders) offering distinct loan types in equilibrium, and then examining how an \emph{increase} in MFI competition affects this interaction.

Third, we contribute to the empirical literature on MFI expansion and informal lending. Existing studies offer mixed findings. A first set of papers focuses explicitly on MFI competition but relies on observational data. The literature finds that in Bangladesh, MFI competition leads to a decline in  moneylender lending (\citet{berg2020microfinance}), whereas  in Uganda, it  does not affect the amounts borrowed from moneylenders (\citet{mcintosh2005rising}). A second set of papers randomizes credit expansion for credible causal identification \citep{banerjee2015six}, but does not focus on or measure MFI competition. They show that microcredit expansion does not uniformly displace informal borrowing; in some settings, it falls slightly, while in others it increases or remains unchanged. However, those studies treat moneylenders and friends or family collectively as ‘informal lenders’, despite their very different character \citep{angelucci2015microcredit,attanasio2015impacts,banerjee2015miracle,crepon2015estimating,tarozzi2015impacts}. We bridge these two literatures by combining baseline data on MFI presence with cluster-randomized MFI expansion.\footnote{Tangentially related, Wydick (2016) notes that to maximize take-up contrasts and statistical power, many RCTs focus on marginal borrowers rather than estimating impacts for inframarginal borrowers.}

The rest of the paper is organized as follows: Section \ref{sec:theory} develops the model, Section \ref{sec:testable} derives the testable prediction, Section \ref{sec:empirics} tests it using an RCT, and Section \ref{sec:conclusion} concludes.

\section{The Framework}\label{sec:theory}

The economy is populated by one borrower,  imperfectly competitive profit-maximizing moneylenders   and $n$   motivated  MFIs,  $n\in\{1,2\}$.     Time is discrete and infinite ($t=1,2,…$). All agents are risk-neutral, infinitely lived, and share a common discount factor $\delta\in(0,1)$. The opportunity cost of capital is 1.  \smallskip
 
\emph{  The borrower.}  The borrower lacks initial wealth, collateral, and access to saving technologies; she must therefore rely on loans in every period to fund a project.\footnote{Only 15\% of our baseline sample households had formal or semi-formal savings (e.g., bank, NGO, savings club).} The project yields a stochastic return: $f(k)$ with probability $p$, and 0 with probability $1-p$, where $k$ represents the investment level. 

\begin{assumption}\label{f}  $f:[0,\infty)\rightarrow [0,\infty)$, and is twice differentiable at all $k>0$. Further, $f'(k)>0$, $f''(k)<0$, $f(0)=0$, and $\lim_{k\rightarrow \infty} f'(k)=0$. 
  \end{assumption}\noindent Given Assumption \ref{f}, the     first best level of investment $k^*=argmax_k\; pf(k)-k=f'^{-1}(\frac{1}{p})$ is well defined.\footnote{  $ k^*$  is well defined since,  from Assumptions \ref{f} and \ref{A2}, we have that  $f(.)$ is concave,   $\lim_{k\rightarrow \infty}f'(k)=0$, and $f'(0)>1+\frac{1-p}{p}m>\frac{1}{p}$. }

 \smallskip

\emph{MFIs:}   MFIs are  motivated entities that  aim to maximize borrower utility, leveraging subsidized funding from governmental and non-governmental sources to access capital at its opportunity cost of 1.  Their lending model is defined by three key characteristics. First, each MFI $i$  has a zero-profit requirement, yielding a net interest rate $r_i\ge 0$.  Second, an informational friction—stemming from the their position outside the local community—renders project outcomes unobservable and prevents contracts from being indexed to returns. Finally, the presence of limited liability rules out any immediate, contemporaneous penalties for borrower default.\footnote{This is reasonable given that the borrower has no income, or assets.}  \smallskip

\emph{Moneylenders}: Two factors define the cost of moneylender loans: capital costs and market power. First, without access to subsidized credit, the moneylenders' cost of capital is $1+c$. Second, existing literature,   e.g. \citet{aleem1990imperfect,hoff1998moneylenders,banerjee2011poor,mookherjee2016theory},    suggests that moneylenders operate in non-competitive environments,  extracting rents through a price-setting markup,  $m'>1$.    Thus, on a loan of $1$, the borrower has to repay $m'(1+c)$. For ease of exposition, let $m=m'(1+c)$.  

 \begin{assumption}\label{A2}  
 $$\frac{p(f'(0)-1)}{1-p}>m>pf'(0).$$\end{assumption}
 
 The assumption that $m>pf'(0)$  ensures that in the absence of MFIs,  the borrower takes no productive loans from the moneylenders (as formalised in Proposition \ref{NoMFI} later).  This reflects the well-documented reality of  ``usurious" informal rates; for example, \citet{mallick2012microfinance} notes a 26\% premium in West Bengal moneylender rates over MFI rates, while \citet{hoff1998moneylenders} highlight cases where informal rates reach double or triple the formal sector's cost. 
  Whereas the assumption that    $\frac{p(f'(0)-1)}{1-p}>m$ ensures that moneylenders serve a functional role in the economy—specifically for ``bridge" financing.\footnote{Formally, this ensures that what we call the second best level of investment $\overline k$ is well defined
  (see \eqref{bark}).}  This upper bound ensures that while moneylender loans are expensive, they are not so prohibitive that they serve no purpose in this economy.

 Moreover,  we assume that   the borrower cannot default on moneylender loans.\footnote{In a classic study of the Chambar region, Aleem (1990)  found that while informal interest rates were high (averaging 78.5\%), the mean default rate was only 2.7\%. Similarly, Siamwalla et al. (1990) showed that rural moneylenders had much higher recovery rates than formal government institutions. } This assumption is grounded in the high repayment rates observed in informal credit markets and stems from three primary mechanisms,  in particular  inter-linked markets, informational advantage, and local influence. As noted by Bell (1988) and Bardhan (1980),  moneylenders often operate in inter-linked markets—acting simultaneously as landlords, employers, or crop traders.   This   allows the them  to deduct repayments directly from wages or harvests.
Second, moneylenders possess a distinct informational advantage due to their proximity to the borrower. They can monitor a borrower’s project success in real-time and accurately gauge their liquidity.
    Finally, local influence arises from the greater wealth and status of the moneylenders.\footnote{Karaivanov and  Kessler (2018) argue that in the event of a default, the borrower risks losing their reputation, social capital, and future support from the entire community.}  The convergence of these factors ensures that the moneylender can enforce repayment far more effectively than formal external institutions.

 \smallskip

\emph{Informational issues:}  The core informational issue is ex-post moral hazard, driven by the intersection of limited liability and  ``information asymmetry"—specifically, the MFIs' inability to verify project outcomes. Since success is unobservable, loan renewals cannot be tied to actual performance. This creates a loophole: a borrower can strategically default by claiming failure without fear of immediate legal or financial reprisal. Consequently, the MFIs' only leverage is the threat of credit termination in the future.\footnote{Other possible threats could involve cutting off access to services that are bundled with the loan, e.g., training;  building of social capital via regular interactions between all members and MFI agents \citep{feigenberg2013economic}. } This dynamic inadvertently  creates a role for moneylenders, since, in case of project failure, the borrower must take bridge-loans from the moneylenders if she wants to maintain access to future loans.

 Finally, we assume that any given MFI can only observe whether the borrower has defaulted on its own loans or not, but not whether she had defaulted on the other MFIs. In this we follow \citet{ghosh2016information}, as well as \citet{de2018market}. Credit information sharing remains extremely limited in many developing countries: as of 2019, over 40 such countries  lacked a credit bureau, and coverage rates were often below 10\% where such systems existed \citep{worldbank2019doingbusiness}. 

 \smallskip

\emph{Timeline.} 
Let  the loan contract of the $i$-th MFI specify the loan amount, $k_{ti}$, and the net rate of interest, $ r_{ti}$.   We next define the time-line at any $t$:

 \begin{enumerate}  
 
 \item []  \emph{Stage 1:  Contracting. }  All active MFIs  simultaneously offer  loan contracts. 
   The total amount borrowed, call it $B_t$,   is invested.\footnote{
 Note that the game structure does not allow for the possibility that the borrower takes productive loans from both MFIs and moneylenders. This is because given Assumption \ref{A2}, $pf'(k)<m$ for all $k$, so that   the borrower has no incentive to do so.}

\item [] \emph{Stage 2: Production.}  Project returns accrue randomly,  $f(B_t)$ in case of success, and 0 otherwise.

\item [] \emph{Stage 3: Settlements of Production Loans.}   The borrower decides on  how many of the MFIs to repay. If the borrower defaults on any MFI that she took a loan from in this period, then   this MFI cuts off access to all future loans.   \mbox{}
 \begin{itemize}
\item [(i)] In case the project fails,   the borrower  takes a bridge loan for the requisite amount.  
\item [(ii)]  In case the project is successful, the borrower repays out of project returns.\end{itemize}
\item [] \emph{Stage 4: Settlements of Bridge-loans.}  The borrower repays the bridge-loan in  case it was from a moneylender. In case the bridge-loan was from MFIs,  the borrower decides on how many of such MFIs to repay.\footnote{At this stage the borrower either gets  income from some other source, say agriculture, or possibly arranges for some asset sale, using the additional time to get more favourable prices for these assets, allowing her to repay the bridge-loan.} repay. If the borrower defaults on any MFI that she took a bridge-loan loan from in this period, then the borrower receives no more loans from this MFI in the future.   \end{enumerate}

The solution concept is stationary   subgame perfect Nash equilibrium (henceforth SSPE).\medskip

$\square$ \emph{A Benchmark: No MFIs}: Suppose there are no MFIs, so that the only source of funding are the moneylenders.  Recall that moneylenders can ensure loan repayment, so that  repayment incentives, in particular the  dynamic incentives of the kind that we analyze later, are not a concern. Thus, the optimization problem facing the household  may be modelled as a static one:
\begin{equation}
\max_k\, pf(k) -mk.\end{equation} Given that $m>pf'(0)$  from  Assumption \ref{A2},  for any positive $k$  the marginal profit from the project 
$
pf'(k)-m$ is negative.  \begin{proposition}\label{NoMFI}  In the absence of MFIs, the borrower does not take a productive loan from the moneylenders.\end{proposition}

Proposition \ref{NoMFI} seems consistent with reality. For example,  \citet{mallick2012microfinance} argues that  in  Bangladesh  moneylender loans were made for productive purposes in only 10\% of  sample villages.  In Thailand, \citet{menkhoff2011village} and \citet{chichaibelu2017borrowing} also argue that informal sector loans are  used more for shock absorption.

\section{ A Single MFI}\label{sec:testable}

We first analyse the  case where there is a single MFI, call it M1.
We characterize  SSPE where,  in every period along the equilibrium path:

(i) M1 provides  a loan of $k$ at a net interest rate of $r$  to the borrower, and 

(ii)   the borrower repays M1 
  using project returns upon success, or via a bridge loan of $(1+r)k$ from a moneylender in the event of project failure.  Thus there is no default. 
\vskip 0.25em 

The borrower's one-shot payoff  along such an SSPE  is\begin{equation}
Z(k,r) \equiv  p\left[f(k)-(1+r)k)\right]-(1-p)m(1+r)k.\label{Z}\end{equation}  Given that M1 is motivated her objective is to maximize the present discounted value of the borrower's expected income stream. Thus M1's   optimisation exercise  is
 \begin{equation}\label{Coop1}
\max_{k,r}\, \frac{1}{1-\delta} Z(k,r),\end{equation} 
subject to the  following no default and break-even constraints:
\begin{eqnarray}
ICSuccess &:&f(k)-(1+r)k+\frac{\delta}{1-\delta} Z(k,r) \ge f(k),\label{D1}\\
 ICFailure&:& -m(1+r)k+\frac{\delta}{1-\delta} Z(k,r) \ge 0,\label{D2}\\
 Breakeven &:& r\ge 0.\label{BE}
\end{eqnarray}
\eqref{D1} is the no default condition when there is project success.  The LHS of \eqref{D1}, $f(k)-(1+r)k+\frac{\delta}{1-\delta} Z(k,r)$,  is the  borrower's expected payoff in case she repays. This is because she continues to obtain further productive loans from M1, having an  income of $Z(k,r)$ in every subsequent period. Whereas the RHS of \eqref{D1},  $f(k)$, is her continuation payoff  if she defaults. This follows  since  the borrower obtains  no further loans from M1, and,  given that $pf'(0)>m$ (Assumption \ref{A2}), she has no incentive to take productive loans from the moneylenders.    The no default condition  in case of project failure is given by \eqref{D2}. The LHS of \eqref{D2} is the   the borrower's payoff if she repays by taking  a bridge-loan  of $(1+r)k$  at a cost of $m(1+r)k$ from the moneylenders, when she continues to get productive loans from M1.  Whereas the RHS of \eqref{D2},  $0$, is her continuation payoff from default.    Finally, \eqref{BE} ensures that  M1 makes a non-negative payoff from making the loan. \smallskip

To characterise  the optimal loan package, call it  $(\widehat k(1),\widehat r(1))$, we require some   more notations:\begin{enumerate}
\item []  Let   $\overline k$   maximise \eqref{Coop1} when $r=0$, i.e. $Z(k,0)$:
\begin{equation}\label{bark}
 \overline k=f'^{-1}(1+\frac{1-p}{p}m).\end{equation} We call $\overline k$ the \emph{ second best level of investment}.\footnote{We call $\overline k$ the \emph{second} best level of investment  since the borrower is  taking bridge-loans from the moneylenders at an interest  rate that exceeds the oppurtunity cost of capital.}$^,$\footnote{$\overline k$  is well defined since,  from Assumptions \ref{f} and \ref{A2}, we have that  $f(.)$ is concave,   $\lim_{k\rightarrow \infty}f'(k)=0$, and $f'(0)>1+\frac{1-p}{p}m$.} Note that  $\overline k$ is less than the first best $k^*$.\footnote{This follows since $k^*=f'^{-1}(\frac{1}{p})>f'^{-1}(1+\frac{1-p}{p}m)=\overline k$, where the inequality follows from the concavity of $f(k)$, and the fact that $m>1$.}

\item [] Next let $\widehat \delta $ solve\footnote{$\widehat \delta$ is well defined since $\lim_{\delta\rightarrow 0}\frac{\delta}{1-\delta}=0$,  $\lim_{\delta\rightarrow 1}\frac{\delta}{1-\delta}\rightarrow \infty,$ and $\frac{\delta}{1-\delta}$ is increasing in $\delta$.}
 \begin{eqnarray}
\frac{\delta}{1-\delta}Z'(0,0)=  m.
\end{eqnarray} We later demonstrate that the optimal loan amount is positive iff $\delta>\hat \delta$. This is because repayment is enforced  by cutting off access to \emph{future} credit. With $\delta$ small, this threat loses its potency, and a positive level of investment cannot be sustained.

\item  []   Finally, let  $\underline k(1,\delta)$ satisfy the no default condition in case of project failure, i.e. \eqref{D2}, with equality: \begin{eqnarray}
\frac{\delta}{1-\delta}\frac{Z(\underline k(1,\delta),0)}{\underline k(1,\delta)}=  m.\label{D2''}
\end{eqnarray} 
We call $\underline k(1,\delta)$    the \emph{incentive constrained level of investment} when there is a single MFI.  $\underline k(1,\delta)$  is well defined whenever $\delta>\widehat \delta$.\footnote{Note that  $\frac{Z(k,0)}{k}=\frac{pf(k)}{k}-(p+(1-p)m)$.  Thus  $\frac{\delta}{1-\delta}\lim_{k\rightarrow \infty} \frac{Z(k,0))}{k} = -\frac{\delta}{1-\delta}(p+(1-p)m)<m$,  since from Assumption \ref{f}, $\lim_{k\rightarrow \infty} \frac{f(k,0))}{k}=0$.   Next,  given that $\frac{\delta}{1-\delta}Z'(0,0)>m$, the  existence of $\underline k(1,\delta)$  follows from the intermediate value theorem.} 

\end{enumerate} 
\eqref{D2''} has an intuitive interpretation.  Suppose the project fails, and the borrower decides not to repay. Such a deviation has both costs and benefits. The cost arises from the fact that the borrower has  no access to MFI  loans  from the next period onwards. The LHS of \eqref{D2''} is the present discounted value of this loss averaged out over the  loan amount.  Whereas the RHS of \eqref{D2''} is the per unit  gain from not repaying, which is the interest rate charged by moneylenders $m$. 

Proposition \ref{M1} below   demonstrates that   $\widehat r(1)=0$ and, moreover, $\widehat k(1)=\min\{\overline k, \underline k(1,\delta)\}$ whenever $\delta$ is not too small.  While the formal proof can be found in the appendix, the result is intuitive. 
A reduction in  the rate of interest  $r$ not only increases borrower utility, it also relaxes the no default conditions as the borrower has to repay a smaller amount. With M1 being benevolent, it selects the minimal $r$ that is consistent with breaking even, which in this framework is  zero.    Next consider the equilibrium level of investment.  Motivated by benevolence, M1 targets the second-best investment level, $\overline k$, provided it satisfies the no-default constraints. Otherwise   M1  implements the maximum value of $k$ consistent with zero default -- in case of project success, as well as project failure.   Out of the two no default conditions,  the one  in case of failure, i.e.  \eqref{D2},  is more stringent  since, given  the prohibitive cost of bridge loans from moneylender, repayment  during project failure is relatively  costlier.  Hence the optimal investment in this case is $\underline k(1,\delta)$, since, given that $\widehat r(1)=0$, the maximal $k$ that satisfies \eqref{D2} is  
$\underline k(1,\delta)$.

\begin{proposition}\label{M1} Let Assumptions \ref{f} and \ref{A2} hold. \begin{enumerate}
\item If  $\delta \le \widehat \delta$, then there exists  no SSPE with a positive level of investment.
\item Let $\delta > \widehat \delta$. There is a unique SSPE such that    \begin{enumerate}
\item the   interest rate  ensures that M1  breaks even, i.e.
$\widehat r(1)= 0,$
\item  there exists $\widetilde \delta (1)$, $0<\widetilde \delta(1)<1$, such that  the aggregate loan amount $\;\widehat {k}(1)$  equals the second best level $\overline k$ if $\delta \ge \widetilde \delta(1)$, and equals the incentive constrained level $\underline k(1,\delta)$ otherwise, and 
\item   the  borrower takes a bridge-loan of \,$\widehat k(1)$ from one of the moneylenders whenever the project fails. \end{enumerate} 
\item The optimal investment $\widehat k(1)$ is strictly  less than the first best $k^*$.
 \end{enumerate}  \end{proposition}

\section{MFI Competition}\label{n2}

We then examine how increased MFI competition affects the amount of MFI and moneylender loans. Consider a scenario where there are two MFIs,  M1 and M2.\footnote{In Remark \ref{n} later we allow for the possibility that the number of MFIs can exceed two, demonstrating that the analysis is robust to this extension.} We shall examine two different kinds of  SSPE, one where  
both MFIs provide an equal amount of productive loans, with the moneylenders providing bridge-loans, and another one where only one of the MFIs provides productive loans, with the other one  providing bridge-loans.  For ease of reference we call the first SSPE symmetric,  and   the second one asymmetric. 

\subsection{Symmetric Equilibrium with Two MFIs}\label{sym}

We first   characterise symmetric  SSPE where, in every period, along the equilibrium path:

(i) both MFIs offer a loan package $(\frac{k}{2},r)$  to the borrower, 

(ii)  there is no default, with repayment  obligations being met by using  project-generated cash flows   upon success,  and by using  bridge loans from moneylenders   in the event of project failure. 

 We shall demonstrate that,  under this equilibrium, aggregate lending by both MFIs and moneylenders is lower than in the case of a single MFI. 
We introduce two notations:\begin{itemize}

\item []  First, for any   $x\in\{0,1,2\}$, let $\widetilde Z(\frac{xk}{2},r)$  denote  the one-shot payoff of the borrower if she has loans from $x$  MFIs,   all  offering  a loan package $(k/2, r)$, and she repays all $x$ MFIs (taking bridge loans from moneylenders if necessary):
\begin{equation}
\widetilde Z(\frac{xk}{2},r) \equiv  p\left[f(\frac{ xk}{2})-(1+r)(\frac{ xk}{2})\right]-(1-p)m(1+r)(\frac{ xk}{2}).\label{tildeZ}\end{equation}  
Note that $\widetilde Z(k,r)=Z(k,r)=pf(k)-k(1+r)[1+(1-p)m]$. Hence $argmax_k\,\widetilde Z(k,0)=\overline k$.

\item []  Second, whenever  $\delta>\widehat \delta$,  let  $\underline k(2,\delta)$ satisfy:  \begin{eqnarray}
\frac{\delta}{1-\delta}\frac{\widetilde Z(k,0) -\widetilde Z(\frac{k}{2},0)}{k-\frac{k}{2}}= m. \label{XX}
\end{eqnarray} \noindent 
We call $\underline k(2,\delta)$    the \emph{incentive constrained level of investment} under a symmetric SSPE. \end{itemize}

 In Proposition \ref{M2:Coop} we show that for a sufficiently large $\delta$, there exists a unique symmetric SSPE  $(\widehat k(2),\widehat r(2))$. In this equilibrium  both M1 and M2 offer   $(\frac{\widehat k(2)}{2},\widehat r(2))$,   with  $\widehat r(2)=0$ and  $\widehat k(2)=\min\{\overline k, \underline k(2,\delta)\}$.   For  $(\widehat k(2),\widehat r(2))$ to constitute an SSPE,  it must be that  M2, say,  cannot  gain from making a one-shot deviation   to  some other loan package   $(k_2,r_2)$, reverting to  the candidate equilibrium strategy subsequently. 
  The objective of M2  is  to maximize
\begin{eqnarray} p\left[f(k_2+\frac{1}{2}  \widehat k(2))-k_2(1+r_2)-\frac{  \widehat k(2)}{2}\right]
- (1-p)m\left[k_2(1+r_2)+\frac{\widehat k(2)}{2}\right] 
+ \frac{\delta}{1-\delta}\widetilde Z(\widehat k(2),0),\label{objno}\end{eqnarray} subject to the various   no default conditions of the borrower,   and the breakeven constraint of M2.  

As in the single MFI case,    $r_2$ equals  zero since the M2 is benevolent. Further, given that moneylender loans are expensive ($m>1$), it is sufficient to consider the no-default conditions  when the project fails.  Hence suppose the project fails. Consider the  no default conditions  such that the borrower wants to repay all the MFIs, rather than repay only  $x$  of them,  $0\le x\le 1$, defaulting on the others, \emph{including} on   M2:\footnote{For $x=0$,  the RHS of \eqref{IC} below  denotes the the borrower's expected income when she defaults on both MFIs, wheras for $x=1$, the RHS of \eqref{IC} is the borrower's expected income when she defaults on M2 alone.}$^,$\footnote{Given symmetry, one can show that these ensure that the no default condition  such that the borrower wants to repay all the MFIs,  rather than default on M1 alone,  also hold.}
\begin{eqnarray}  
\hskip 0em -mk_2(1+r_2)-m\frac{1}{2}  \widehat k(2)+\frac{\delta}{1-\delta}\widetilde Z(\widehat k(2),0)\ge    
-m \frac{x}{2}  \widehat k(2)+\frac{\delta}{1-\delta}\widetilde Z(\frac{x\widehat k(2)}{2},0),\forall x\le 1.\label{IC}\end{eqnarray}    
One critical insight is that it is sufficient to consider the case where $x=1$.   The reasoning is strategic: defaulting on a single MFI is  more  ``attractive"  than defaulting on both, as it allows the borrower to maintain a credit relationship with the remaining MFI.  Thus it is sufficient to consider the following no default condition: \begin{eqnarray}
\frac{\delta}{1-\delta}\frac{\widetilde Z(k,0) -\widetilde Z(\frac{k}{2},0)}{k-\frac{k}{2}}\ge  m, \label{XX'}
\end{eqnarray}
where, as in \eqref{D2''},  the LHS denotes the per unit cost, and the RHS denotes the per unit gain from default.  Thus the aggregate MFI investment $\widehat k(2)$ equals the second best level $\overline k$ if it does not violate \eqref{XX'}. Otherwise,   the MFI chooses the maximal  aggregate investment level    that ensures that \eqref{XX'}  is not violated, i.e. $\underline k(2,\delta)$.

\begin{proposition}\label{M2:Coop} Let Assumptions \ref{f} and \ref{A2} hold. 
\begin{enumerate}  \item Suppose that  $\delta\le \widehat \delta$.  Then  no  symmetric SSPE  where the investment level $k$ is positive exists.
\item Suppose  that  $\delta> \widehat \delta$.  Then there is  a unique symmetric SSPE  such that 
 \begin{enumerate}
\item    all MFIs charge the minimal interest rate that ensures break-even, i.e.
$\widehat r(2)= 0,$ and 
\item  there exists $\widetilde \delta (2)$, $0<\widetilde \delta(2)<1$, such that  the aggregate loan amount $\;\widehat {k}(2)$  equals the second best level $\overline k$ if $\delta \ge \widetilde \delta(2)$, and equals the incentive constrained level $\underline k(2,\delta)$ otherwise, and
\item   the  borrower takes a bridge-loan from one of the moneylenders whenever the project fails. \end{enumerate}
\item The optimal investment $\widehat k(2)$ is less than the first best $k^*$.
 \end{enumerate}  \end{proposition}

$\square$ \emph{Effect of an increase in MFI competition}: Next, we compare  lending levels under a two-MFI symmetric SSPE with those  in the single-MFI case.

\begin{proposition}\label{Compp}  Let Assumptions \ref{f} and \ref{A2} hold.  Compared to the single MFI case, both aggregate MFI and moneylender loans are smaller under the symmetric SSPE (in Proposition \ref{M2:Coop}). These loan amounts are strictly smaller if  $\delta>\widehat \delta$.\end{proposition}

Intuitively, default is very costly when there is a single MFI since following a default the borrower gets no further MFI loans. Whereas with two MFIs, the borrower can default on one MFI, and can still continue to get loans from the other MFI. Thus default incentives are larger when the number of MFIs equals two.  Why does a reduction in aggregate loan amount reduce such default incentives? Given that   the per unit cost of default is  $ \frac{\widetilde Z(k,0) -\widetilde Z(\frac{k}{2},0)}{k-\frac{k}{2}}$ (see \eqref{XX'}), from the concavity of $\widetilde Z(k,0)$ this cost increases when the aggregate loan amount decreases. Consequently aggregate MFI loans are smaller under competition.

These results are consistent with \citet{de2018market} who argue that the amount borrowed from MFIs are larger if there is a single MFI rather than perfect competition among MFIs.  Our results differ in that we  focus on the case where the MFIs are not atomistic, thus allowing us to analyze the effect of a \emph{gradual} increase in MFI competition. Further, we allow for double-dipping and do not assume that lenders can commit to their contracts, aspects that appears to be empirically relevant.

\subsection{Asymmetric Equilibrium with Two MFIs}\label{sec:MFIBL}\label{asym}

We next   characterise  \emph{asymmetric} SSPE where the two MFIs play very different roles. In particular  along the equilibrium path, in every period:

(i) the borrower  obtains a loan of $k$ at a net interest rate of $r_1$  from M1, and 

(ii)  the borrower  repays   M1,  using the  project returns  in case of  project success, and   a bridge-loan of $(1+r_1)k$  at a net interest rate of $r_2$ from M2   in case the project fails. \vskip 0.25em

Under this  asymmetric SSPE    the per period payoff of the borrower is given by 
\begin{equation}
Z^*(k,r_1,r_2)=p[f(k)-(1+r_1)k]-(1-p)(1+r_1)(1+r_2)k.\end{equation}  We shall demonstrate that, compared to the case where there is a single MFI,  aggregate MFI loan  is larger under the optimal asymmetric SSPE.



Note that under  the asymmetric SSPE there is an additional layer of ex-post moral hazard: the risk of strategic default on M2. Specifically, if the borrower defaults on M1, she  loses access to all future productive loans from that institution. This eliminates the subsequent need for bridge loans from M2. Consequently, the borrower’s optimal strategy is to extract an additional $k(1+r_1)$  from M2	
  and immediately default. The  no-default conditions discussed next  are  adjusted to account for this default risk.  
  
  As in Propositions \ref{M1} and \ref{M2:Coop}, given that $m>1$,  it is sufficient to consider the no default conditions under project failure:
\begin{eqnarray}
 ICFailure1&:& -(1+r_1)(1+r_2)k+\frac{\delta}{1-\delta} Z^*(k,r_1,r_2) \ge (1+r_1)k,
 \label{F2}\\
  ICFailure2&: &-(1+r_1)(1+r_2)k+\frac{\delta}{1-\delta} Z^*(k,r_1,r_2) \ge \frac{\delta}{1-\delta} Z(k,r_1).\label{F3}
\end{eqnarray}
 The LHS of both \eqref{F2} and \eqref{F3}  denote the borrower's expected payoff from taking a bridge-loan of $(1+r_1)k$ from M2, and using it to repay M1, and later repaying M2 an amount of $(1+r_1)(1+r_2)k$, when she continues to get  loans from both MFIs in the future. The RHS of \eqref{F2} denotes the borrower's payoff if she  defaults on both M1 and  M2. Whereas the RHS of \eqref{F3}  denotes her payoff if she takes a bridge-loan from M2, uses it to repay M1, but defaults on M2. In this case   she continues to get loans from M1 alone from the next period onwards.

Thus the optimisation problem facing M1    is $\max_{k,r_1}\, \frac{1}{1-\delta} Z^*(k,r_1,r_2),\label{Max'''}$ subject to \eqref{F2}, \eqref{F3} and $r_1\ge 0$. Similarly, the one facing M2 is $\max_{r_2}\, \frac{1}{1-\delta} Z^*(k,r_1,r_2),\label{Max'''}$ subject to \eqref{F2}, \eqref{F3} and $r_2\ge 0$.  We introduce two notations before characterising  the optimal loan package  $(\widehat k^*(2),\widehat r_1^*(2),r_2^*(2))$:\begin{enumerate}
\item [] Let $\widehat \delta^*$ solve 
\begin{equation}\frac{\delta}{1-\delta} Z^*(0,0,0) = 2.\end{equation} 
We later demonstrate that  the equilibrium $\widehat k^*(2)$ is positive iff $\delta>\widehat \delta^*$.
\item []   Whenever $\delta>\widehat \delta^*$,  let $\underline k^*(2,\delta)$ solve 
\begin{equation}\frac{\delta}{1-\delta} \frac{Z^*(k,0,0)}{k} = 2.\end{equation} 
We call $\underline k^*(2,\delta)$ the \emph{incentive constrained level of investment }under the asymmetric SSPE. 

\end{enumerate}

Proposition \ref{Mbl} below   demonstrates that   $\widehat r_1^*(2)=\widehat r_2^*(2)=0$, a direct consequence of the MFIs' benevolence. Furthermore,  the optimal investment level    $\widehat k^*(2)=\min\{k^*, \underline k^*(2,\delta)\}$, whenever  the discount factor $\delta$  is sufficiently large.  Clearly, $\widehat k^*(2)$ equals the first best level $k^*$ if \eqref{F2} and \eqref{F3} remain slack; otherwise, one of the no default conditions would bind. 
Regarding condition (17), since the volume of the productive loan from M1
  is independent of whether the borrower defaults on M2 or not,   the constraint is independent of k and simplifies to: 
\begin{equation}
\delta\ge \delta^0\equiv \frac{1}{1+(1-p)(m-1)}.\end{equation} Hence, assuming that $\delta\ge \delta^0$, so that $k^*$ violates \eqref{F2},   the MFI chooses the maximal  $k$   that satisfies \eqref{F2}, i.e. $\underline k^*(2,\delta)$.

\begin{proposition}[Two MFIs, one providing productive, and the other providing bridge loans] \label{Mbl}  Let Assumptions \ref{f} and \ref{A2} hold.   \begin{enumerate}
\item Suppose $\delta \le \max\{\delta^0,\widehat \delta^*\}$. Then there is  no  asymmetric SSPE where the investment level  is positive.
\item Suppose $\delta >\max\{\delta^0, \widehat \delta^*\}$.    Then  there is a unique  asymmetric SSPE $(\widehat k^*(2),\widehat r_1^*(2))$  such that: 
 \begin{enumerate}
\item    both MFIs charge the minimal interest rate that ensures break-even, i.e.
$\widehat r_1^*(2)=\widehat r_2^*(2)= 0,$
\item  there exists $\widetilde {\delta}^* (2)$, $0<\widetilde {\delta}^*(2)<1$, such that  the  loan amount $\widehat {k}^*(2)$  equals the first best level $ k^*$ if $\delta \ge\widetilde {\delta}^*(2)$, and equals the incentive constrained level $\underline k^*(2,\delta)$ otherwise, and 
\item in case of project failure the borrower takes a bridge-loan from M2.
\end{enumerate}
\item The investment level $\widehat k(2)$ is weakly less than the first best level, with the two being equal whenever $\delta\ge \widetilde {\delta}^*(2)$.
 \end{enumerate}  \end{proposition}

\vskip 0.5em
\begin{remark} Given their superior local information and operational flexibility, moneylenders are often better positioned to provide bridge financing. Consequently, under a staggered repayment schedule, borrowers may logically source initial instalments from moneylenders before transitioning to MFIs for later tranches.  One can formalise this by assuming that,  in case of project failure, a fraction $\beta$ of the bridge-loan is taken from MFIs, and $(1-\beta)$ from moneylenders, where $0\le \beta\le 1$.  While our broader analysis -- not reported here -- confirms that aggregate MFI lending increases even under this extension,  for ease of exposition we focus on the case where $\beta=1$.\footnote{The general analysis is available on request.}\end{remark}

$\square$ \emph{Effect of an increase in MFI competition}: Next, we compare lending levels under a two-MFI asymmetric SSPE with those  in the single-MFI case.

\begin{proposition}\label{Mbl:comp} Let Assumptions \ref{f} and \ref{A2}  hold.   Whenever $m>2$,   the productive loan  amount under the asymmetric SSPE, $\widehat k^*(2)$,  exceeds $\widehat k(1)$, the productive loan amount under a single MFI. Further, there is no moneylender loan under the asymmetric SSPE.\end{proposition}  

Intuitively, when $m$ is large, in particular $m>2$, the advantages of substituting moneylender bridge loans with low cost MFI alternatives become more pronounced.  Moreover, with $m$ large, the ex post moral hazard problem under single MFI lending is also large. For both these reasons $\widehat k^*(2)$  exceeds $\widehat k(1)$ (see Figure 1)..

\begin{remark} Interestingly, given the  the additional moral hazard inherent in asymmetric SSPE,    aggregate lending may actually be higher under a monopoly than under competition. This  can occur   when $m$ is not too large, formally    $\widehat \delta^*> \widehat \delta$.  In that case    there exists $\underline \delta$,  $\widehat \delta<\underline \delta$, such that     $\widehat k^*(2)$ is less than $\widehat k(1)$ whenever $\widehat \delta<\delta< \underline \delta$ (see Figure 2).\end{remark}

$\square$ \emph{Testable Hypothesis}:  Both Propositions \ref{Compp} and \ref{Mbl:comp} suggest that increased MFI competition reduces moneylender loans, though for different reasons. This leads to our testable hypothesis.   \vskip 1em
 
 \textbf{ H1:} An increase in MFI competition reduces moneylender lending. \vskip 1em

By contrast, the model yields no unambiguous prediction for MFI lending: competition reduces MFI credit under the symmetric SSPE but increases aggregate MFI lending under the asymmetric SSPE. Accordingly, we do not formulate a directional hypothesis for MFI borrowing.

 \vskip 1em

We discuss a few robustness issues next. Remark \ref{n} examines the robustness of the analysis to allowing for mor than two MFIs.   Given the recent debate over possible  mission drift among  (some) MFIs \citep{mersland2010microfinance,armendariz2011mission}, Remark \ref{profits} examines partially profit oriented MFIs.    \smallskip

 \begin{remark}[Extending the analysis to  $n>2$]\label{n}       For a general $n$, we first consider a symmetric SSPE where  every MFI offers a loan package $(\widehat k(n)/n, 0)$.   Let MFI $i$ make a one shot deviation to   $(k_i,r_i)$ before reverting to its  equilibrium strategy.  Its optimisation problem is 
\begin{eqnarray}\max_{k_i,r_i} \;  p\left[f(k_i+\frac{(n-1)}{n}  \widehat k(n))-k_i(1+r_i)-\frac{(n-1)}{n}  \widehat k(n)\right]
&-& (1-p)m\left[k_i(1+r_i)+\frac{(n-1)}{n}  \widehat k(n)\right] \nonumber\\
&+& \frac{\delta}{1-\delta}Z(\widehat k(n),0)\label{objno}\end{eqnarray} subject to its no default,   and the breakeven constraints.

As in the two MFI case,  it is sufficient to consider   the incentive constraints when the project fails.
First consider the incentive constraints such that the borrower wants to repay all the MFIs, rather than only  $x$  of them,  $0\le x\le n-1$, defaulting on the others, \emph{including}  the $i$-th MFI:
\begin{eqnarray}  
\hskip 0em -mk_i(1+r_i)-m\frac{(n-1)}{n}  \widehat k(n)]+\frac{\delta}{1-\delta}Z(\widehat k(n),0)\ge    
-m \frac{x}{n}  \widehat k(n)+\frac{\delta}{1-\delta}Z(x\widehat k(n)/n,0),\;\forall x.\label{IIC}\end{eqnarray}    
Similarly, the incentive constraints ensuring that the borrower repays all the MFIs, rather than only  the $i$-th MFI and $(x-1)$ of the others, are given by:
 \begin{eqnarray}
 -mk_i(1+r_i)-m\frac{n-1}{n}  \widehat k(n)+ \frac{\delta}{1-\delta}Z(\widehat k(n),0)\ge    
-mk_i(1+r_i)-m \frac{(x-1)}{n}  \widehat k(n)+  \frac{\delta}{1-\delta}Z(\frac{x\widehat k(n)}{n},0), \; \forall  x.\label{IIC'''}\end{eqnarray}  Thus the $i$-th MFI's problem is to maximize (\ref{objno}), subject to (\ref{IIC}), (\ref{IIC'''}),   and $r_i\ge 0$. Using the fact that $f(.)$ is concave, one can show that whenever either (\ref{IIC}), or (\ref{IIC'''}), holds for $x=n-1$, it holds for all $x<n-1$.  The problem now simplifies to maximising  \eqref{tildeZ} subject to  \eqref{IIC}, and \eqref{IIC'''} both holding for   $x=n-1$. The rest of the analysis mimics that for the case where $n=2$, and the results are qualitatively similar.

 We next extend the asymmetric SSPE analysed in Section \ref{asym} to a general $n$.  We focus on equilibria where a single MFI, say M1,  provides productive loans, whereas all other MFIs provide bridge-loans.  Given that bridge loans cannot be used for productive purposes, if the borrower defaults on M1,   then  it is optimal  to immediately take  bridge-loans from  all other MFIs, and then default on all of them.  Consequently the optimisation problem in this case is identical to that in  Section \ref{asym}, so that the results remain unchanged. How does the analysis change if say $m$ of the MFIs, $1<m<n$, provide productive loans and $n-m$ provide bridge loans?   This would increase default incentives relative to the case where productive loans are provided by a single MFI since, as in Section \ref{sym}, the borrower can default on one of the $m$ MFIs and continue to get productive loans from the others. Consequently, the magnitude of MFI loans would be smaller relative to the case where productive loans are provided by a single MFI.
 
 \end{remark}

\begin{remark}[MFIs maximize a weighted average of borrower utility and profits]\label{profits}   Consider a scenario where the MFIs maximize a weighted average of borrower utility and MFI profits. Let the weight on MFI profit be $\mu$, $0\le \mu\le 1$,  for all MFIs.  From continuity it is clear that   there exists $\mu^*>0$ such that the solution coincides with that in Propositions \ref{M1}, \ref{M2:Coop} and \ref{Mbl}  whenever $\mu\le \mu^*$. 

Next consider a scenario where $\mu=1$, i.e. the MFIs are profit maximising.\footnote{Whether a greater emphasis on profitability leads MFIs to drift from their social objectives is debated. \citet{mersland2010microfinance} find no evidence of industry-wide mission drift, while \citet{armendariz2011mission} characterize conditions under which it may arise.} While a  complete analysis is beyond our scope,   we analyse a special   case where the government fixes the  MFI interest rate  at $\overline r$,  $0\le \overline r<m$.  Such an assumption is realistic in certain contexts.\footnote{
In April 2017, the National Bank of Cambodia (NBC)  introduced a 18\% per annum hard cap on interest rates for all MFIs (see IMF Working Paper WP/21/107, Impact of Interest Rate Cap on Financial Inclusion in Cambodia).  Similarly, in 2010, Bangladesh’s Microcredit Regulatory Authority (MRA) implemented a cap of 
 maximum of 27\% per annum on a ``declining balance" basis (see Microcredit Regulatory Authority (MRA) Circular No. 51 (2019) and MRA Rules (2010)). }  Further, this is a first step towards analysing the issue in its full generality.  First, suppose there is a single MFI. Given that the MFI is profit-maximizing,  the loan amount will be the maximal possible amount such that there is no default.  In that case  the optimal level of investment $k'$ satisfies 
  \begin{eqnarray}
\frac{\delta}{1-\delta}\frac{Z(k',\overline r)}{k'}=  m.\label{X1}
\end{eqnarray} 
Next suppose $n=2$. Under the symmetric equilibrium the optimal level of investment $k''$ solves 
  \begin{eqnarray}
\frac{\delta}{1-\delta}\frac{\widetilde Z(k'',\overline r)-\widetilde Z(\frac{k''}{2},\overline r)}{k''-\frac{k''}{2}}=  m.\label{X2}
\end{eqnarray} 
We can mimic the  argument in Proposition \ref{M2:Coop}  to show that both MFI, as well as moneylender lending is smaller  compared to the case where there is a single MFI.\footnote{The argument follows  since from \eqref{X1}, \eqref{X2} and, from the concavity of $\widetilde Z(k,r)$, that    
$\frac{Z(k, \overline r)}{k}>\frac{\widetilde Z(k,\overline r)-\widetilde Z(\frac{k}{2},\overline r)}{k-\frac{k}{2}}$.} Finally, consider the asymmetric SSPE.  In this case  
  the optimal level of investment $k'''$ solves
\begin{equation}\frac{\delta}{1-\delta} \frac{Z^*(k''',\overline r ,\overline r)}{k'''} = 2.\label{X3}\end{equation} 
Compared to the single MFI case, it is clear that aggregate MFI loans are larger whenever $m>2$.\footnote{From \eqref{X1} and \eqref{X3},  and the fact that $m>2$, we find that  $\frac{Z(k',\overline r)}{k'}>\frac{Z^*(k''',\overline r ,\overline r)}{k'''} $. The claim now follows from  the concavity of $f(.)$.} Thus, whenever the MFI rate of  interest is exogenously given,   the analysis goes through qualitatively even when the MFIs are profit maximisers.

\end{remark}

\FloatBarrier

\section{Empirical Analysis}\label{sec:empirics}

To test the model’s predictions, we analyze data from a cluster-randomized trial in which the MFI RDRS expanded into seven subdistricts in northern and northwestern Bangladesh. The expansion was randomized across 146 villages -- 97 assigned to treatment, 49 to control -- stratified by the number of MFIs present at baseline. A village census first recorded the number of MFIs operating in each village, followed by a baseline household survey in which respondents were asked whether they would be interested in borrowing if a new MFI began lending in their village. 

Six months later, RDRS began operations in treatment villages, offering credit exclusively to individuals who had expressed interest during the baseline survey. Three years post-baseline, an endline survey targeted 2,903 households—20 per village, comprising 10 randomly selected from those who had expressed interest at baseline and 10 randomly selected from those who had not. 

Figure \ref{fig:mfi_cdf} shows the distribution of MFI presence before and after the intervention. At baseline, villages had on average 5.2 MFIs, with only two villages having none. The analysis in this paper thus assesses the effects of MFI competition rather than entry. Further details on the experiment can be found in \citet{chowdhury2020contract}.


\FloatBarrier


\subsection{Econometric approach}\label{subsec:econometricapproach}

Our primary outcomes are the number of loans and amounts borrowed by households, across credit sources.\footnote{We did not collect data on default (rates) from MFIs, which would be necessary to analyze that outcome.} Given that these outcomes are nonnegative and feature a large proportion of zeros, an exponential mean model is more appropriate than a linear specification.\footnote{Transformations such as the inverse hyperbolic sine can compromise interpretability, as resulting estimates are not invariant to scale (e.g., whether amounts are measured in Taka or thousands of Taka) \citep{mullahy2024transform}.} Because MFI competition can affect both borrowing incidence and intensity, separating total effects into extensive and intensive margins is problematic without additional assumptions \citep{staub2014causal}. We therefore focus on total effects using the IV Poisson model and report extensive‐margin results for completeness.


The structural model linking MFI competition (measured by the number of MFIs in village $j$, $MFI_{j}$) to household-level borrowing outcomes $y_{ij}$ for household $i$ in village $j$ is:

\begin{align}\label{eq:specification1}
\begin{split}
E(y_{ij}|MFI_{j},X_{ij},\Omega_{ij}) &= exp(\alpha + \beta MFI_{j} + X_{ij} \gamma + \Omega_{ij}) \\
         &= exp(\alpha + \beta MFI\_competition_{j} + X_{ij} \gamma)\eta_{ij}
\end{split}
\end{align}

where $X_{ij}$ is a vector of baseline covariates (including baseline realizations of the outcome to improve precision \citep{mckenzie2012beyond}), $\Omega_{ij}$ captures unobserved village-level heterogeneity, with $\eta_{ij} = exp(\Omega_{ij}) > 0$. This leads to the exponential regression model by \citet{mullahy1997instrumental}\footnote{A conceptual advantage of the multiplicative model is that it treats observables and unobservables (omitted variables and measurement error) symmetrically: all are inside the exponential mean function \eqref{eq:specification1}. In contrast, the additive IV Poisson model is hard to motivate directly, with all observed variables having multiplicative effects but the unobservables having an additive effect.}:

\begin{equation}\label{eq:specification1}
E(y_{ij}|X_{ij},\eta_{ij}) = exp(\alpha + \beta MFI\_competition_{j} + X_{ij} \gamma)\eta_{ij} + \varepsilon_{ij}
\end{equation}

where $E(\varepsilon_{ij} | X_{ij}, \eta_{ij}) = 0$ by construction.

Since $MFI_{j}$ is potentially endogenous (e.g., MFIs may prioritize poorer villages \citep{salim2013revealed}), we instrument $MFI_{j}$ with randomized treatment assignment ($assignment_{j}$). We estimate equation \eqref{eq:specification1} via GMM following \citet{mullahy1997instrumental}, corresponding to quasi–maximum likelihood (QML) estimation of the conditional mean. Random assignment ensures $E(\varepsilon_{ij} | X_{ij}, assignment_{j}) = 0$ and $E(\eta_{ij} | assignment_{j})$ is constant, normalized to equal 1. The coefficient of interest $\beta$ captures the marginal effect of one additional MFI having entered a village 2-3 years before the endline survey. We report transformed coefficients $\beta^{*} = exp(\beta) - 1$, interpreted as semi-elasticities (i.e., 100*$\beta^{*}$ equals the percentage change in expected outcomes per additional MFI).

We estimate specification \ref{eq:specification1} on both the full sample and the subsample of households who expressed interest in borrowing at baseline. While the subsample might be more reactive to changes in the credit market \citep{mckenzie2025designing}, the sample size is reduced. For completeness, we also report village-level estimates. For moneylender interest rates, which are selectively observed among a changing pool of moneylender borrowers, we favor these village-average-based estimates over household-level estimates.

To explore heterogeneity in treatment effects across the distribution of $MFI_j$, we extend the multiplicative exponential mean model \eqref{eq:specification1} by replacing the linear term with a smooth nonparametric function $h(\cdot)$:

\begin{align}\label{eq:specification2}
\begin{split}
E(y_{ij}|MFI_{j},X_{ij},\Omega_{ij}) &= exp(\alpha + h(MFI_{j}) + X_{ij} \gamma)\eta_{ij}
\end{split}
\end{align}

Identification in this setting\textemdash with a binary instrument, a count-valued endogenous regressor, and discrete or semi-continuous outcomes\textemdash has recently been established by \citet{chernozhukov2024estimating} under a Copula Invariance condition: the joint dependence between unobserved determinants of treatment intensity and outcomes remains invariant across instrument values. Given that MFI entry was randomly assigned across villages, this condition is particularly plausible. The instrument also shifts the full support of $MFI_j$ (Figure \ref{fig:mfi_cdf}), satisfying their support condition. We estimate $h(\cdot)$ via two-step IV using penalized splines \citep{wood2017generalized}.

For both the full and ‘interested in borrowing’ samples, baseline characteristics and attrition are balanced across treatment and control villages (Table \ref{tab:summarystats1}-\ref{tab:summarystats2})).

\FloatBarrier

\subsection{Results}

Table \ref{tab:nrloans} reports the estimates of MFI competition on household borrowing as measured by the number of loans taken from a range of sources. An additional MFI operating in a village significantly reduces the number of loans from moneylenders by 30–40\%, relative to a baseline of 0.24–0.27 loans per household (column (3)), consistent with the model’s prediction.

We cannot reject the null that the number of loans from MFIs, the number of loans from friends and family, or the total number of loans are unaffected by MFI competition (columns (1), (4) and (5)).\footnote{Loans are classified into MFI and non-MFI loans. The non-MFI category includes moneylender loans ( \% of all loans), loans from acquintances (family and friends; \% of all loans), commercial banks ( \% of all loans), and other, less frequent sources.} Power simulations described in Appendix~\ref{app:power} indicate that the design has 80\% power to detect a reduction in the number of MFI loans of around 11.6\%, or 0.08 loans.

The estimated negative effect for moneylender loans is robust to accounting for attrition (Table \ref{tab:attrition}), and to linear specifications (Table \ref{tab:nrloanslinear}). In these linear IV specifications, the Kleibergen–Paap rk Wald F-statistics are around 20, comfortably above the conventional threshold of 10, confirming the strength of the randomized MFI-entry instrument. For completeness, we also estimate the extensive margin (whether a household has at least one loan from a given source) using an IV-LPM and an IV-Probit model; the estimates point in the same direction as those for the number of loans (Table \ref{tab:extensive}). Among households expressing an `interested in borrowing' at baseline, the effect on having at least one loan from any source is negative, consistent with our theoretical predictions. The semiparametric estimates shown in Figure \ref{fig:semip} show that the decrease in borrowing from moneylenders occurs and is discernible when the number of MFIs in the village increases beyond five.

\begin{table}[!htbp]
\small
\caption{\small{Effect of MFI competition on the number of loans by credit source
\label{tab:nrloans}}}
\centering
  \begin{threeparttable}
\begin{tabular}{l*{5}{>{\centering\arraybackslash}p{2.2cm}}}
\toprule
\multicolumn{1}{c}{} & \multicolumn{1}{c}{(1)}                                          & \multicolumn{1}{c}{(2)}                                             & \multicolumn{1}{c}{(3)}                                                 & \multicolumn{1}{c}{(4)}       & \multicolumn{1}{c}{(5)}                               \\
 & \begin{tabular}[c]{@{}c@{}}\# of loans \\ from MFIs \textcolor{white}{iii}\end{tabular} & \begin{tabular}[c]{@{}c@{}}\# of loans\\ from non-MFIs\end{tabular} & \begin{tabular}[c]{@{}c@{}}\# of loans from\\ moneylenders\end{tabular} & \begin{tabular}[c]{@{}c@{}}\# of loans from\\ acquintances\end{tabular} & \begin{tabular}[c]{@{}c@{}}\# of loans\\ (any)\end{tabular} \\
 \midrule
\multicolumn{6}{@{}l}{\textit{Panel A: Full sample}}\\
\vspace{-.7em}\\
\# of MFIs  &        0.01   &       -0.08*  &       -0.31***&       -0.02   &       -0.06   \\
            &      (0.06)   &      (0.05)   &      (0.13)   &      (0.07)   &      (0.04)   \\
\vspace{-.7em}\\
E(dep. var.)&        0.72   &        1.20   &        0.24   &        0.76   &        1.92   \\
(st. dev.)  &        0.80   &        1.16   &        0.57   &        0.99   &        1.53   \\
Zeros (\%)  &        46.1   &        32.0   &        81.4   &        51.4   &        18.3   \\
Observations&        2735   &        2735   &        2735   &        2735   &        2735   \\
\\
\multicolumn{6}{@{}l}{\textit{Panel B: `Willing to borrow' sample}}\\
\vspace{-.7em}\\
\# of MFIs  &        0.06   &       -0.10*  &       -0.44***&       -0.06   &       -0.06   \\
            &      (0.06)   &      (0.06)   &      (0.22)   &      (0.10)   &      (0.04)   \\
\vspace{-.7em}\\
E(dep. var.)&        0.76   &        1.23   &        0.27   &        0.76   &        1.99   \\
(st. dev.)  &        0.80   &        1.14   &        0.59   &        0.99   &        1.51   \\
Zeros (\%)  &        42.6   &        30.7   &        79.0   &        51.6   &        16.0   \\
Observations&        1374   &        1374   &        1374   &        1374   &        1374   \\
\\
\multicolumn{6}{@{}l}{\textit{Panel C: Village-level averages}}\\
\vspace{-.6em}\\
\# of MFIs  &       -0.00   &       -0.08   &       -0.36** &       -0.01   &       -0.06   \\
            &      (0.07)   &      (0.05)   &      (0.21)   &      (0.08)   &      (0.05)   \\
\\
E(dep. var.)&        0.71   &        1.20   &        0.24   &        0.76   &        1.91   \\
(st. dev.)  &        0.37   &        0.45   &        0.20   &        0.37   &        0.64   \\
Zeros (\%)  &         1.4   &         0.0   &        14.4   &         0.0   &         0.0   \\
Observations&         146   &         146   &         146   &         146   &         146   \\
\bottomrule
\end{tabular}
    \begin{tablenotes}
\item \emph{Notes:} Estimates are based on GMM estimation of Mullahy’s (1997) IV Poisson model of the marginal effect of MFI competition—proxied by the number of MFIs in a village—on household loan counts in the year preceding the endline survey. The estimation is Poisson QML–based, requiring only correct specification of the conditional mean and being robust to excess zeros \citep{staub2013consistent}. The endogenous regressor is instrumented with randomized MFI entry (RDRS). Reported coefficients are $exp(\beta) - 1$ and can be interpreted as semi-elasticities. Columns show effects on loans from: MFI credit providers (1); any non-MFI lenders including moneylenders and acquaintances (2); moneylenders (3); friends, family, and acquaintances (4); and all sources combined (5). E(dep. var.) refers to the baseline mean of the dependent variable in the estimation sample. Baseline controls include the outcome’s baseline value; gender, age, literacy, and education of the household head; household size and composition, number of income sources, landholdings, wealth index, saving behavior, household- and village-level shock indices, and district fixed effects (see \ref{app:construction}). Panel C uses baseline village averages of these variables. Robust standard errors clustered at the village level are in parentheses. $^{*}$, $^{**}$, $^{***}$ denote statistical significance at the 10, 5, and 1 percent levels, respectively.
    \end{tablenotes}
\end{threeparttable}
\end{table}

Regarding amounts borrowed, MFI competition does not have a statistically significant effect on total borrowing (column (5)), borrowing from MFIs (column (1)), non-MFIs (column (2)), or friends and acquaintances (column (4)) (Table \ref{tab:amts}). However, the amount borrowed from moneylenders declines by roughly 20–35 percent, with statistical significance concentrated among households who expressed interest in borrowing at baseline (column 3). The result largely holds when accounting for attrition (Table \ref{tab:attrition}) or when applying a different or no winsorization to the outcome \ref{tab:winsor1}. Linear IV estimates - not our preferred specification, carry the same sign but lack statistical significance (Table \ref{tab:amtslinear}). The semiparametric estimates are shown in Figure \ref{fig:semip_amts}.

Romano-Wolf p-values reported in Table \ref{tab:tableMHT} indicate that the results for the number of moneylender loans retain statistical significance (in the full sample) when accounting for multiple hypothesis testing, whereas those on amount borrowed from moneylenders do not.

Only 18.6\% of households had a moneylender loan at endline; because of the resulting small and selective sample, we de-emphasize results on moneylender interest rates, which are considered suggestive only. Some of the estimates are negative and marginally statistically significant, while lacking robustness (Table \ref{tab:interest}).

\vspace{1cm}

\begin{figure}[!htbp]
    \centering
    \subfloat{\includegraphics[width=.495\textwidth]{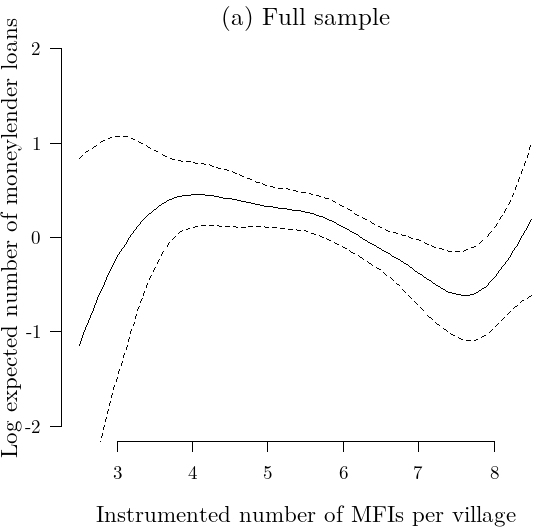}}
    \hspace{-.2cm}
    \subfloat{\includegraphics[width=.495\textwidth]{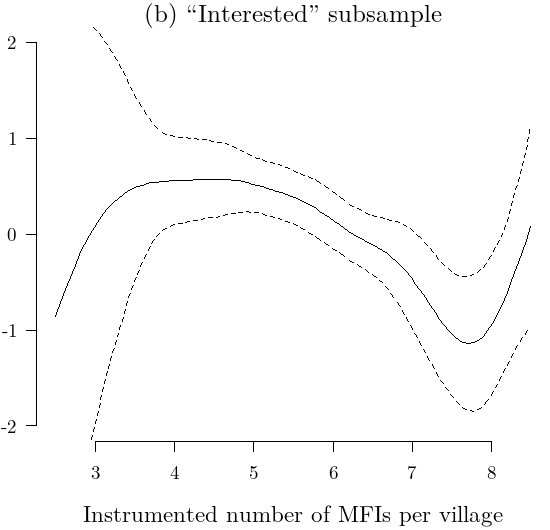}}
    \caption{Estimated smooth effect of MFI competition in (a) the full sample, and (b) the "interested" in MFI borrowing subsample. The plotted range excludes sparsely populated tails where the semiparametric estimates become imprecise.}
    \label{fig:semip}
\end{figure}

\FloatBarrier

\begin{table}[!htbp]
\small
\caption{\small{Effect of MFI competition on total amounts borrowed by credit source \label{tab:amts}}}
\centering
  \begin{threeparttable}
\begin{tabular}{l*{5}{>{\centering\arraybackslash}p{2.2cm}}}
\toprule
\multicolumn{1}{l}{} & \multicolumn{1}{c}{(1)}                                          & \multicolumn{1}{c}{(2)}                                             & \multicolumn{1}{c}{(3)}                                                 & \multicolumn{1}{c}{(4)}      & \multicolumn{1}{c}{(5)}                    \\
 & \begin{tabular}[c]{@{}c@{}}Amount \\(1000 Tk) \\borrowed \\ from MFIs \textcolor{white}{iii}\end{tabular} & \begin{tabular}[c]{@{}c@{}}Amount \\(1000 Tk) \\borrowed \\ from non-MFIs\end{tabular} & \begin{tabular}[c]{@{}c@{}}Amount \\(1000 Tk) \\borrowed from\\ moneylenders\end{tabular} & \begin{tabular}[c]{@{}c@{}}Amount \\(1000 Tk) \\borrowed from\\ acquintances\end{tabular} & \begin{tabular}[c]{@{}c@{}}Amount \\(1000 Tk) \\borrowed\\(any source) \end{tabular} \\
 \midrule
\multicolumn{6}{@{}l}{\textit{Panel A: Full sample}}\\
\vspace{-.7em}\\
\# of MFIs         &        0.02   &        0.03   &       -0.22*  &        0.02   &        0.01   \\
                    &      (0.06)   &      (0.08)   &      (0.15)   &      (0.10)   &      (0.05)   \\
\vspace{-.6em}\\
E(dep. var.)  & 12.36 & 7.34 & 1.54 & 3.69 & 20.07 \\
(st. dev.)    & (16.81) & (14.59) & (4.81) & (8.47) & (26.00) \\
\vspace{-.6em}\\
Zeros (\%)  & 46.1\% & 32.0\% & 81.4\% & 51.4\% & 18.3\% \\
\vspace{-.6em}\\
Observations                   &        2735   &        2735   &        2735   &        2735  &        2735 \\
\\
\multicolumn{6}{@{}l}{\textit{Panel B: `Willing to borrow' sample}}\\
\vspace{-.7em}\\
\# of MFIs         &        0.07   &       -0.06   &       -0.35** &       -0.01   &        0.01   \\
                    &      (0.07)   &      (0.08)   &      (0.20)   &      (0.13)   &      (0.05)   \\
\vspace{-.6em}\\
E(dep. var.)  &   12.90 &   7.18   &   1.73    &   3.51   &  20.42   \\
(st. dev.)    & (16.70) &   (13.65) &   (5.04) &   (7.98) &   (25.09) \\
\vspace{-.6em}\\
Zeros (\%)  &  42.6\% &   30.7\%  &   79.0\%  &   51.6\%  &   0.00\%  \\
\vspace{-.6em}\\
Observations                  &        1374   &        1374   &        1374   &        1374 &        1374   \\
\\
\multicolumn{6}{@{}l}{\textit{Panel C: Village-level averages}}\\
\vspace{-.7em}\\
\# of MFIs          &        0.01   &        0.03   &       -0.14   &        0.01   &       -0.00   \\
                    &      (0.07)   &      (0.07)   &      (0.17)   &      (0.09)   &      (0.07)   \\
\vspace{-.6em}\\
E(dep. var.)  & 12.20 & 7.32 & 1.55 & 3.68 & 19.88 \\
(st. dev.)    & (7.11)  & (4.45) & (1.53) & (2.50) & (9.39) \\
\vspace{-.6em}\\
Zeros (\%)  & 1.4\% & 0.0\% & 14.4\% & 0.0\% & 0.0\% \\
\vspace{-.6em}\\
Observations                 &         146   &         146   &         146   &         146   &         146   \\    
\bottomrule
\end{tabular}
    \begin{tablenotes}
\item \emph{Notes:} Estimates are based on GMM estimation of Mullahy’s (1997) IV Poisson model of the marginal effect of MFI competition—proxied by the number of MFIs in a village—on the total amounts borrowed by households in the year preceding the endline survey. The estimation is Poisson QML–based, requiring only correct specification of the conditional mean and being robust to excess zeros \citep{staub2013consistent}. The endogenous regressor is instrumented with randomized MFI entry (RDRS). Reported coefficients are $exp(\beta) - 1$ and can be interpreted as semi-elasticities. Columns show effects on amounts borrowed from: MFI credit providers (1); any non-MFI lenders including moneylenders and acquaintances (2); moneylenders (3); friends, family, and acquaintances (4); and all sources combined (5). E(dep. var.) refers to the baseline mean of the dependent variable in the estimation sample. Baseline controls include the outcome’s baseline value; gender, age, literacy, and education of the household head; household size and composition, number of income sources, landholdings, wealth index, saving behavior, household- and village-level shock indices, and district fixed effects (see \ref{app:construction}). Panel C uses baseline village averages of these variables. The dependent variable is winsorized at the 99$^{\text{th}}$ percentile to limit the influence of outliers. Robust standard errors clustered at the village level are in parentheses. $^{*}$, $^{**}$, $^{***}$ denote statistical significance at the 10, 5, and 1 percent levels, respectively.
    \end{tablenotes}
\end{threeparttable}
\end{table}

\clearpage

\normalsize

\FloatBarrier

\section{Conclusion}\label{sec:conclusion}

This paper examines the evolution of moneylending in the face of intensifying competition between socially motivated MFIs. We develop a dynamic model of ex-post moral hazard where repayment is enforced through the threat of future credit exclusion. Our analysis identifies a critical functional complementarity: moneylenders leverage their informational and other advantages to provide ``bridge loans" when MFI-funded projects fail, allowing borrowers to maintain MFI access. This symbiosis helps explain the persistence of high-interest moneylending despite global MFI expansion.

We demonstrate that MFI competition affects credit markets through two diverging channels. First, competition can weaken dynamic incentives by making exclusion less consequential, thereby tightening MFI credit constraints and reducing the amount lent by both MFIs and moneylenders. Second, borrowers may shift to MFIs for bridge loans.  Despite such bridge-loans also being susceptible to moral hazard, we show that aggregate MFI loans increase, while moneylender loans go to zero.  Consequently, while MFI competition consistently erodes the moneylender’s market share, its impact on total MFI lending remains theoretically ambiguous. 

This framework    also contributes to the longstanding debate on whether MFI and moneylender loans are substitutes or complements (e.g., \citet{banerjee2015six,banerjee2025can}). We show that the relationship depends both on whether the market is underserved or already served by MFIs and on whether MFIs themselves provide bridge finance. MFI entry into underserved markets substitutes productive lending away from moneylenders, whereas in already-served markets moneylender bridge loans can be complementary to productive MFI lending. When MFIs themselves provide bridge finance, however, they directly substitute for moneylenders.

We test this prediction using a cluster-randomized trial in Bangladesh, where the randomized entry of an additional MFI into already-served villages increased local competition. Moneylender borrowing—consistent with the model—declines, reflected in fewer loans and lower total amounts, while MFI borrowing is unaffected. Effects on moneylender interest rates are inconclusive.

 
These insights carry two main policy implications. First, socially motivated MFIs may achieve greater impact by expanding into underserved areas rather than competing in saturated markets. Second, credit bureaus can help sustain both formal and informal lending in competitive markets by mitigating ex-post moral hazard \citep{luoto2007credit,de2010supply}. More broadly, our findings demonstrate how competition can reshape imperfect credit markets by weakening repayment incentives.

\newpage

\bibliographystyle{plainnat}
\bibliography{bibliography}

\newpage


\clearpage






\clearpage




\begin{titlepage}
  \centering
  \vskip 60pt
  \LARGE Online Appendix \\ Microfinance Competition in the Presence of Moneylenders: Theory and Evidence

  \vskip 3em
  \centering
\large Shyamal Chowdhury \hspace{.9cm} Prabal Roy Chowdhury \hspace{.9cm} Joeri Smits
\end{titlepage}

\appendix
\setcounter{page}{1}


\clearpage

\appendix

\counterwithin{figure}{section}
\counterwithin{table}{section}
\renewcommand\thefigure{\thesection\arabic{figure}}
\renewcommand\thetable{\thesection\arabic{table}}

\renewcommand{\thesection}{Appendix \Alph{section}}
\renewcommand{\thesubsection}{\Alph{section}.\arabic{subsection}}

\FloatBarrier

\section{Theory}\label{appsec:theoryappendix}

\normalsize

\emph{Proof of Proposition \ref{M1}.}   1. and 2. We prove the first two parts simultaneously. 

\emph{ Existence}:  We first show that the strategies described in Proposition \ref{M1} constitute an SSPE.  The proof proceeds in several steps.\smallskip

\emph{Step 1}. If \eqref{D2} holds, then so does \eqref{D1} (since $m>1$). Thus it is sufficient to maximize $Z(k,r)$ subject to  \eqref{D2} and \eqref{BE}. \smallskip

\emph{Step 2}. Next note that a decrease in $r$ relaxes   \eqref{D2}, while it  increases $Z(k,r)$. Thus, for any $k$, the optimum involves selecting the minimum $r$ such that $r\ge 0$, i.e. $r=0$.  \smallskip

\emph{Step 3}. The optimisation problem now simplifies to maximising $Z(k,0)$ subject to  \begin{eqnarray}
\frac{\delta}{1-\delta}\frac{Z(k,0)}{k}\ge  m.\label{D2''''}
\end{eqnarray} 

\emph{Step 4}. Clearly $\widehat k(1)=\overline k$ whenever  $\overline k$ satisfies the no default condition \eqref{D2''''}.  \smallskip

\emph{Step 5}. Next suppose that $\overline k$ violates \eqref{D2''''}.  There are two cases. 
 
 Case (i) $\frac{\delta}{1-\delta}Z'(0,0)\le  m$: Given that $Z(k,0)$  is concave in $k$,  $\frac{Z(k,0)}{k}$ is decreasing in $k$, and thus $\frac{Z(k,0)}{k}<\lim_{k\rightarrow 0} \frac{Z(k,0)}{k}=Z'(0,0)$. Hence  $\frac{\delta}{1-\delta}Z'(0,0)<  m$ for all $k$ positive. Thus 
  \eqref{D2''''}  fails for all $k$ positive. Thus $\widehat k(1)=0$.
 
 Case (ii) $\frac{\delta}{1-\delta} Z'(0,0)> m$:  Given that $Z(k,0)$ is concave, and maximised at $\overline k$, $Z(k,0)$ is increasing in $k$ for all $k<\overline k$. Thus the $\widehat k(1)$ is the largest $k$ that satisfies \eqref{D2''''}. Finally, given that  $\frac{Z(k,0)}{k}$ is decreasing in $k$, the  optimal investment level  solves \eqref{D2''''}  with an equality, i.e.  $\underline k(1,\delta)$. \smallskip

\emph{Step 6}. Finally we show that there exists $0<\widetilde\delta(1) <1$ such that, for $\delta=\widetilde \delta(1)$, 
$$\frac{\delta}{1-\delta}\frac{Z(\overline k, 0)}{\overline k}=m.$$
This follows since $\lim_{\delta\rightarrow 0}\frac{\delta}{1-\delta}\frac{Z(\overline k, 0)}{\overline k}=0<m$,  $\lim_{\delta\rightarrow 1}\frac{\delta}{1-\delta}\frac{Z(\overline k, 0)}{\overline k}=\infty, $ and $\frac{\delta}{1-\delta}\frac{Z(\overline k, 0)}{\overline k}$ is increasing in $\delta$. 
Thus  the no default condition \eqref{D2''''}  binds iff $\delta<\widetilde \delta(1)$.

 \emph{Uniqueness}:  Uniqueness follows since the preceding argument shows that $(\widehat k(1),\widehat r(1))$ is the unique solution to  M1's optimisation problem. 

3. The result follows since, by definition, $\widehat k(1) \le\min\{\overline k,\underline k(1,\delta)\} \le  \overline k<k^*$.
\qed\vskip 1em

\emph{Proof of Proposition \ref{M2:Coop}}.  \emph{  Existence}: We first show that the strategies described in  Proposition  \ref{M2:Coop} constitutes an SSPE.

The argument  proceeds in various steps.\smallskip 

\emph{Step 1}.    $\frac{\widetilde Z(k,r) -\widetilde Z(\frac{xk}{2},r)}{k-\frac{xk}{2}}$ is decreasing in $x$. \smallskip 

Note that for any $K',K$,  where $K'>K$,
{\small\begin{eqnarray}
\frac{d}{dK}\left( \frac{\widetilde Z(K',r) -\widetilde Z(K,r)}{K'-K}\right)=\frac {\frac{\widetilde Z(K',r) -\widetilde Z(K,r)}{K'-K} -\widetilde Z_1(K,r)}{K'-K}<0,\label{K}
\end{eqnarray}}since $\widetilde Z(K,r)$ is concave in $K$, and $K'>K$.  Hence    the claim follows. \smallskip 

\emph{Step 2}. We next argue that one can ignore the incentive constraints when the project succeeds.\smallskip 

Suppose there is project success. The incentive constraints such that the borrower wants to repay all the MFIs, rather than only  $x$  of them,  $0\le x\le 1$, defaulting on the others, \emph{including}   M2 are:
{\small \begin{eqnarray}
\hskip -2em  f(k_i+\frac{1}{2}  \widehat k(2)) -k_i(1+r_i)
-\frac{1}{2}  \widehat k(2)+\frac{\delta}{1-\delta}\widetilde Z(\widehat k(2),0)\ge    f(k_i+\frac{1}{2}  \widehat k(2))
 - \frac{x}{2}  \widehat k(2)+\frac{\delta}{1-\delta}\widetilde Z(\frac{x\widehat k(2)}{2},0),\;\forall x.\label{IC1}\end{eqnarray}    }
Similarly, the incentive constraint ensuring that the borrower repays all the MFIs, rather than only  M2  is  given by:
{\small  \begin{eqnarray}
\hskip -2em   f(k_i+\frac{1}{2}  \widehat k(2))
-k_i(1+r_i)-\frac{1}{2}  \widehat k(2)+ \frac{\delta}{1-\delta}\widetilde Z(\widehat k(2),0)
\ge     f(k_i+\frac{1}{2}  \widehat k(2))
-k_i(1+r_i)+  \frac{\delta}{1-\delta}\widetilde Z(\frac{\widehat k(2)}{2},0).\label{IC2}\end{eqnarray}}

Next suppose there is project failure.  The incentive constraints such that the borrower wants to repay all the MFIs, rather than only  $x$  of them,  $0\le x\le 1$, defaulting on the others, \emph{including}   M2 are:
\begin{eqnarray}  
\hskip 0em -mk_2(1+r_2)-m\frac{1}{2}  \widehat k(2)+\frac{\delta}{1-\delta}\widetilde Z(\widehat k(2),0)\ge    
-m \frac{x}{2}  \widehat k(2)+\frac{\delta}{1-\delta}\widetilde Z(\frac{x\widehat k(2)}{2},0),\forall x\le 1.\label{IC}\end{eqnarray}    
Next  consider  the incentive constraint ensuring that the borrower repays all the MFIs, rather than only  M2 in case of project failure:
 \begin{eqnarray}
 -mk_2(1+r_2)-m\frac{1}{2}  \widehat k(2)+ \frac{\delta}{1-\delta}\widetilde  Z(\widehat k(2),0)\ge    
-mk_2(1+r_2)+  \frac{\delta}{1-\delta}\widetilde Z(\frac{\widehat k(2)}{2},0).\label{IC'''}\end{eqnarray}

Given that $m>1$, \eqref{IC1} holds whenever \eqref{IC} holds, and  \eqref{IC2} holds whenever \eqref{IC'''} holds.


\smallskip

 Thus M2's  problem is to maximize (\ref{objno}), subject to (\ref{IC}), (\ref{IC'''}),   and $r_2\ge 0$.\smallskip

\emph{Step 3}. $r_2=0$. \smallskip

 Note that a reduction in $r_2$     increases the maximand  \eqref{objno}, relaxes   (\ref{IC}), while it has no effect on (\ref{IC'''}). Thus in equilibrium  $r_2=0$. \smallskip
 
\emph{ Step 4}. Given that $r_2=0$, (\ref{IC})  can be re-arranged as follows,
{\small \begin{equation} \frac{\delta}{1-\delta}  \frac{\left[\widetilde Z(\widehat k(2), 0)-\widetilde Z(\frac{x\widehat k(2)}{2},0)\right]}    {[k_2+\frac{1}{2}  \widehat k(2)]-\frac{x  \widehat k(2)}{2}} \ge  m,\;\;\forall 0\le x\le 1,\label{DICshock'}
 \end{equation}}
 while (\ref{IC'''})  can be re-arranged as follows:
{\small \begin{equation}  \frac{\delta}{1-\delta}  \frac{\left[\widetilde Z(\widehat k(2), 0)-\widetilde Z(\frac{\widehat k(2)}{2},0)\right]}     {  \widehat k(2)-\frac{ \widehat k(2)}{2}} \ge  m.\label{DICshocki'}
 \end{equation}}

\emph{Step 5}. We proceed by ignoring   (\ref{DICshocki'}), and solving the resultant relaxed problem, i.e. maximize (\ref{objno}), subject to \eqref{DICshock'}.  We shall later show that the solution to the relaxed problem solves (\ref{DICshocki'}).

 First suppose  that   $\overline k\le \underline k(2,\delta).$  From the concavity of $\widetilde Z(k,0)$, the LHS of \eqref{DICshock'} is decreasing in $k$, and hence   $\overline k$ satisfies  \eqref{DICshock'}.  Further, given that     $\overline k\le \underline k(2,\delta)$, by definition $\widehat k(2)=\overline k$, so that M1 offers $(\frac{\overline k}{2},0)$. Thus   $k_2=\frac{\overline k}{2}$ maximizes \eqref{objno}, and    it satisfies  \eqref{DICshock'}.  Further, given that $k_2=\frac{\overline k}{2}$,   for $x=1$   \eqref{DICshock'}  coincides with \eqref{DICshocki'}, so that \eqref{DICshocki'} is also satisfied.
 
 Next suppose  that $\overline k>\underline k(2,\delta).$  Then, by definition,  $\widehat k(2)=\underline k(2,\delta)$, so that M1 offers $(\frac{\underline k(2,\delta)}{2},0)$.   Thus, any $k_2>\frac{\underline k(2,\delta)}{2}$ violates    \eqref{DICshock'} for $x=1$.   Thus given that the maximand is increasing in $k$ for $k\le \overline k$, optimally  $k_2=\frac{\underline k(2,\delta)}{2}$.  This satisfies \eqref{DICshock'} for $x=1$, and thus \eqref{DICshocki'}.  Finally, given that  \eqref{DICshock'} is satisfied  for $x=1$, from Step 1 earlier,  \eqref{DICshock'} is satisfied  for $x=0.$  \smallskip
  
\emph{  Step 6}. We next  show that there exists $0<\widetilde\delta(2) <1$ such that, for $\delta=\widetilde \delta(2)$, 
$$\frac{\delta}{1-\delta}\frac{\widetilde Z(\overline k, 0)-\widetilde Z(\frac{\overline k}{2})}{\overline k-\frac{\overline k}{2}}=m.$$

This follows since $\lim_{\delta\rightarrow 0}\frac{\delta}{1-\delta}\frac{\widetilde Z(\overline k, 0)-\widetilde Z(\frac{\overline k}{2})}{\overline k-\frac{\overline k}{2}}=0<m$,  $\lim_{\delta\rightarrow 1}\frac{\delta}{1-\delta}\frac{\widetilde Z(\overline k, 0)-\widetilde Z(\frac{\overline k}{2})}{\overline k-\frac{\overline k}{2}}=\infty, $ and $\frac{\delta}{1-\delta}\frac{\widetilde Z(\overline k, 0)-\widetilde Z(\frac{\overline k}{2})}{\overline k-\frac{\overline k}{2}}$ is increasing in $\delta$. \medskip
 
\emph{Uniqueness}: We next argue that no other symmetric SSPE exists.  In any symmetric SSPE, it must be that $\hat r(2)=0$. Otherwise, either MFI can lower the rate of interest. This would increase borrower utility, and relax the no default conditions.  Next, suppose   there exists a symmetric SSPE where the aggregate loan amount is $k'$, where $k'<\widehat k(2)$, and both MFIs offer $k'/2$. Clearly, the no default condition is slack at $k'$. But then M1, say, can offer a loan package $(k'/2+\epsilon,0)$. For $\epsilon$ small, the no default conditions are satisfied. Further, borrower utility is higher as $k'<\overline k$.

3. The result follows since, by definition, $\widehat k(2)\le\min\{\overline k,\underline k(2,\delta)\}\le  \overline k<k^*$.\qed\vskip 1em

   \emph{Proof of Proposition \ref{Compp}}. If $\delta\le \widehat \delta$, then aggregate MFI loan is zero in both cases. Next suppose that $\delta> \widehat \delta$. Given that $Z(k,0)$ is concave, we have that $\frac{Z(k,0)}{k}>\frac{Z(k,0)-Z(\frac{k}{2},0)}{k-\frac{k}{2}}$.  This ensures that (a) $\widetilde \delta(1)<\widetilde\delta (2)$, and   (b)  $\underline k(1,\delta)>\underline k(2,\delta)$. Hence, for $\delta\ge \widetilde\delta(2)$, $\widehat k(1)=\widehat k(2)=\overline k$. For $\widetilde\delta (1)\le \delta< \widetilde\delta(2)$,  $\widehat k(1)=\overline k>\underline k(2,\delta)=\widehat k(1)$. Finally, for  $ \delta< \widetilde\delta(1)$,  $\widehat k(1)=\underline k(1,\delta)>\underline k(2,\delta)=\widehat k(1)$, whenever $\underline k(1,\delta)>0$.  
   \qed\vskip 1em

  \emph{Proof of Proposition \ref{Mbl}}. 1. and 2.  \emph{Existence}: The proof proceeds via several steps. \smallskip

\emph{ Step 1}. Consider the no default condition in case of project success:
\begin{equation}
ICSuccess: f(k)-(1+r_1)k+\frac{\delta}{1-\delta} Z^*(k,r_1,r_2) \ge f(k) +(1+r_1)k \label{F1}.\end{equation} If  \eqref{F2} holds, then so does \eqref{F1}. So, while considering the optimisation problem facing M1 (resp. M2)  it is sufficient to consider constraints \eqref{F2}, \eqref{F3} and that $r_1\ge 0$ (resp. $r_2\ge 0$).\smallskip

\emph{Step 2}. We next argue that optimally $r_1$ and $r_2$ both equal zero. This follows since  a reduction in both $r_1$ and $r_2$ increases the maximand,  and relaxes \eqref{F2}. 
How does a reduction in $r_1$ or $r_2$ affect \eqref{F3}? Note that  \eqref{F3} can be re-written as 
\begin{eqnarray}
\frac{\delta}{1-\delta} \frac{(1-p)[m-(1+r_1)]}{1+r_1}\ge 1.\end{eqnarray}
Thus a decrease in $r_1$ relaxes \eqref{F3}, whereas a decrease in $r_2$ does not affect \eqref{F3}.   Thus optimally $r_1=r_2=0$. \smallskip
 
\emph{ Step 3}.  Given \emph{Step 2},  the optimisation problem simplifies to the following:
\begin{subequations}\label{Coop}\begin{align}
&\max_{k}\;  Z^*(k,0,0),\\
&s.t.\nonumber\\
&\frac{\delta}{1-\delta} \frac{Z^*(k,0,0)}{k} \ge 2,\label{F2'''}\\
 & \frac{\delta}{1-\delta} (1-p)(m-1)\ge 1.\label{F3''}
\end{align}
\end{subequations}

\emph{Step 4}. Note that if \eqref{F3''} fails, i.e. 
$$\delta <\delta^0\equiv  \frac{1}{1+(1-p)(m-1)},$$   then the optimisation problem does not have a solution.\smallskip

 \emph{ Step 5}.   Let $\widetilde \delta^*(2)$ solve
\begin{equation}\frac{\delta}{1-\delta} \frac{Z^*(k^*,0,0)}{k^*} = 2.\end{equation} 
$\widetilde \delta^*(2)$ is well defined  since $\lim_{\delta\rightarrow 0}\frac{\delta}{1-\delta} =0<$,  $\lim_{\delta\rightarrow 1}\frac{\delta}{1-\delta}=\infty, $ and $\frac{\delta}{1-\delta}$ is increasing in $\delta$. 
  \smallskip

\emph{Step 6}.  Given \emph{Step 4,}  assume that \eqref{F3''} holds, i.e. $\delta \ge\delta^0$. 
 
(i)  In case $\delta\ge \widetilde \delta^*(2)$, the first best level of investment $k^*$ satisfies  the no default condition  \eqref{F2'''}. This follows since $ \widetilde \delta(2)^*$ satisfies   $\frac{\delta}{1-\delta} \frac{Z^*(k,0,0)}{k} = 2$, and $\frac{\delta}{1-\delta}$ is increasing in $\delta$.      Hence    the optimal level of investment  equals $k^*$.\smallskip

(ii)  Next suppose that $k^*$ violates  \eqref{F2'''}.  

First, if $\delta\le \widehat\delta^*$, then   $\frac{\delta}{1-\delta} Z^*(0,0,0) \le  2$.  This follows since $ \widehat\delta^*$ satisfies   $\frac{\delta}{1-\delta} Z^*(0,0,0) =  2$, and $\frac{\delta}{1-\delta}$ is increasing in $\delta$. From the concavity of $Z^*(k,0,0)$,  $\frac{Z^*(k,0,0)}{k}$ is decreasing in $k$. Thus  \eqref{F2'''} fails for all $k$ positive.   

 Next  if $\delta> \widehat\delta^*$, then   $\frac{\delta}{1-\delta} Z^*(0,0,0) >  2$.
    Given that $Z^*(k,0,0)$ is concave, and is maximised at $k^*$, $Z(k,0)$ is increasing in $k$ for all $k<k^*$. Moreover, from the concavity of $Z(k,0)$, $\frac{Z(k,0)}{k}$ is decreasing in $k$.  Hence the  optimal solution  involves the maximal $k$ such that $\frac{\delta}{1-\delta}\frac{Z^*(k,0,0)}{k}\ge  2$. Hence the optimal investment  level $\widehat k^*(2)=k^*$  if $k^*$ satisfies \eqref{F2'''}, and $\widehat k^*(2)=\underline k^*(2,\delta)$ otherwise. \medskip
    
\emph{Uniqueness}: We next argue that no other asymmetric SSPE exists.  In any ssymmetric SSPE, it must be that $\hat r^*_1(2)=\widehat r^*_2(2) =0$. Otherwise, either MFI can lower the rate of interest. This would increase borrower utility, and relax the no default conditions.  Next, suppose   there exists an asymmetric SSPE where the aggregate loan amount is $k'$, where $k'<\widehat k^*(2)$. Clearly, the no default condition is slack at $k'$. But then M1, say, can offer a loan package $(k'+\epsilon,0)$. For $\epsilon$ small, the no default conditions are satisfied. Further, borrower utility is higher as $k'<k^*$.

3. The result follows since, by definition, $\widehat k^*(2) \le\min\{k^*,\underline k^*(2,\delta)\} \le  k^*$.\qed\vskip 1em

\emph{Proof of Proposition \ref{Mbl:comp}.}  Consider productive loans first.  If  $\widehat {k}^*(2)=k^*$, then $\widehat {k}^*(2) =k^*>\overline k\ge \widehat k(1)$, so that the claim follows. Next suppose that  $\widehat {k}^*(2) =\underline k^*(2,\delta)$.  It is sufficient to show that $\underline k^*(2,\delta)> \underline k(1,\delta)$.   Given   $\underline k(1,\delta)$ is positive, and solves  \eqref{D2''}:
\begin{equation}
\frac{f(k)}{k}=\frac{m(1-\delta)}{p\delta}+\frac{p+(1-p)m}{p}.\end{equation}
From \eqref{F2'''},   $\underline k^*(2,\delta)$ solves 
\begin{equation}
\frac{f(k)}{k}=\frac{2(1-\delta)}{p\delta}+\frac{1}{p}.\end{equation}
Next, given that $m>2$, we have that $\frac{m(1-\delta)}{p\delta}+\frac{p+(1-p)m}{p}>\frac{2(1-\delta)}{p\delta}+\frac{1}{p}$, which in turn implies that $\frac{f(\underline k(1,\delta))}{\underline k(1,\delta)}>\frac{f(\underline  k^*(2))}{\underline  k^*(2).}$ Given that $f(k,0)$ is concave in $k$, $\underline k(1,\delta)<\underline k^*(2,\delta)$. 

Finally, while moneylender loans are strictly positive when there is a single MFI,  under the asymmetric SSPE the borrower takes  no moneylender loans. \qed\vskip 1em

\clearpage

\section{Empirics}\label{appsec:empiricalappendix}

\subsection*{Construction of Control Variables}\label{app:construction}

Baseline controls are included to improve estimation precision and to adjust for potential chance imbalances across groups, as recommended even in randomized evaluations \citep{bruhn2009pursuit}.

The household head's \textit{educational attainment} is measured in completed years of schooling.

Household \textit{landholdings} are measured in decimals, where one decimal equals one-hundredth of an acre ($\approx$ 40.5 m$^{2}$).

The \textit{wealth index} is based on the first principal component from a polychoric principal component analysis (PCA), following \citep{kolenikov2009socioeconomic}, to accommodate the mix of binary indicators and asset counts. Inputs include counts of household assets (e.g., radios, TVs, sewing machines, rickshaws, bicycles, tubewells, fishing nets, ornaments, furniture, and mobile phones); livestock counts (poultry, sheep, goats, cows, and oxen); housing materials (wall and roof type); and homeownership status. The index is standardized to have mean zero and unit variance in the baseline sample. Household \textit{landholdings} in decimals are included as a separate regressor.

The \textit{household composition equivalence scale} follows the functional form 
$(\#\text{ adults} + \beta_1 \times \#\text{ children})^{\beta_2}$, 
as proposed by \citet{yatchew2003}, with parameters calibrated for Bangladesh by \citet{hasan2015}: 
$\beta_1 = 0.74$ and $\beta_2 = 0.17$. This base-independent summary measure of household composition 
is included to capture variation in household labor supply and consumption needs that may correlate with borrowing behavior.

The \textit{savings mechanism} indicator equals 1 if the household reports saving with a bank, NGO program, or savings group, and 0 if it either reports not to save or to only save at home.

The \textit{occupation} of the household head is classified into five categories: agricultural self-employment, agricultural wage labor, non-agricultural self-employment, non-agricultural wage/salaried labor, and not economically active. To capture diversification and income risk, the household’s \textit{number of income sources} is the count of distinct income-generating activities held by any member at baseline (either as their primary or secondary income source).

The \textit{shock environment}, which may affect demand for emergency credit, is measured at both the household and village level. At the household level, we include the number of shocks experienced in the year preceding baseline, across 19 types recorded in the survey (e.g., death of a main earner, illness expenses, crop disease, flood, house loss). To complement the household-specific realizations, we construct a community-level shock index by calculating the proportions of households experiencing each type of shock and extracting the first principal component from these proportions using PCA.

\subsection*{Power simulations}\label{app:power}

To assess whether the absence of statistically significant effects of MFI competition on MFI borrowing could be explained by limited statistical power, we conduct Monte Carlo power simulations calibrated to the empirical setting and estimation strategy used in the paper.

Simulated outcomes are generated from an exponential mean model consistent with the IV Poisson specification employed in the main analysis. Specifically, for each household $i$ in village $j$,
\[
E(y_{ij}) = \exp\!\left(\ln(\overline y_0) + \beta\,(MFI_j - \overline{\mathrm{MFI}}_0)\right),
\]
where $\overline y_0$ denotes the mean outcome in control villages at endline and $\overline{\mathrm{MFI}}_0$ is the corresponding control-group mean number of MFIs. Outcomes are drawn from a Poisson distribution with this conditional mean.

For each candidate effect size, we re-estimate the IV Poisson model using the same instrument, controls, and village-level clustering of standard errors as in the main specification. Power is computed as the share of 1{,}000 replications in which the null hypothesis of no effect is rejected at the 5 percent level.

The simulations indicate that the design has 80\% power to detect a reduction in the number of MFI loans of around 11.6\%, or 0.08 loans.

\subsection*{Appendix Tables and Figures}\label{app:tablesfigures}

\begin{figure}[!htbp]
    \centering
    \includegraphics[width=\textwidth]{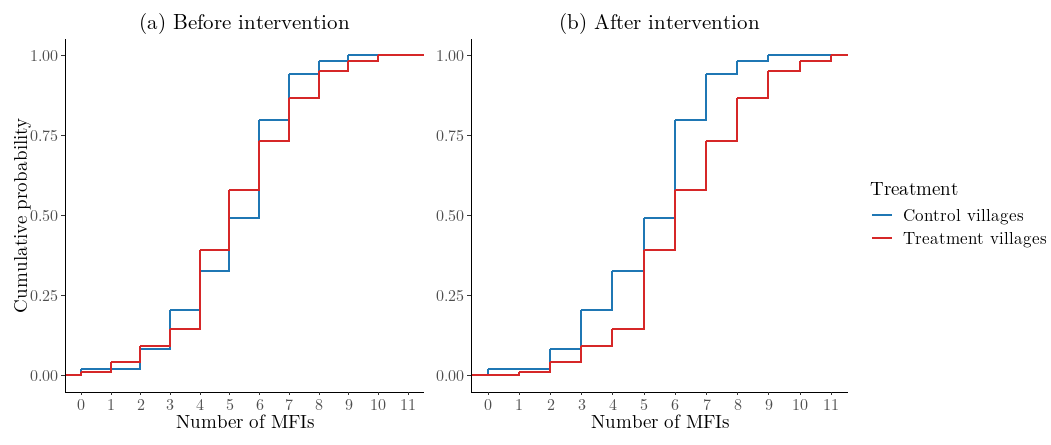}
    \caption{
The distribution of the number of MFIs at baseline (prior to the randomized intervention, panel (a)) and after the randomized entry of RDRS in the treatment villages (panel (b)).}
    \label{fig:mfi_cdf}
\end{figure}

\begin{table}[!htbp]
\vspace{-1cm}
\caption{Baseline summary statistics and balance tests, on the full sample\label{tab:summarystats1}}
  \centering
    \vspace{-.35cm}
  \begin{adjustbox}{max width=\textwidth}
  \begin{threeparttable}
  {\renewcommand{\arraystretch}{0.93}
\begin{tabular}{@{}lllllllll@{}}
\toprule
& \multicolumn{4}{c}{Full baseline sample (n=2903)} & & \multicolumn{3}{c}{Subsample re-interviewed}\\
& \multicolumn{4}{c}{} & & \multicolumn{3}{c}{at endline (n=2735)}\\
\cmidrule(l){2-5} \cmidrule(l){7-9}
\begin{tabular}[c]{@{}l@{}} \\ \\ \end{tabular}     & \begin{tabular}[c]{@{}l@{}}Full\\sample\\ (n=2903)\end{tabular} & \begin{tabular}[c]{@{}l@{}}Control\\ villages\\ (n=975)\end{tabular} & \begin{tabular}[c]{@{}l@{}}Treatment\\ villages\\ (n=1928)\end{tabular} & \begin{tabular}[c]{@{}l@{}}Treated -\\ Controls\\(p-value)\end{tabular} & & \begin{tabular}[c]{@{}l@{}}Control\\ villages\\ (n=914)\end{tabular} & \begin{tabular}[c]{@{}l@{}}Treatment\\ villages\\ (n=1921)\end{tabular} & \begin{tabular}[c]{@{}l@{}}Treated -\\ Controls\\(p-value)\end{tabular} \\
\cmidrule(l){2-9}
\vspace{-2em}\\
\textit{Panel A: Baseline outcomes} & & & &  \\
\vspace{-1.5em}\\
\# of MFI loans             & \begin{tabular}[c]{@{}l@{}}0.57\\ (0.71)\end{tabular} & \begin{tabular}[c]{@{}l@{}}0.54\\(0.68)\end{tabular} & \begin{tabular}[c]{@{}l@{}}0.58\\(0.72)\end{tabular} & 0.564 && \begin{tabular}[c]{@{}l@{}}0.78\\ (1.45)\end{tabular} & \begin{tabular}[c]{@{}l@{}}0.88\\ (1.58)\end{tabular} & 0.348 \\
\vspace{-0.95em}\\
\# of non-MFI loans             & \begin{tabular}[c]{@{}l@{}}0.92\\(1.07)\end{tabular} & \begin{tabular}[c]{@{}l@{}}0.92\\(1.06)\end{tabular} & \begin{tabular}[c]{@{}l@{}}0.92\\(1.07)\end{tabular} & 0.952  && \begin{tabular}[c]{@{}l@{}}0.92\\ (1.07)\end{tabular} & \begin{tabular}[c]{@{}l@{}}0.92\\ (1.07)\end{tabular} & 0.982 \\
\vspace{-0.95em}\\
\# of moneylender loans             & \begin{tabular}[c]{@{}l@{}}0.20\\(0.56)\end{tabular} & \begin{tabular}[c]{@{}l@{}}0.23\\(0.60)\end{tabular} & \begin{tabular}[c]{@{}l@{}}0.18\\(0.53)\end{tabular} & 0.155  && \begin{tabular}[c]{@{}l@{}}0.23\\ (0.61)\end{tabular} & \begin{tabular}[c]{@{}l@{}}0.19\\ (0.54)\end{tabular} & 0.173 \\
\vspace{-0.95em}\\
\# of loans from acquaintances       & \begin{tabular}[c]{@{}l@{}}0.62\\(0.90)\end{tabular} & \begin{tabular}[c]{@{}l@{}}0.61\\(0.90)\end{tabular} & \begin{tabular}[c]{@{}l@{}}0.62\\(0.90)\end{tabular} & 0.846 && \begin{tabular}[c]{@{}l@{}}0.60\\ (0.89)\end{tabular} & \begin{tabular}[c]{@{}l@{}}0.62\\ (0.90)\end{tabular} & 0.718 \\
\vspace{-0.95em}\\
\# of (any) loans             & \begin{tabular}[c]{@{}l@{}}1.49\\(1.26)\end{tabular} & \begin{tabular}[c]{@{}l@{}}1.47\\(1.23)\end{tabular} & \begin{tabular}[c]{@{}l@{}}1.50\\(1.27)\end{tabular} & 0.727  && \begin{tabular}[c]{@{}l@{}}1.70\\ (1.75)\end{tabular} & \begin{tabular}[c]{@{}l@{}}1.80\\ (1.90)\end{tabular} & 0.372 \\
\vspace{-0.95em}\\
\begin{tabular}[c]{@{}l@{}}Amount borrowed from\\MFIs (1000 Tk)\end{tabular}   &\begin{tabular}[c]{@{}l@{}}6.27\\(9.32)\end{tabular} & \begin{tabular}[c]{@{}l@{}}6.21\\(9.53)\end{tabular} & \begin{tabular}[c]{@{}l@{}}6.29\\(9.21)\end{tabular} & 0.917  && \begin{tabular}[c]{@{}l@{}}9.04\\ (19.44)\end{tabular} & \begin{tabular}[c]{@{}l@{}}9.65\\ (19.04)\end{tabular} & 0.667 \\  
\vspace{-0.95em}\\
\begin{tabular}[c]{@{}l@{}}Amount borrowed from\\non-MFIs\end{tabular}   &\begin{tabular}[c]{@{}l@{}}3.57\\(7.93)\end{tabular} & \begin{tabular}[c]{@{}l@{}}3.41\\(7.27)\end{tabular} & \begin{tabular}[c]{@{}l@{}}3.65\\(8.25)\end{tabular} & 0.526   && \begin{tabular}[c]{@{}l@{}}3.48\\ (7.42)\end{tabular} & \begin{tabular}[c]{@{}l@{}}3.73\\ (8.42)\end{tabular} & 0.546 \\  
\vspace{-0.95em}\\
\begin{tabular}[c]{@{}l@{}}Amount borrowed from\\moneylenders\end{tabular}   &\begin{tabular}[c]{@{}l@{}}0.81\\(3.63)\end{tabular} & \begin{tabular}[c]{@{}l@{}}0.91\\(3.59)\end{tabular} & \begin{tabular}[c]{@{}l@{}}0.77\\(3.64)\end{tabular} & 0.447 && \begin{tabular}[c]{@{}l@{}}0.92\\ (3.61)\end{tabular} & \begin{tabular}[c]{@{}l@{}}0.80\\ (3.73)\end{tabular} & 0.502 \\  
\vspace{-0.95em}\\
\begin{tabular}[c]{@{}l@{}}Amount borrowed from\\acquintances\end{tabular}   &\begin{tabular}[c]{@{}l@{}}2.33\\(6.63)\end{tabular} & \begin{tabular}[c]{@{}l@{}}2.19\\(6.13)\end{tabular} & \begin{tabular}[c]{@{}l@{}}2.45\\(6.88)\end{tabular} & 0.491   && \begin{tabular}[c]{@{}l@{}}2.23\\ (6.29)\end{tabular} & \begin{tabular}[c]{@{}l@{}}2.44\\ (7.02)\end{tabular} & 0.494 \\  
\vspace{-0.95em}\\
\begin{tabular}[c]{@{}l@{}}Total amount borrowed\\(any source)\end{tabular}   &\begin{tabular}[c]{@{}l@{}}9.84\\(12.46)\end{tabular} & \begin{tabular}[c]{@{}l@{}}9.62\\(12.44)\end{tabular} & \begin{tabular}[c]{@{}l@{}}9.95\\(12.47)\end{tabular} & 0.693 && \begin{tabular}[c]{@{}l@{}}12.52\\ (21.13)\end{tabular} & \begin{tabular}[c]{@{}l@{}}13.38\\ (20.89)\end{tabular} & 0.550 \\    
\vspace{-1em}\\
\textit{Panel B: Control variables} & & & &  \\
\vspace{-1.4em}\\
Female head of HH          & \begin{tabular}[c]{@{}l@{}}0.09\\ (0.29)\end{tabular}   & \begin{tabular}[c]{@{}l@{}}0.07\\ (0.26)\end{tabular}   & \begin{tabular}[c]{@{}l@{}}0.10\\ (0.30)\end{tabular}   & 0.038 && \begin{tabular}[c]{@{}l@{}}0.07\\ (0.25)\end{tabular} & \begin{tabular}[c]{@{}l@{}}0.10\\ (0.30)\end{tabular} & 0.020 \\    
\vspace{-0.95em}\\
Head of HH's age          & \begin{tabular}[c]{@{}l@{}}39.72\\ (10.91)\end{tabular}   & \begin{tabular}[c]{@{}l@{}}39.47\\ (10.98)\end{tabular}   & \begin{tabular}[c]{@{}l@{}}39.85\\ (10.88)\end{tabular}   & 0.454 && \begin{tabular}[c]{@{}l@{}}39.66\\ (11.04)\end{tabular} & \begin{tabular}[c]{@{}l@{}}40.11\\ (10.8)\end{tabular} & 0.377 \\    
\vspace{-0.95em}\\
Head of HH's education (years)          & \begin{tabular}[c]{@{}l@{}}2.16\\ (3.25)\end{tabular}   & \begin{tabular}[c]{@{}l@{}}2.12\\ (3.24)\end{tabular}   & \begin{tabular}[c]{@{}l@{}}2.18\\ (3.25)\end{tabular}   & 0.713  && \begin{tabular}[c]{@{}l@{}}2.15\\ (3.25)\end{tabular} & \begin{tabular}[c]{@{}l@{}}2.18\\ (3.28)\end{tabular} & 0.814 \\    
\vspace{-0.95em}\\
Literate household                  & \begin{tabular}[c]{@{}l@{}}0.86\\ (0.34)\end{tabular}   & \begin{tabular}[c]{@{}l@{}}0.86\\ (0.35)\end{tabular}   & \begin{tabular}[c]{@{}l@{}}0.87\\ (0.34)\end{tabular}   & 0.482  && \begin{tabular}[c]{@{}l@{}}2.15\\ (3.25)\end{tabular} & \begin{tabular}[c]{@{}l@{}}2.18\\ (3.28)\end{tabular} & 0.814 \\    
\vspace{-0.95em}\\
Muslim household       & \begin{tabular}[c]{@{}l@{}}0.78\\ (0.42)\end{tabular}   & \begin{tabular}[c]{@{}l@{}}0.80\\ (0.40)\end{tabular}   & \begin{tabular}[c]{@{}l@{}}0.77\\ (0.42)\end{tabular}   & 0.554   && \begin{tabular}[c]{@{}l@{}}0.79\\ (0.41)\end{tabular} & \begin{tabular}[c]{@{}l@{}}0.76\\ (0.43)\end{tabular} & 0.525 \\    
\vspace{-0.95em}\\
Household size             & \begin{tabular}[c]{@{}l@{}}4.13\\ (1.33)\end{tabular}   & \begin{tabular}[c]{@{}l@{}}4.16\\ (1.31)\end{tabular}   & \begin{tabular}[c]{@{}l@{}}4.12\\ (1.35)\end{tabular}   & 0.548  && \begin{tabular}[c]{@{}l@{}}4.18\\ (1.33)\end{tabular} & \begin{tabular}[c]{@{}l@{}}4.15\\ (1.35)\end{tabular} & 0.577 \\    
\vspace{-0.95em}\\
HH composition equiv. scale             & \begin{tabular}[c]{@{}l@{}}1.39\\ (0.08)\end{tabular}   & \begin{tabular}[c]{@{}l@{}}4.16\\ (1.31)\end{tabular}   & \begin{tabular}[c]{@{}l@{}}4.12\\ (1.35)\end{tabular}   & 0.548  && \begin{tabular}[c]{@{}l@{}}4.18\\ (1.33)\end{tabular} & \begin{tabular}[c]{@{}l@{}}4.15\\ (1.35)\end{tabular} & 0.539 \\    
\vspace{-0.95em}\\
Landholdings (in decimals)              & \begin{tabular}[c]{@{}l@{}}13.48\\ (22.67)\end{tabular} & \begin{tabular}[c]{@{}l@{}}13.97\\ (22.18)\end{tabular} & \begin{tabular}[c]{@{}l@{}}13.23\\ (22.90)\end{tabular} & 0.566  && \begin{tabular}[c]{@{}l@{}}14.42\\ (22.64)\end{tabular} & \begin{tabular}[c]{@{}l@{}}13.65\\ (23.16)\end{tabular} & 0.552 \\    
\vspace{-0.95em}\\
HH's wealth index               & \begin{tabular}[c]{@{}l@{}}0.00\\ (1.00)\end{tabular}   & \begin{tabular}[c]{@{}l@{}}0.04\\ (1.64)\end{tabular}  & \begin{tabular}[c]{@{}l@{}}0.02\\ (1.74)\end{tabular}   & 0.868 && \begin{tabular}[c]{@{}l@{}}0.04\\ (0.96)\end{tabular} & \begin{tabular}[c]{@{}l@{}}0.02\\ (1.03)\end{tabular} & 0.853 \\    
\vspace{-0.95em}\\
HH uses a savings mechanism          & \begin{tabular}[c]{@{}l@{}}0.15\\ (0.36)\end{tabular}   & \begin{tabular}[c]{@{}l@{}}0.16\\ (0.37)\end{tabular}   & \begin{tabular}[c]{@{}l@{}}0.15\\ (0.35)\end{tabular}   & 0.401  && \begin{tabular}[c]{@{}l@{}}0.16\\ (0.37)\end{tabular} & \begin{tabular}[c]{@{}l@{}}0.15\\ (0.35)\end{tabular} & 0.325 \\    

\vspace{-0.95em}\\
HH's \# of income sources         & \begin{tabular}[c]{@{}l@{}}0.69\\ (0.89)\end{tabular}   & \begin{tabular}[c]{@{}l@{}}0.73\\ (0.92)\end{tabular}   & \begin{tabular}[c]{@{}l@{}}0.67\\ (0.87)\end{tabular}   & 0.226  && \begin{tabular}[c]{@{}l@{}}0.67\\ (0.87)\end{tabular} & \begin{tabular}[c]{@{}l@{}}0.73\\ (0.93)\end{tabular} & 0.348 \\    
\vspace{-0.95em}\\
HH's \# of shocks experienced         & \begin{tabular}[c]{@{}l@{}}0.69\\ (0.89)\end{tabular}   & \begin{tabular}[c]{@{}l@{}}4.16\\ (1.31)\end{tabular}   & \begin{tabular}[c]{@{}l@{}}4.12\\ (1.35)\end{tabular}   & 0.367  && \begin{tabular}[c]{@{}l@{}}3.62\\ (0.81)\end{tabular} & \begin{tabular}[c]{@{}l@{}}3.57\\ (0.84)\end{tabular} & 0.390 \\    
\vspace{-0.95em}\\
Village shock exposure index         & \begin{tabular}[c]{@{}l@{}}0\\ (1)\end{tabular}   & \begin{tabular}[c]{@{}l@{}}0.06\\ (0.97)\end{tabular}   & \begin{tabular}[c]{@{}l@{}}-0.03\\ (1.01)\end{tabular}   & 0.590  && \begin{tabular}[c]{@{}l@{}}0.07\\ (0.98)\end{tabular} & \begin{tabular}[c]{@{}l@{}}-0.04\\ (1.02)\end{tabular} & 0.556 \\    

\vspace{-1em}\\
\textit{Panel C: Attrition} & & & &  \\
\vspace{-1.5em}\\
Re-interviewed at endline        & \begin{tabular}[c]{@{}l@{}}0.94\\ (0.23)\end{tabular}   & \begin{tabular}[c]{@{}l@{}}0.94\\ (0.24)\end{tabular}   & \begin{tabular}[c]{@{}l@{}}0.94\\ (0.23)\end{tabular}   & 0.551  && & & \\
\vspace{-1em}\\
\textit{Panel D: Omnibus test} & & & &  && & & \\
\vspace{-0.9em}\\
Joint orthogonality F-test (p-value)        &   &   &  & 0.941 && & & 0.826  \\
\bottomrule
\end{tabular}
}
\begin{tablenotes}
\item \emph{Notes:} Baseline means and standard deviations are shown for treatment and control villages. P-values are based on regressions of the treatment indicator on the variable in the left-most column (or, for the omnibus test in Panel D, on all baseline covariates, the total number of loans, the total amount borrowed, and the re-interview indicator), with standard errors clustered at the village level.
\end{tablenotes}
\end{threeparttable}
\end{adjustbox}
\end{table}

\begin{table}[!htbp]
\vspace{-1.35cm}
\caption{Baseline summary statistics and balance tests, on the subsample that at baseline expressed an interest in borrowing\label{tab:summarystats2}}
  \centering
    \vspace{-.35cm}
  \begin{adjustbox}{max width=1.015\textwidth}
  \begin{threeparttable}
  {\renewcommand{\arraystretch}{0.93}
\begin{tabular}{@{}lllllllll@{}}
\toprule
& \multicolumn{4}{c}{Baseline (interested in borrrowing) sample} & & \multicolumn{3}{c}{Subsample re-interviewed}\\
& \multicolumn{4}{c}{(n=1450)} & & \multicolumn{3}{c}{at endline (n=1374)}\\
\cmidrule(l){2-5} \cmidrule(l){7-9}
\begin{tabular}[c]{@{}l@{}} \\ \\ \end{tabular}     & \begin{tabular}[c]{@{}l@{}}Full\\sample\\ (n=1450)\end{tabular} & \begin{tabular}[c]{@{}l@{}}Control\\ villages\\ (n=487)\end{tabular} & \begin{tabular}[c]{@{}l@{}}Treatment\\ villages\\ (n=963)\end{tabular} & \begin{tabular}[c]{@{}l@{}}Treated -\\ Controls\\(p-value)\end{tabular} & & \begin{tabular}[c]{@{}l@{}}Control\\ villages\\ (n=460)\end{tabular} & \begin{tabular}[c]{@{}l@{}}Treatment\\ villages\\ (n=914)\end{tabular} & \begin{tabular}[c]{@{}l@{}}Treated -\\ Controls\\(p-value)\end{tabular} \\
\cmidrule(l){2-9}
\vspace{-2em}\\
\textit{Panel A: Baseline outcomes} & & & &  \\
\vspace{-1.5em}\\
\# of MFI loans             & \begin{tabular}[c]{@{}l@{}}0.56\\ (0.68)\end{tabular} & \begin{tabular}[c]{@{}l@{}}0.55\\(0.66)\end{tabular} & \begin{tabular}[c]{@{}l@{}}0.56\\(0.69)\end{tabular} & 0.837 && \begin{tabular}[c]{@{}l@{}}0.56\\ (0.66)\end{tabular} & \begin{tabular}[c]{@{}l@{}}0.58\\ (0.70)\end{tabular} & 0.746 \\
\vspace{-0.95em}\\
\# of non-MFI loans              & \begin{tabular}[c]{@{}l@{}}0.95\\(1.08)\end{tabular} & \begin{tabular}[c]{@{}l@{}}1.03\\(1.10)\end{tabular} & \begin{tabular}[c]{@{}l@{}}0.91\\(1.06)\end{tabular} & 0.172  && \begin{tabular}[c]{@{}l@{}}1.02\\ (1.10)\end{tabular} & \begin{tabular}[c]{@{}l@{}}0.91\\ (1.07)\end{tabular} & 0.203 \\
\vspace{-0.95em}\\
\# of moneylender loans             & \begin{tabular}[c]{@{}l@{}}0.22\\(0.60)\end{tabular} & \begin{tabular}[c]{@{}l@{}}0.26\\(0.65)\end{tabular} & \begin{tabular}[c]{@{}l@{}}0.20\\(0.57)\end{tabular} & 0.160  && \begin{tabular}[c]{@{}l@{}}0.26\\ (0.66)\end{tabular} & \begin{tabular}[c]{@{}l@{}}0.21\\ (0.58)\end{tabular} & 0.200 \\
\vspace{-0.95em}\\
\# of loans from acquaintances       & \begin{tabular}[c]{@{}l@{}}0.62\\(0.90)\end{tabular} & \begin{tabular}[c]{@{}l@{}}0.67\\(0.93)\end{tabular} & \begin{tabular}[c]{@{}l@{}}0.59\\(0.89)\end{tabular} & 0.230 && \begin{tabular}[c]{@{}l@{}}0.65\\ (0.92)\end{tabular} & \begin{tabular}[c]{@{}l@{}}0.58\\ (0.89)\end{tabular} & 0.290 \\
\vspace{-0.95em}\\
\# of (any) loans             & \begin{tabular}[c]{@{}l@{}}1.51\\(1.23)\end{tabular} & \begin{tabular}[c]{@{}l@{}}1.58\\(1.24)\end{tabular} & \begin{tabular}[c]{@{}l@{}}1.47\\(1.23)\end{tabular} & 0.204  && \begin{tabular}[c]{@{}l@{}}1.58\\ (1.25)\end{tabular} & \begin{tabular}[c]{@{}l@{}}1.49\\ (1.24)\end{tabular} & 0.277 \\
\vspace{-0.95em}\\
\begin{tabular}[c]{@{}l@{}}Amount borrowed from\\MFIs (1000 Tk)\end{tabular}   &\begin{tabular}[c]{@{}l@{}}6.01\\(8.72)\end{tabular} & \begin{tabular}[c]{@{}l@{}}8.34\\(16.11)\end{tabular} & \begin{tabular}[c]{@{}l@{}}5.96\\(8.59)\end{tabular} & 0.856  && \begin{tabular}[c]{@{}l@{}}6.14\\ (8.71)\end{tabular} & \begin{tabular}[c]{@{}l@{}}8.72\\ (16.56)\end{tabular} & 0.981 \\  
\vspace{-0.95em}\\
\begin{tabular}[c]{@{}l@{}}Amount borrowed from\\non-MFIs\end{tabular}   &\begin{tabular}[c]{@{}l@{}}3.66\\(8.26)\end{tabular} & \begin{tabular}[c]{@{}l@{}}6.10\\(8.99)\end{tabular} & \begin{tabular}[c]{@{}l@{}}3.55\\(8.29)\end{tabular} & 0.517   && \begin{tabular}[c]{@{}l@{}}4.01\\ (8.43)\end{tabular} & \begin{tabular}[c]{@{}l@{}}3.63\\ (8.46)\end{tabular} & 0.490 \\  
\vspace{-0.95em}\\
\begin{tabular}[c]{@{}l@{}}Amount borrowed from\\moneylenders\end{tabular}   &\begin{tabular}[c]{@{}l@{}}0.91\\(3.70)\end{tabular} & \begin{tabular}[c]{@{}l@{}}1.05\\(3.77)\end{tabular} & \begin{tabular}[c]{@{}l@{}}0.85\\(3.66)\end{tabular} & 0.441 && \begin{tabular}[c]{@{}l@{}}1.08\\ (3.86)\end{tabular} & \begin{tabular}[c]{@{}l@{}}0.88\\ (3.75)\end{tabular} & 0.464 \\  
\vspace{-0.95em}\\
\begin{tabular}[c]{@{}l@{}}Amount borrowed from\\acquaintances\end{tabular}   &\begin{tabular}[c]{@{}l@{}}2.30\\(7.05)\end{tabular} & \begin{tabular}[c]{@{}l@{}}2.43\\(7.11)\end{tabular} & \begin{tabular}[c]{@{}l@{}}2.24\\(7.03)\end{tabular} & 0.653   && \begin{tabular}[c]{@{}l@{}}2.49\\ (7.30)\end{tabular} & \begin{tabular}[c]{@{}l@{}}2.28\\ (7.17)\end{tabular} & 0.626 \\  
\vspace{-0.95em}\\
\begin{tabular}[c]{@{}l@{}}Total amount borrowed\\(any source)\end{tabular}   &\begin{tabular}[c]{@{}l@{}}9.67\\(12.14)\end{tabular} & \begin{tabular}[c]{@{}l@{}}9.99\\(12.60)\end{tabular} & \begin{tabular}[c]{@{}l@{}}9.51\\(11.90)\end{tabular} & 0.599 && \begin{tabular}[c]{@{}l@{}}10.16\\ (12.76)\end{tabular} & \begin{tabular}[c]{@{}l@{}}9.77\\ (12.07)\end{tabular} & 0.673 \\    
\vspace{-1em}\\
\textit{Panel B: Control variables} & & & &  \\
\vspace{-1.4em}\\
Female head of HH          & \begin{tabular}[c]{@{}l@{}}0.11\\ (0.31)\end{tabular}   & \begin{tabular}[c]{@{}l@{}}0.08\\ (0.28)\end{tabular}   & \begin{tabular}[c]{@{}l@{}}0.12\\ (0.32)\end{tabular}   & 0.094 && \begin{tabular}[c]{@{}l@{}}0.08\\ (0.27)\end{tabular} & \begin{tabular}[c]{@{}l@{}}0.12\\ (0.32)\end{tabular} & 0.033 \\    
\vspace{-0.95em}\\
Head of HH's age          & \begin{tabular}[c]{@{}l@{}}39.17\\ (10.23)\end{tabular}   & \begin{tabular}[c]{@{}l@{}}38.77\\ (10.02)\end{tabular}   & \begin{tabular}[c]{@{}l@{}}39.36\\ (10.34)\end{tabular}   & 0.347 && \begin{tabular}[c]{@{}l@{}}38.99\\ (10.00)\end{tabular} & \begin{tabular}[c]{@{}l@{}}39.56\\ (10.33)\end{tabular} & 0.364 \\    
\vspace{-0.95em}\\
Head of HH’s education (years)  & \begin{tabular}[c]{@{}l@{}}2.02\\ (3.14)\end{tabular}   & \begin{tabular}[c]{@{}l@{}}1.86\\ (3.05)\end{tabular}   & \begin{tabular}[c]{@{}l@{}}2.11\\ (3.19)\end{tabular}   & 0.179  && \begin{tabular}[c]{@{}l@{}}1.86\\ (3.05)\end{tabular} & \begin{tabular}[c]{@{}l@{}}2.11\\ (3.20)\end{tabular} & 0.178 \\    
\vspace{-0.95em}\\
Literate household & \begin{tabular}[c]{@{}l@{}}0.85\\ (0.35)\end{tabular}   & \begin{tabular}[c]{@{}l@{}}0.85\\ (0.36)\end{tabular}   & \begin{tabular}[c]{@{}l@{}}0.86\\ (0.35)\end{tabular}   & 0.692   && \begin{tabular}[c]{@{}l@{}}0.85\\ (0.36)\end{tabular} & \begin{tabular}[c]{@{}l@{}}0.86\\ (0.35)\end{tabular} & 0.524 \\    
\vspace{-0.95em}\\
Muslim household       & \begin{tabular}[c]{@{}l@{}}0.78\\ (0.41)\end{tabular}   & \begin{tabular}[c]{@{}l@{}}0.82\\ (0.39)\end{tabular}   & \begin{tabular}[c]{@{}l@{}}0.76\\ (0.43)\end{tabular}   & 0.325   && \begin{tabular}[c]{@{}l@{}}0.81\\ (0.39)\end{tabular} & \begin{tabular}[c]{@{}l@{}}0.76\\ (0.43)\end{tabular} & 0.334 \\    
\vspace{-0.95em}\\
Household size             & \begin{tabular}[c]{@{}l@{}}4.10\\ (1.33)\end{tabular}   & \begin{tabular}[c]{@{}l@{}}4.17\\ (1.26)\end{tabular}   & \begin{tabular}[c]{@{}l@{}}4.07\\ (1.36)\end{tabular}   & 0.273  && \begin{tabular}[c]{@{}l@{}}4.19\\ (1.27)\end{tabular} & \begin{tabular}[c]{@{}l@{}}4.11\\ (1.36)\end{tabular} & 0.335 \\    
\vspace{-0.95em}\\
HH composition equiv. scale& \begin{tabular}[c]{@{}l@{}}1.39\\ (0.08)\end{tabular}   & \begin{tabular}[c]{@{}l@{}}1.39\\ (0.08)\end{tabular}   & \begin{tabular}[c]{@{}l@{}}1.38\\ (0.08)\end{tabular}   & 0.152  && \begin{tabular}[c]{@{}l@{}}1.39\\ (0.77)\end{tabular} & \begin{tabular}[c]{@{}l@{}}1.39\\ (0.08)\end{tabular} & 0.207 \\    
\vspace{-0.95em}\\
Landholdings (in decimals)              & \begin{tabular}[c]{@{}l@{}}12.35\\ (21.51)\end{tabular} & \begin{tabular}[c]{@{}l@{}}12.81\\ (21.62)\end{tabular} & \begin{tabular}[c]{@{}l@{}}12.12\\ (21.46)\end{tabular} & 0.652  && \begin{tabular}[c]{@{}l@{}}13.20\\ (22.08)\end{tabular} & \begin{tabular}[c]{@{}l@{}}12.46\\ (21.69)\end{tabular} & 0.640 \\    
\vspace{-0.95em}\\
HH's wealth index               & \begin{tabular}[c]{@{}l@{}}-0.04\\ (1.00)\end{tabular}   & \begin{tabular}[c]{@{}l@{}}-0.06\\ (0.93)\end{tabular}  & \begin{tabular}[c]{@{}l@{}}-0.03\\ (1.04)\end{tabular}   & 0.751 && \begin{tabular}[c]{@{}l@{}}-0.04\\ (0.94)\end{tabular} & \begin{tabular}[c]{@{}l@{}}-0.01\\ (1.04)\end{tabular} & 0.639 \\    
\vspace{-0.95em}\\
HH uses a savings mechanism               & \begin{tabular}[c]{@{}l@{}}0.16\\ (0.37)\end{tabular}   & \begin{tabular}[c]{@{}l@{}}0.18\\ (0.39)\end{tabular}  & \begin{tabular}[c]{@{}l@{}}0.15\\ (0.36)\end{tabular}   & 0.185 && \begin{tabular}[c]{@{}l@{}}0.18\\ (0.39)\end{tabular} & \begin{tabular}[c]{@{}l@{}}0.15\\ (0.36)\end{tabular} & 0.175 \\   
\vspace{-0.95em}\\
HH's \# of income sources               & \begin{tabular}[c]{@{}l@{}}3.55\\ (0.83)\end{tabular}   & \begin{tabular}[c]{@{}l@{}}3.59\\ (0.79)\end{tabular}  & \begin{tabular}[c]{@{}l@{}}3.53\\ (0.85)\end{tabular}   & 0.261 && \begin{tabular}[c]{@{}l@{}}3.59\\ (0.79)\end{tabular} & \begin{tabular}[c]{@{}l@{}}3.54\\ (0.85)\end{tabular} & 0.388 \\   
\vspace{-0.95em}\\
HH's \# of shocks experienced              & \begin{tabular}[c]{@{}l@{}}0.71\\ (0.92)\end{tabular}   & \begin{tabular}[c]{@{}l@{}}0.78\\ (0.96)\end{tabular}  & \begin{tabular}[c]{@{}l@{}}0.67\\ (0.89)\end{tabular}   & 0.179 && \begin{tabular}[c]{@{}l@{}}0.77\\ (0.96)\end{tabular} & \begin{tabular}[c]{@{}l@{}}0.67\\ (0.89)\end{tabular} & 0.212 \\   
\vspace{-0.95em}\\
Village shock exposure index              & \begin{tabular}[c]{@{}l@{}}0.00\\ (1.00)\end{tabular}   & \begin{tabular}[c]{@{}l@{}}0.06\\ (0.98)\end{tabular}  & \begin{tabular}[c]{@{}l@{}}-0.03\\ (1.01)\end{tabular}   & 0.582 && \begin{tabular}[c]{@{}l@{}}0.07\\ (0.98)\end{tabular} & \begin{tabular}[c]{@{}l@{}}-0.03\\ (1.02)\end{tabular} & 0.573 \\   
\vspace{-1em}\\
\textit{Panel C: Attrition} & & & &  \\
\vspace{-1.5em}\\
Re-interviewed at endline        & \begin{tabular}[c]{@{}l@{}}0.95\\ (0.22)\end{tabular}   & \begin{tabular}[c]{@{}l@{}}0.94\\ (0.23)\end{tabular}   & \begin{tabular}[c]{@{}l@{}}0.95\\ (0.22)\end{tabular}   & 0.739 && & & \\
\vspace{-1em}\\
\textit{Panel D: Omnibus test} & & & &  && & & \\
\vspace{-0.9em}\\
Joint orthogonality F-test (p-value)        &   &   &  & 0.396 && & & 0.375  \\
\bottomrule
\end{tabular}
}
\begin{tablenotes}
\item \emph{Notes:} Baseline means and standard deviations are shown for treatment and control villages. P-values are based on regressions of the treatment indicator on the variable in the left-most column (or, for the omnibus test in Panel D, on all baseline covariates, the total number of loans, the total amount borrowed, and the re-interview indicator), with standard errors clustered at the village level.
\end{tablenotes}
\end{threeparttable}
\end{adjustbox}
\end{table}

\begin{table}[t]
\centering
\caption{Attrition robustness checks: IPW weighting and Lee bounds}
\vspace{-0.6em} 
\label{tab:attrition}

\newlength{\colW}   \setlength{\colW}{0.135\textwidth}
\newlength{\colOne} \setlength{\colOne}{0.05\textwidth}
\newlength{\labW}
\setlength{\labW}{\dimexpr\textwidth - \colOne - 3\colW - 8\tabcolsep\relax}

\newcolumntype{Y}{>{\centering\arraybackslash}p{\colW}}
\newcolumntype{O}{>{\centering\arraybackslash}p{\colOne}}
\newcolumntype{L}{>{\raggedright\arraybackslash}p{\labW}}

\footnotesize
\begin{tabular}{@{} L O Y Y Y @{}} 
\toprule
                     & (1) & (2) & (3) & (4) \\
                     &     &     & \multicolumn{2}{c}{Lee bounds} \\  \cmidrule(lr){4-5}
                     & Point    & Point est. & \multicolumn{2}{c}{$[$LB , UB$]$}  \\ 
                     &  est.          &   IPW                       & $\{$LB 95\% CI$\}$ & $\{$UB 95\% CI$\}$ \\
\midrule
\multicolumn{5}{l}{\textit{\hspace{-1em} Panel A: Full sample}}\\
\multirow{2}{=}{\# of loans from MFIs}
   & 0.01 & 0.01 & \multicolumn{2}{c}{[-0.008 ; 0.016]} \\
    [-1.15em]\\
  & (0.06) & (0.06) & \{-0.151,0.160\} & \{-0.140,0.180\} \\
\addlinespace
[-0.72em]\\
\multirow{2}{=}{\# of loans from non-MFIs}
   & -0.08 & -0.08 & \multicolumn{2}{c}{[-0.103 ; -0.077]} \\
   [-1.1em]\\
   & (0.05) & (0.05) & \{-0.218,0.017\} & \{-0.172,0.030\} \\
\addlinespace
[-0.72em]\\
\multirow{2}{=}{\# of loans from moneylenders}
   & -0.31 & -0.32 & \multicolumn{2}{c}{[-0.361 ; -0.306]} \\
    [-1.15em]\\
   & (0.13) & (0.13) & \{-0.662,-0.123\} & \{-0.598,-0.134\} \\
\addlinespace
[-0.72em]\\
\multirow{2}{=}{\# of loans from acquaintances}
   & -0.02 & -0.02 & \multicolumn{2}{c}{[-0.063 ; -0.015]} \\
    [-1.15em]\\
  & (0.07) & (0.07) & \{-0.209,0.200\} & \{-0.146,0.214\} \\
\addlinespace
[-0.72em]\\
\multirow{2}{=}{\# of loans (any)}
   & -0.06 & -0.06 & \multicolumn{2}{c}{[-0.075 ; -0.056]} \\
    [-1.15em]\\
   & (0.04) & (0.04) & \{-0.178,0.012\} & \{-0.152,0.028\} \\
\addlinespace
[-0.72em]\\
\multirow{2}{=}{Amount borrowed from MFIs}
   & 0.01 & -0.01 & \multicolumn{2}{c}{[-0.015 ; 0.018]} \\
    [-1.15em]\\
   & (0.06) & (0.07) & \{-0.149,0.168\} & \{-0.123,0.219\} \\
\addlinespace
[-0.72em]\\
\multirow{2}{=}{Amount borrowed from non-MFIs}
   & 0.02 & 0.03 & \multicolumn{2}{c}{[-0.043 ; 0.029]} \\
    [-1.15em]\\
   & (0.07) & (0.08) & \{-0.311,0.339\} & \{-0.132,0.255\} \\
\addlinespace
[-0.72em]\\
\multirow{2}{=}{Amount borrowed from moneylenders}
   & -0.27 & -0.27 & \multicolumn{2}{c}{[-0.429 ; -0.264]} \\
    [-1.15em]\\
   & (0.17) & (0.18) & \{-0.749,0.130\} & \{-0.633,0.055\} \\
\addlinespace
[-0.72em]\\
\multirow{2}{=}{Amount borrowed from acquaintances}
   & 0.01 & 0.02 & \multicolumn{2}{c}{[-0.054 ; 0.021]} \\
    [-1.15em]\\
   & (0.09) & (0.09) & \{-0.341,0.408\} & \{-0.163,0.300\} \\
\addlinespace
[-0.72em]\\
\multirow{2}{=}{Amount borrowed (any source)}
   & -0.00 & 0.00 & \multicolumn{2}{c}{[-0.035 ; 0.005]} \\
    [-1.15em]\\
   & (0.05) & (0.05) & \{-0.182,0.135\} & \{-0.106,0.151\} \\
\addlinespace
[-0.72em]\\
\multicolumn{5}{l}{\textit{\hspace{-1em} Panel B: `Willing to borrow' sample}}\\
\multirow{2}{=}{\# of loans from MFIs}
   & 0.06 & 0.05 & \multicolumn{2}{c}{[-0.054 ; 0.069]} \\
    [-1.15em]\\
   & (0.06) & (0.06) & \{-0.114,0.261\} & \{-0.117,0.266\} \\
\addlinespace
[-0.72em]\\
\multirow{2}{=}{\# of loans from non-MFIs}
   & -0.10 & -0.10 & \multicolumn{2}{c}{[-0.111 ; -0.100]} \\
    [-1.15em]\\
   & (0.06) & (0.06) & \{-0.247,0.014\} & \{-0.220,0.019\} \\
\addlinespace
[-0.72em]\\
\multirow{2}{=}{\# of loans from moneylenders}
   & -0.44 & -0.44 & \multicolumn{2}{c}{[-0.451 ; -0.433]} \\
    [-1.15em]\\
   & (0.22) & (0.23) & \{-0.871,-0.194\} & \{-0.867,-0.206\} \\
\addlinespace
[-0.72em]\\
\multirow{2}{=}{\# of loans from acquaintances}
   & -0.06 & -0.06 & \multicolumn{2}{c}{[-0.081 ; -0.050]} \\
    [-1.15em]\\
   & (0.10) & (0.10) & \{-0.271,0.245\} & \{-0.233,0.223\} \\
\addlinespace
[-0.72em]\\
\multirow{2}{=}{\# of loans (any)}
   & -0.06 & -0.06 & \multicolumn{2}{c}{[-0.068 ; -0.054]} \\
    [-1.15em]\\
   & (0.04) & (0.04) & \{-0.179,0.029\} & \{-0.155,0.029\} \\
\addlinespace
[-0.72em]\\
\multirow{2}{=}{Amount borrowed from MFIs}
   & 0.07 & 0.04 & \multicolumn{2}{c}{[0.060 ; 0.079]} \\
    [-1.15em]\\
   & (0.07) & (0.06) & \{-0.106,0.291\} & \{-0.091,0.301\} \\
\addlinespace
[-0.72em]\\
\multirow{2}{=}{Amount borrowed from non-MFIs}
   & -0.06 & -0.05 & \multicolumn{2}{c}{[-0.097 ; -0.051]} \\
    [-1.15em]\\
   & (0.09) & (0.09) & \{-0.328,0.332\} & \{-0.225,0.200\} \\
\addlinespace
[-0.72em]\\
\multirow{2}{=}{Amount borrowed from moneylenders}
   & -0.38 & -0.39 & \multicolumn{2}{c}{[-0.385 ; -0.374]} \\
    [-1.15em]\\
   & (0.21) & (0.23) & \{-0.937,-0.098\} & \{-0.883,-0.129\} \\
\addlinespace
[-0.72em]\\
\multirow{2}{=}{Amount borrowed from acquaintances}
   & 0.02 & 0.00 & \multicolumn{2}{c}{[-0.094 ; 0.026]} \\
    [-1.15em]\\
   & (0.13) & (0.13) & \{-0.372,0.758\} & \{-0.247,0.560\} \\
\addlinespace
[-0.72em]\\
\multirow{2}{=}{Amount borrowed (any source)}
   & 0.01 & -0.00 & \multicolumn{2}{c}{[-0.022 ; 0.012]} \\
    [-1.15em]\\
   & (0.05) & (0.05) & \{-0.131,0.144\} & \{-0.103,0.155\} \\
\midrule
\multicolumn{5}{@{}p{\textwidth}@{}}{\footnotesize
\textit{Notes:} Columns (1) and (2) report baseline and inverse-probability weighted (IPW) point estimates, respectively, with cluster-robust standard errors in parentheses. The IPW estimates re-weight observations by the inverse of the estimated probability of being re-interviewed at endline, defined as $w(z,x) = [Pr(R=1|z,x)/Pr(R=1|x)]^{-1}$, where $R=1$ if the household was re-interviewed, $x$ are the covariates included in the regression whose robustness is evaluated, and $z$ are auxiliary variables predicting attrition \citep{john1998analysis,wooldridge2002inverse}. The set $z$ includes the respondent’s gender, age, literacy status, and years of education, enumerator indicators, and an indicator for whether the baseline interview took place on a Friday (the day of congregational prayer for Muslims). Columns (3)–(4) report Lee bounds with 95\% bootstrap confidence intervals (in braces), based on 1,000 bootstrap replications.}\\
\end{tabular}
\end{table}

\begin{table}[!htbp]
\small
\caption{Linear IV estimates of the effect of MFI competition on loan counts by credit source\label{tab:nrloanslinear}}
\centering
  \begin{threeparttable}
\begin{tabular}{llllll}
\toprule
\multicolumn{1}{c}{} & \multicolumn{1}{c}{(1)}                                          & \multicolumn{1}{c}{(2)}                                             & \multicolumn{1}{c}{(3)}                                                 & \multicolumn{1}{c}{(4)}       & \multicolumn{1}{c}{(5)}                               \\
 & \begin{tabular}[c]{@{}l@{}}\# of loans \\ from MFIs \textcolor{white}{iii}\end{tabular} & \begin{tabular}[c]{@{}l@{}}\# of loans\\ from non-MFIs\end{tabular} & \begin{tabular}[c]{@{}l@{}}\# of loans from\\ moneylenders\end{tabular} & \begin{tabular}[c]{@{}l@{}}\# of loans from\\ acquintances\end{tabular} & \begin{tabular}[c]{@{}l@{}}\# of loans\\ (any)\end{tabular} \\
\midrule
\multicolumn{6}{@{}l}{\textit{Panel A: Full sample}}\\
\vspace{-.6em}\\
\# of MFIs  &        0.00   &       -0.09   &       -0.04*  &       -0.04   &       -0.09   \\
            &      (0.03)   &      (0.06)   &      (0.03)   &      (0.05)   &      (0.07)   \\
\vspace{-.6em}\\
E(dep. var.)  & 0.72 & 1.20 & 0.24 & 0.76 & 1.92 \\
(st. dev.)    & (0.80) & (1.16) & (0.57) & (0.99) & (1.53) \\
\vspace{-.6em}\\
\% of zeros in dep. var. & 46.1\% & 32.0\% & 81.4\% & 51.4\% & 18.3\% \\
\vspace{-.6em}\\
Observations                   &        2735   &        2735   &        2735   &       2735   &        2735   \\
\vspace{-.6em}\\
Kleibergen--Paap F-stat.  & 20.93 & 21.09 & 20.96 & 21.10 & 20.80 \\
\\
\multicolumn{6}{@{}l}{\textit{Panel B: `Willing to borrow' sample}}\\

\vspace{-.6em}\\
\# of MFIs  &        0.01   &       -0.10   &       -0.08** &       -0.02   &       -0.08   \\
            &      (0.04)   &      (0.07)   &      (0.04)   &      (0.06)   &      (0.08)   \\
\vspace{-.6em}\\
E(dep. var.)  & 0.76 & 1.23 & 0.27 & 0.76 & 1.99 \\
(st. dev.)    & (0.80) & (1.14) & (0.59) & (0.99) & (1.51) \\
\vspace{-.6em}\\
Zeros (\%)  & 42.6\% & 30.7\% & 79.0\% & 51.6\% & 16.0\% \\
\vspace{-.6em}\\
Observations                   &        1374   &        1374   &        1374   &       1374   &        1374   \\
\vspace{-.6em}\\
Kleibergen--Paap F-stat.  & 21.18 & 21.01 & 21.05 & 21.02 & 21.07 \\
\\

\multicolumn{6}{@{}l}{\textit{Panel C: Village-level}}\\
\vspace{-.6em}\\

\# of MFIs  &       -0.01   &       -0.07   &       -0.05   &       -0.02   &       -0.08   \\
            &      (0.04)   &      (0.06)   &      (0.03)   &      (0.05)   &      (0.08)   \\
\vspace{-.6em}\\
E(dep. var.)  & 0.71 & 1.20 & 0.24 & 0.76 & 1.91 \\
(st. dev.)    & (0.37) & (0.45) & (0.20) & (0.37) & (0.64) \\
\vspace{-.6em}\\
Zeros (\%)  & 1.4\% & 0.7\% & 14.4\% & 1.4\% & 0.0\% \\
\vspace{-.6em}\\
Observations                   &         146   &         146   &         146   &        146   &         146   \\
\vspace{-.6em}\\
Kleibergen--Paap F-stat. & 15.12 & 16.26 & 15.20 & 16.20 & 15.00 \\
\bottomrule
\end{tabular}
    \begin{tablenotes}
\item \emph{Notes:} The table reports linear IV estimates of the effect of MFI competition on household loan counts by credit source. The number of MFIs in a village is instrumented with randomized MFI entry (RDRS). E(dep. var.) refers to the baseline mean of the dependent variable in the estimation sample. All specifications include the same baseline controls as in Table~\ref{tab:nrloans}. Panel C uses baseline village averages of these variables. Robust standard errors clustered at the village level are in parentheses. Kleibergen--Paap rk Wald F-statistics are reported at the bottom of each panel. $^{*}$, $^{**}$, $^{***}$ denote statistical significance at the 10, 5, and 1 percent levels, respectively.
\end{tablenotes}
\end{threeparttable}
\end{table}

\begin{landscape}

\begin{table}[]
\centering
\begin{threeparttable}
\caption{Effects of MFI competition on the extensive margin of borrowing (probability of having at least one loan from a given source)}
\label{tab:extensive}
\footnotesize
\begin{tabular}{@{}lllllllllll@{}}
\toprule
 & \multicolumn{2}{l}{$\geq 1$ MFI loan} & \multicolumn{2}{l}{$\geq 1$ non-MFI loan} & \multicolumn{2}{l}{\begin{tabular}[c]{@{}l@{}}$\geq 1$ moneylender \\ loan\end{tabular}} & \multicolumn{2}{l}{\begin{tabular}[c]{@{}l@{}}$\geq 1$ loan\\ from acquintances\end{tabular}} & \multicolumn{2}{l}{\begin{tabular}[c]{@{}l@{}}$\geq 1$ loan\\ (any source)\end{tabular}} \\ \midrule
 & 2SLS          & IV-probit          & 2SLS            & IV-probit            & 2SLS                                    & IV-probit                                   & 2SLS                                      & IV-probit                                      & 2SLS                                    & IV-probit                                   \\

\midrule
\textit{Panel A: Full sample}  &&&&&&&&&&\\     
\# of MFIs    &       0.012         &       0.034         &      -0.033         &      -0.099         &      -0.032         &      -0.113\sym{*}  &       0.010         &       0.030         &      -0.024         &      -0.102         \\
            &     (0.022)         &     (0.060)         &     (0.023)         &     (0.069)         &     (0.021)         &     (0.067)         &     (0.023)         &     (0.062)         &     (0.017)         &     (0.072)         \\
\vspace{-.6em}\\
Observations        &        2903         &        2903         &        2903         &        2903         &        2903         &        2903         &        2903         &        2903         &        2903         &        2903         \\
\\
\begin{tabular}[c]{@{}l@{}}\textit{Panel B: `Willing to }\\ \textit{borrow' sample}\end{tabular}  &&&&&&&&&&\\   

\# of MFIs    &       0.028         &       0.087         &      -0.056\sym{*}  &      -0.164\sym{*}  &      -0.050\sym{*}  &      -0.193\sym{**} &       0.003         &       0.010         &      -0.039\sym{*}  &      -0.178\sym{*}  \\
            &     (0.025)         &     (0.078)         &     (0.030)         &     (0.086)         &     (0.028)         &     (0.089)         &     (0.029)         &     (0.080)         &     (0.021)         &     (0.092)         \\
\vspace{-.6em}\\
Observations          &        1374         &        1374         &        1374         &        1374         &        1374         &        1374         &        1374         &        1374         &        1374         &        1374         \\
\\
\begin{tabular}[c]{@{}l@{}}\textit{Panel C: Village-level}\\ \textit{averages}\end{tabular}  &&&&&&&&&&\\     
\# of MFIs    &       0.011         &&      -0.023         &&      -0.028         &&       0.030         &&      -0.025\sym{*}  \\
            &     (0.018)         &&     (0.031)         &&     (0.030)         &&     (0.021)         &&     (0.015)         \\
\vspace{-.6em}\\
Observations &         146         &&         146         &&         146         &&         146         &&         146         \\
\bottomrule
\end{tabular}
\begin{tablenotes}[flushleft]
\footnotesize
\item \emph{Notes:} The table reports estimates of the effect of MFI competition on the probability that a household has at least one loan from a given source. Odd-numbered columns report IV linear probability model (2SLS) estimates, and even-numbered columns report IV probit estimates. The number of MFIs in a village is instrumented with randomized MFI entry (RDRS). E(dep. var.) refers to the baseline mean of the dependent variable in the estimation sample. All specifications include the same baseline controls as in Table~\ref{tab:nrloans}. Robust standard errors clustered at the village level are in parentheses. $^{*}$, $^{**}$, $^{***}$ denote statistical significance at the 10, 5, and 1 percent levels, respectively.

\end{tablenotes}

\end{threeparttable}
\end{table}

\begin{table}[htbp]
\centering
\begin{threeparttable}
\caption{Robustness to winsorization: Effect of MFI competition on amounts borrowed}
\label{tab:winsor1}

\footnotesize                           
\setlength{\tabcolsep}{6.75pt}           
\renewcommand{\arraystretch}{1.05}      

\begin{tabular}{l*{15}{c}}              
\toprule
Amount borrowed from:           & \multicolumn{3}{c}{\shortstack{MFIs}} 
            & \multicolumn{3}{c}{\shortstack{Non-MFIs}} 
            & \multicolumn{3}{c}{\shortstack{Moneylenders}} 
            & \multicolumn{3}{c}{\shortstack{Acquaintances}} 
            & \multicolumn{3}{c}{\shortstack{Any source}} \\
\cmidrule(lr){2-4}\cmidrule(lr){5-7}\cmidrule(lr){8-10}\cmidrule(lr){11-13}\cmidrule(lr){14-16}
(1000 Tk.)            &  & wins. & wins. &  & wins. & wins.&  & wins. & wins.&  & wins. & wins.&  & wins. & wins.\\
            & raw & 99$^{\text{th}}$ & 98$^{\text{th}}$& raw & 99$^{\text{th}}$ & 98$^{\text{th}}$& raw & 99$^{\text{th}}$ & 98$^{\text{th}}$& raw & 99$^{\text{th}}$ & 98$^{\text{th}}$& raw & 99$^{\text{th}}$ & 98$^{\text{th}}$\\ 
\midrule
\textit{Panel A: Full sample}  &&&&&&&&&&&&&&&\\     
\# of MFIs  &        0.01   &        0.01   &        0.01   &        0.14   &        0.02   &       -0.02   &       -0.15   &       -0.27*  &       -0.28** &        0.07   &        0.01   &       -0.00   &        0.04   &       -0.00   &       -0.01   \\
            &      (0.07)   &      (0.06)   &      (0.06)   &      (0.09)   &      (0.07)   &      (0.07)   &      (0.20)   &      (0.17)   &      (0.16)   &      (0.10)   &      (0.09)   &      (0.09)   &      (0.05)   &      (0.05)   &      (0.05)   \\
\vspace{-.7em}\\
E(dep. var.)&       12.98   &       12.36   &       12.14   &        8.68   &        7.34   &        6.77   &        1.95   &        1.54   &        1.38   &        4.67   &        3.69   &        3.50   &       21.66   &       20.07   &       19.35   \\
(st. dev.)  &       23.76   &       16.81   &       15.95   &       31.09   &       14.59   &       11.56   &       10.34   &        4.81   &        3.94   &       24.00   &        8.47   &        7.46   &       41.42   &       26.00   &       22.79   \\
\vspace{-.6em}\\
Zeros (\%)  &        46.1   &        46.1   &        46.1   &        32.0   &        32.0   &        32.0   &        81.4   &        81.4   &        81.4   &        51.4   &        51.4   &        51.4   &        18.3   &        18.3   &        18.3   \\
\vspace{-.6em}\\
Observations&         2735   &         2735 &         2735  &         2735  &         2735  &         2735  &         2735  &         2735  &         2735  &         2735  &         2735  &         2735  &         2735  &         2735  &   2735 \\     
\\

\begin{tabular}[c]{@{}l@{}}\textit{Panel B: `Willing to }\\ \textit{borrow' sample}\end{tabular}  &&&&&&&&&&&&&&&\\     
\# of MFIs  &        0.08   &        0.07   &        0.08   &       -0.03   &       -0.06   &       -0.07   &       -0.34** &       -0.38** &       -0.38** &       -0.03   &        0.02   &       -0.01   &        0.01   &        0.01   &        0.01   \\
            &      (0.07)   &      (0.07)   &      (0.07)   &      (0.10)   &      (0.09)   &      (0.09)   &      (0.20)   &      (0.21)   &      (0.20)   &      (0.13)   &      (0.13)   &      (0.13)   &      (0.06)   &      (0.05)   &      (0.05)   \\
\vspace{-.7em}\\
E(dep. var.)&       13.92   &       12.90   &       12.69   &        7.87   &        7.18   &        6.73   &        2.38   &        1.73   &        1.55   &        4.10   &        3.51   &        3.36   &       21.79   &       20.42   &       19.81   \\
(st. dev.)  &       28.14   &       16.70   &       15.86   &       20.55   &       13.65   &       11.12   &       13.31   &        5.04   &        4.15   &       14.34   &        7.98   &        7.16   &       37.18   &       25.09   &       22.33   \\
\vspace{-.6em}\\
Zeros (\%)  &        42.6   &        42.6   &        42.6   &        30.7   &        30.7   &        30.7   &        79.0   &        79.0   &        79.0   &        51.6   &        51.6   &        51.6   &        16.0   &        16.0   &        16.0   \\
\vspace{-.6em}\\
Observations&         1374   &         1374 &         1374  &         1374  &         1374  &         1374  &         1374  &         1374  &         1374  &         1374  &         1374  &         1374  &         1374  &         1374  &         1374  \\
\\
\begin{tabular}[c]{@{}l@{}}\textit{Panel C: Village-level}\\ \textit{averages}\end{tabular}  &&&&&&&&&&&&&&&\\     
\# of MFIs  &       -0.03   &        0.00   &        0.00   &        0.11   &        0.03   &       -0.01   &        0.08   &       -0.19   &       -0.25   &        0.07   &        0.02   &       -0.00   &        0.01   &       -0.01   &       -0.02   \\
            &      (0.08)   &      (0.07)   &      (0.07)   &      (0.08)   &      (0.07)   &      (0.07)   &      (0.23)   &      (0.23)   &      (0.23)   &      (0.10)   &      (0.09)   &      (0.09)   &      (0.07)   &      (0.06)   &      (0.06)   \\
\vspace{-.7em}\\
E(dep. var.)&       12.84   &       12.20   &       11.98   &        8.61   &        7.32   &        6.75   &        1.96   &        1.55   &        1.39   &        4.65   &        3.68   &        3.49   &       21.45   &       19.88   &       19.17   \\
(st. dev.)  &        8.35   &        7.11   &        6.86   &        7.63   &        4.45   &        3.70   &        2.76   &        1.53   &        1.30   &        5.42   &        2.50   &        2.27   &       12.35   &        9.39   &        8.73   \\
\vspace{-.6em}\\
Zeros (\%)  &         1.4   &         1.4   &         1.4   &         0.0   &         0.0   &         0.0   &        14.4   &        14.4   &        14.4   &         0.0   &         0.0   &         0.0   &         0.0   &         0.0   &         0.0   \\
\vspace{-.6em}\\
Observations&         146   &         146   &         146   &         146   &         146   &         146   &         146   &         146   &         146   &         146   &         146   &         146   &         146   &         146   &         146   \\
\bottomrule
\end{tabular}
\begin{tablenotes}[flushleft]
\footnotesize
\item \emph{Notes:} Each column reports IV Poisson (GMM) estimates of the effect of MFI competition on the amount borrowed from the indicated source. The number of MFIs in a village is instrumented with randomized MFI entry (RDRS). Columns labeled “raw” use the untransformed outcome, while columns labeled “wins.” winsorize the dependent variable at the indicated percentile. E(dep. var.) refers to the baseline mean of the dependent variable in the estimation sample. All specifications include the same baseline controls as in Table~\ref{tab:nrloans}. Panel C uses baseline village averages of these variables as controls. Robust standard errors clustered at the village level are in parentheses. Cells report $\exp(\widehat\beta)-1$. $^{*}$, $^{**}$, $^{***}$ denote statistical significance at the 10, 5, and 1 percent levels, respectively.
\end{tablenotes}
\end{threeparttable}
\end{table}

\end{landscape}

\begin{table}[!htbp]
\small
\caption{Linear IV estimates of the effect of MFI competition on total amounts borrowed by credit source \label{tab:amtslinear}}
\centering
  \begin{threeparttable}
\begin{tabular}{l*{5}{>{\centering\arraybackslash}p{2.2cm}}}
\toprule
\multicolumn{1}{l}{} & \multicolumn{1}{c}{(1)}                                          & \multicolumn{1}{c}{(2)}                                             & \multicolumn{1}{c}{(3)}                                                 & \multicolumn{1}{c}{(4)}      & \multicolumn{1}{c}{(5)}                    \\
 & \begin{tabular}[c]{@{}c@{}}Amount \\(1000 Tk) \\borrowed \\ from MFIs \textcolor{white}{iii}\end{tabular} & \begin{tabular}[c]{@{}c@{}}Amount \\(1000 Tk) \\borrowed \\ from non-MFIs\end{tabular} & \begin{tabular}[c]{@{}c@{}}Amount \\(1000 Tk) \\borrowed from\\ moneylenders\end{tabular} & \begin{tabular}[c]{@{}c@{}}Amount \\(1000 Tk) \\borrowed from\\ acquintances\end{tabular} & \begin{tabular}[c]{@{}c@{}}Amount \\(1000 Tk) \\borrowed\\(any source) \end{tabular} \\
\midrule

\multicolumn{6}{@{}l}{\textit{Panel A: Full sample}}\\
\vspace{-.6em}\\
\# of MFIs  &        0.07   &        0.19   &       -0.23   &       -0.08   &        0.08   \\
            &      (0.61)   &      (0.54)   &      (0.21)   &      (0.33)   &      (0.87)   \\
\vspace{-.6em}\\
E(dep. var.)&       12.36   &        7.34   &        1.54   &        3.69   &       20.07   \\
(st. dev.)  &       16.81   &       14.59   &        4.81   &        8.47   &       26.00   \\
\vspace{-.6em}\\
Zeros (\%)  &        46.1   &        32.0   &        81.4   &        51.4   &        18.3   \\
\vspace{-.6em}\\
Observations&        2735   &        2735   &        2735   &        2735   &        2735   \\
\vspace{-.6em}\\
Kleibergen--Paap F-stat.  & 21.15 & 21.09 & 21.04 & 21.07 & 20.86 \\
\\
\multicolumn{6}{@{}l}{\textit{Panel B: `Willing to borrow' sample}}\\
\vspace{-.6em}\\
\# of MFIs  &        0.38   &       -0.22   &       -0.45   &        0.02   &       -0.01   \\
            &      (0.75)   &      (0.66)   &      (0.32)   &      (0.39)   &      (1.06)   \\
\vspace{-.6em}\\
E(dep. var.)&       12.90   &        7.18   &        1.73   &        3.51   &       20.42   \\
(st. dev.)  &       16.70   &       13.65   &        5.04   &        7.98   &       25.09   \\
\vspace{-.6em}\\
Zeros (\%)  &        42.6   &        30.7   &        79.0   &        51.6   &        16.0   \\
\vspace{-.6em}\\
Observations&        1374   &        1374   &        1374   &        1374   &        1374   \\
\vspace{-.6em}\\
Kleibergen--Paap F-stat. & 21.38 & 21.08 & 21.13 & 21.07 & 21.16 \\
\\
\multicolumn{6}{@{}l}{\textit{Panel C: Village-level averages}}\\
\vspace{-.6em}\\
\# of MFIs  &       -0.03   &        0.39   &       -0.20   &        0.01   &        0.04   \\
            &      (0.73)   &      (0.57)   &      (0.23)   &      (0.35)   &      (1.03)   \\
\vspace{-.6em}\\
E(dep. var.)&       12.20   &        7.32   &        1.55   &        3.68   &       19.88   \\
(st. dev.)  &        7.11   &        4.45   &        1.53   &        2.50   &        9.39   \\
\vspace{-.6em}\\
Zeros (\%)  &         1.4   &         0.0   &        14.4   &         0.0   &         0.0   \\
\vspace{-.6em}\\
Observations&         146   &         146   &         146   &         146   &         146   \\
\vspace{-.6em}\\
Kleibergen--Paap F-stat.  & 16.70 & 16.89 & 16.16 & 16.53 & 15.03 \\
\bottomrule
\end{tabular}

\begin{tablenotes}
\small
\item \emph{Notes:} The table reports linear IV (2SLS) estimates of the effect of MFI competition on household borrowing amounts. The number of MFIs in a village is instrumented with randomized MFI entry (RDRS). The dependent variable is the total amount borrowed (in 1,000 Taka) from: (1) MFI credit providers, (2) non-MFIs (including moneylenders and acquaintances), (3) moneylenders, (4) acquaintances, and (5) all sources combined. E(dep. var.) refers to the baseline mean of the dependent variable in the estimation sample. All specifications include the same baseline controls as in Table~\ref{tab:nrloans}. Panel~C uses baseline village averages of these variables as controls. The dependent variable is winsorized at the 99$^{\text{th}}$ percentile, as in Table~\ref{tab:amts}. Robust standard errors clustered at the village level are in parentheses. Kleibergen--Paap rk Wald F-statistics are reported at the bottom of each panel. $^{*}$, $^{**}$, $^{***}$ denote statistical significance at the 10, 5, and 1 percent levels, respectively.
\end{tablenotes}
\end{threeparttable}
\end{table}

\newpage

\begin{figure}
    \centering
    \subfloat{\includegraphics[width=.495\textwidth]{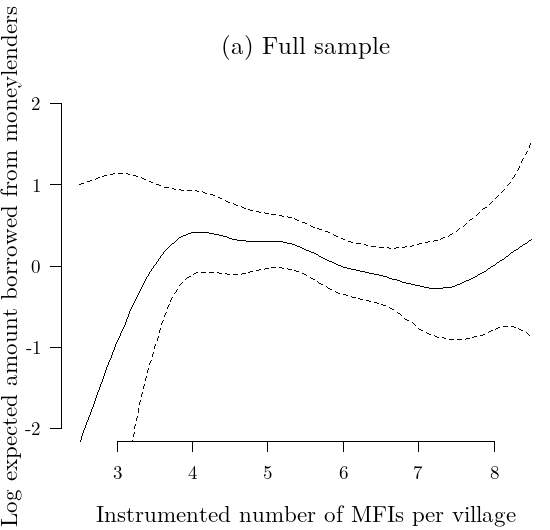}}
    \hspace{-.2cm}
    \subfloat{\includegraphics[width=.495\textwidth]{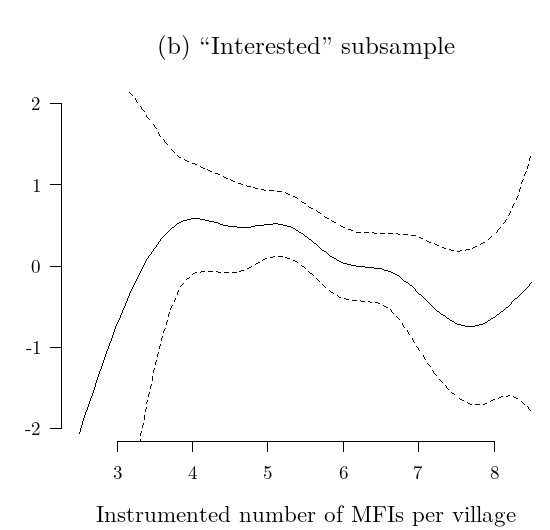}}
    \caption{Estimated smooth effect of MFI competition on amount borrowed from moneylenders in (a) the full sample, and (b) the "interested" in MFI borrowing subsample. The plotted range excludes sparsely populated tails where the semiparametric estimates become imprecise.}
    \label{fig:semip_amts}
\end{figure}

\begin{landscape}

\begin{table}[!htbp]
\footnotesize
\caption{Effects of MFI competition on borrowing outcomes with Romano–Wolf step-down adjusted p-values\label{tab:tableMHT}}
\centering
  \begin{threeparttable}
  \setlength{\tabcolsep}{7pt}
\begin{tabular}{lllllllllllll}
\toprule
\multicolumn{1}{l}{} & \multicolumn{1}{l}{(1)}                                          & \multicolumn{1}{l}{(2)}                                             & \multicolumn{1}{l}{(3)}                                                 & \multicolumn{1}{l}{(4)}      & \multicolumn{1}{l}{(5)}                    && \multicolumn{1}{l}{(6)}                                          & \multicolumn{1}{l}{(7)}                                             & \multicolumn{1}{l}{(8)}                                                 & \multicolumn{1}{l}{(9)} & \multicolumn{1}{l}{(10)} \\
           & \multicolumn{5}{c}{\# of loans from...}                                                                                                                                                       &  & \multicolumn{5}{c}{Amount borrowed (1,000 Tk) from...}                                                                                                                                            \\ \cmidrule(lr){2-6} \cmidrule(l){8-12} 
           & MFIs & non-MFIs & \begin{tabular}[c]{@{}l@{}}money-\textcolor{white}{ab}\\ lenders\end{tabular} & \begin{tabular}[c]{@{}l@{}}acquain-\\ tances\end{tabular} & \begin{tabular}[c]{@{}l@{}}any \\ source\end{tabular} &  & MFIs & non-MFIs & \begin{tabular}[c]{@{}l@{}}money-\textcolor{white}{ab}\\ lenders\end{tabular} & \begin{tabular}[c]{@{}l@{}}acquain-\\ tances\end{tabular} & \begin{tabular}[c]{@{}l@{}}any \\ source\end{tabular} \\ 

\midrule
\multicolumn{6}{@{}l}{\textit{Panel A: Full sample}}\\
\vspace{-.7em}\\
\# of MFIs \textcolor{white}{abcd} &        0.01   &       -0.08  &       -0.31&       -0.02   &       -0.06      &&        0.02   &        0.03   &       -0.22  &        0.02   &        0.01   \\
            &   (0.8837)   &   (0.0754)   &   (0.0035)   &   (0.7595)   &   (0.1078)   &&   (0.8744)   &   (0.7743)   &   (0.0625)   &   (0.8762)   &   (0.9736)   \\
            &   [0.9989]   &   [0.3027]   &   [0.0385]   &   [0.9979]   &   [0.3872]   &&   [0.9989]   &   [0.9979]   &   [0.3027]   &   [0.9989]   &   [0.9989]   \\
\vspace{-.7em}\\
E(dep. var.)&        0.72   &        1.20   &        0.24   &        0.76   &        1.92   && 12.36 & 7.34 & 1.54 & 3.69 & 20.07 \\
(st. dev.)  &        0.80   &        1.16   &        0.57   &        0.99   &        1.53  && (16.81) & (14.59) & (4.81) & (8.47) & (26.00) \\
\vspace{-.6em}\\
Zeros (\%)      &        46.1   &        32.0   &        81.4   &        51.4   &        18.3   && 46.1\% & 32.0\% & 81.4\% & 51.4\% & 18.3\% \\
\vspace{-.6em}\\
Observations&        2735   &        2735   &        2735   &        2735   &        2735      &&        2735   &        2735   &        2735   &        2735  &        2735 \\
\\
\multicolumn{6}{@{}l}{\textit{Panel B: `Willing to borrow' sample}}\\
\vspace{-.7em}\\
\# of MFIs  &        0.06   &       -0.10  &       -0.44&       -0.06   &       -0.06  &&        0.07   &       -0.06   &       -0.35 &       -0.01   &        0.01   \\
            &   (0.3373)   &   (0.0533)   &   (0.0088)   &   (0.5522)   &   (0.1474)   &&   (0.3052)   &   (0.5039)   &   (0.0226)   &   (0.8987)   &   (0.8949)   \\
            &   [0.8231]   &   [0.2670]   &   [0.0923]   &   [0.9231]   &   [0.5374]   &&   [0.8044]   &   [0.9231]   &   [0.1747]   &   [0.9868]   &   [0.9868]   \\
\vspace{-.6em}\\
E(dep. var.)  &   12.90 &   7.18   &   1.73    &   3.51   &  20.42    &&   12.90 &   7.18   &   1.73    &   3.51   &  20.42   \\
(st. dev.)    & (16.70) &   (13.65) &   (5.04) &   (7.98) &   (25.09)   && (16.70) &   (13.65) &   (5.04) &   (7.98) &   (25.09) \\
[0.2em]
\vspace{-.6em}\\
Zeros (\%)      &  42.6\% &   30.7\%  &   79.0\%  &   51.6\%  &   0.00\%   &&  42.6\% &   30.7\%  &   79.0\%  &   51.6\%  &   0.00\%  \\
\vspace{-.6em}\\
Observations  &        1374   &        1374   &        1374   &        1374 &        1374 &&        1374   &        1374   &        1374   &        1374 &        1374   \\
\bottomrule
\end{tabular}
\begin{tablenotes}
\small
\item \emph{Notes:} The table reports estimates of the effect of MFI competition on borrowing outcomes, together with step-down adjusted p-values that are robust to multiple hypothesis testing using the Romano--Wolf procedure \citep{romano2005exact,romano2005stepwise}. The number of MFIs in a village is instrumented with randomized MFI entry (RDRS). For each coefficient, model-based p-values computed using robust standard errors clustered at the village level are shown in parentheses, and the corresponding Romano--Wolf step-down adjusted p-values are shown in square brackets. Columns (1)--(5) report effects on the number of loans, and columns (6)--(10) report effects on amounts borrowed (in 1,000 Taka). E(dep. var.) refers to the baseline mean of the dependent variable in the estimation sample. All specifications include the same baseline controls as in Table~\ref{tab:nrloans}. Amount outcomes are winsorized at the 99$^{\text{th}}$ percentile, as in Table~\ref{tab:amts}. Statistical significance indicators are omitted.
\end{tablenotes}
\end{threeparttable}
\end{table}

\end{landscape}

\begin{table}[!htbp]
\caption{Suggestive IV estimates of the effect of MFI competition on moneylender interest rates}
\label{tab:interest}
\centering
  \begin{threeparttable}
\footnotesize
\setlength{\tabcolsep}{10pt}
\begin{tabular}{lllllll}
\toprule
\multicolumn{1}{l}{} & \multicolumn{1}{l}{(1)} & \multicolumn{1}{l}{(2)}  & \multicolumn{1}{l}{(3)}  && \multicolumn{1}{l}{(4)}  & \multicolumn{1}{l}{(5)} \\

 & \multicolumn{3}{c}{Moneylender interest rate (levels)} && 
   \multicolumn{2}{c}{log moneylender interest rate} \\
 \cline{2-4} \cline{6-7}
 & raw \textcolor{white}{abcd} & wins. 99$^{\text{th}}$ & raw && raw \textcolor{white}{abcde}& raw \\
    
\midrule
\multicolumn{6}{@{}l}{\textit{Panel A: Full sample}}\\
\vspace{-.6em}\\
\# of MFIs  &      -22.51*  &      -14.62   &      -11.60   &&       -0.08   &       -0.02   \\
            &     (13.61)   &     (10.40)   &     (10.67)   &&      (0.07)   &      (0.07)   \\
\vspace{-.6em}\\
E(dep. var.)&      135.57   &      131.89   &      131.27   &&        4.65   &        4.66   \\
(st. dev.)  &      134.45   &      104.59   &      102.17   &&        0.71   &        0.68   \\
\vspace{-.6em}\\        
Observations&         503   &         503   &         474   &&         503   &         451   \\
Baseline mean(Y) as control &               &               &   \checkmark    &&               &   \checkmark  \\  
\\
\multicolumn{6}{@{}l}{\textit{Panel B: `Willing to borrow' sample}}\\
\vspace{-.6em}\\
\# of MFIs  &       -5.21   &       -0.75   &        4.13   &&        0.04   &        0.12   \\
            &     (12.46)   &     (11.28)   &     (11.21)   & &     (0.07)   &      (0.07)   \\
\vspace{-.6em}\\
E(dep. var.)&      137.38   &      135.11   &      134.54   &  &      4.67   &        4.68   \\
(st. dev.)  &      124.78   &      109.43   &      107.84   &   &     0.70   &        0.69   \\
\vspace{-.6em}\\        
Observations&         285   &         285   &         270   &   &      285   &         257   \\
Baseline mean(Y) as control &               &               &   \checkmark    &&               &   \checkmark  \\  
\\
\multicolumn{6}{@{}l}{\textit{Panel C: Village-level averages}}\\
\vspace{-.6em}\\
\# of MFIs  &      -63.37*  &      -35.41** &      -15.76   &&       -0.17*  &       -0.05   \\
            &     (35.98)   &     (17.13)   &     (17.77)   & &     (0.11)   &      (0.08)   \\
\vspace{-.6em}\\
E(dep. var.)&      147.70   &      136.63   &      135.39   &  &      4.68   &        4.69   \\
(st. dev.)  &      164.01   &       75.19   &       62.93   &   &     0.50   &        0.41   \\
\vspace{-.6em}\\        
Observations&         123   &         123   &         107   &    &     123   &         107   \\
Baseline mean(Y) as control &               &               &   \checkmark    & &              &   \checkmark  \\  
\bottomrule
\end{tabular}
\begin{tablenotes}
\footnotesize
\item \emph{Notes:} The table reports IV (2SLS) estimates of the effect of MFI competition on moneylender interest rates. 
The number of MFIs operating in a village is instrumented with randomized MFI entry (RDRS). 
The dependent variable is the moneylender interest rate in levels (columns (1)--(3)) or its logarithm (columns (4)--(5)). 
Columns (4)--(5) report exponentiated coefficients, $\exp(\beta)-1$, which can be interpreted as semi-elasticities.

In Panels A and B, the unit of observation is a household that reported at least one moneylender loan in the year preceding the endline survey and the associated interest rate; 
in Panel C, the unit of observation is a village with at least one reported moneylender loan. 
Because interest-rate information is observed for only a limited and selective subset of households 
(approximately 18.6\% of endline moneylender borrowers), columns (1) and (2) omit the baseline realization of the outcome to preserve sample size, 
while columns (3) and (5) include baseline covariates at the village level. E(dep. var.) refers to the baseline mean of the dependent variable in the estimation sample.

All specifications include the same baseline controls as in Table~\ref{tab:nrloans}, where applicable. 
The dependent variable in column (2) is winsorized at the 99$^{\text{th}}$ percentile to limit the influence of outliers. 
Robust standard errors clustered at the village level are reported in parentheses. 
$^{*}$, $^{**}$, and $^{***}$ denote statistical significance at the 10, 5, and 1 percent levels, respectively.
\end{tablenotes}
\end{threeparttable}
\end{table}


 \end{document}

  {\bf Many missing, many redundant, many incomplete.}

 Aleem (1990)
 
 Amin et al. (2003)
 
 Baquero, Hamadi, and Heinin (2018)
 
 Barman, D., Mathur, H. \& V. Kalra (2009) Role of Microfinance Interventions in Financial Inclusion: A Comparative Study of Microfinance Models, Journal of Business Perspective, 13(3): 51-59.
 
  Berg et al. (2013)

 Bose (1998)
 
Bulow and Rogoff (1989)
 
 de Janvry, A., McIntosh, C., Sadoulet, E., 2010. Journal of Development Economics, 93, 173-188.
 
 Demont (2012)
 
 de Quidt, Fetzer, and Ghatak (2018)
 
 Fernando, N.A. 2006. Understanding and Dealing with High Interest Rates on Microcredit: A Note to Policy Makers in Asia and the Pacific Region. Asian Development Bank.
 
 Gertler and Gruber (2002)

  Gertler et al. (2009)

   Ghosh and Ray (2016)
 
 Heckman, James J., Hidehiko Ichimura, Petra E. Todd. Matching as an Econometric Evaluation Estimator: Evidence from Evaluating a Job Training Programme. Review of Economic Studies 64: 605-654.
 
 Hoff, K., Stiglitz, J.E., 1998. Money-lenders and Bankers: Price-Increasing Subsidies in a Monopolistically Competitive Market. Journal of Development Economics 55: 485-518.
 
 Islam, Asadul. 2011. Medium- and Long-term Participation in Microcredit: An Evaluation Using a New Panel Dataset from Bangladesh. American Journal of Agricultural Economics 93(3): 847-866.
 
    Islam and Maitra (2012)
 
   Islam, Nguyen and Smyth (2014)
   
   Iqbal (1988)
 
 Jain, S. and G. Mansuri. 2003. A Little at a Time: The Use of Regularly Scheduled Repayments in Microfinance Programs, Journal of Development Economics 72(2): 253-279
 
  Khandker (2000)
 
  Kocher (1997)
 
 Mallick, Debdulal. 2012. Microfinance and Moneylender Interest Rate: Evidence from Bangladesh. World Development 40 (6): 1181-1189.
 
 Menkhoff and  Rungruxsirivorn (2011)
 
 Marconi, R., Mosley, P., 2005. Bolivia during the global crisis 1998-2004: towards a macroeconomics of microfinance. Sheffield Economic Research Paper Series 2005-007.
 
 McIntosh, C., Wydick, B., 2005a. Competition and microfinance. Journal of Development Economics 78, 271-298.
 
 McIntosh, C., de Janvry, A., Sadoulet, E., 2005b. How increasing competition among microfinance institutions affects incumbent lenders. Economic Journal 115, 987-1004.

 McKernan, S.-M., Pitt, M. M., \& Moskowitz, D. 2005. Use of the formal and informal financial sectors: does gender matter? Empirical evidence from rural Bangladesh. World Bank Policy Research Paper 3491, World Bank, Washington, DC.
 
 Meyer, R. L. 2002. The demand for flexible microfinance products: Lessons from Bangladesh. Journal of International Development, 14(3), 351-368.
 
 Mookherjee, D., A. Motta. 2016: A Theory of Interactions Between MFIs and Informal Lenders. Journal of Development Economics, 121: 191-200.  
 
 Morduch, J. 1998. Does microfinance really help the poor? New evidence from flagship programs in Bangladesh. New York University.

 Munshi and Rosenzweig (2009)
 
 Navajas, S., Conning, J., Gonzalez-Vega, C., 2003. Lending Technologies, Competition and Consolidation in the Market for Microfinance in Bolivia, Journal of International Development 15: 747-770.
 
 Peterson and Rajan (1995)
 
 Pitt, Mark M., Shahidur R. Khandker. 1998. The Impact of Group-Based Credit Programs on Poor Households in Bangladesh: Does the Gender of Participants Matter. Journal of Political Economy, 106 (5): 958-996.
 
 Porteous, D., 2006. Competition and Microcredit interest Rates. CGAP Focus Note No. 33. February 2006.
 
 Ravallion, Martin, and Shaohua Chen. 2005. Hidden Impact: Household Saving in Response to a Poor-Area Development Project, Journal of Public Economics, 89: 2183-2204.
 
 Rodman, D., and J. Morduch. 2014. The impact of microcredit on the poor in Bangladesh: revisiting the evidence. Journal of Development Studies. 50(4): 583-604.
 
 Rosenbaum, Paul R. and Donald B. Rubin. 1983. The Cenral Role of the Propensity Score in Observational Studies for Causal Effects. Biometrika 70 (1): 41-55.
 
 Salim, Mir M. 2013. Revealed Objective Functions of Microfinance Institutions: Evidence from Bangladesh. Journal of Development Economics 104: 34-55..

 Sinha, S., \& Matin, I. (1998). Informal credit transactions of microcredit borrowers in rural Bangladesh. IDS Bulletin, 29(4), 66-80.
 
 Tedeschi, Gwendolyn Alexander; Karlan, Dean. 2010. Cross-Sectional Impact Analysis: Bias from Dropouts. Perspectives on Global Development and Technology, 9 (3-4): 270-291(22)
 
 United Nations Interagency Committee on Integrated Rural Development for Asia and the Pacific, 1992. Partners in Rural Poverty Alleviation: MFI Cooperation. United Nations,
New York.
 
 van Tassel, E., 2002. Signal Jamming in New Credit Markets. Journal of Money, Credit and Banking 34: 469-490.
 
 Vogelgesang, Ulrike. 2003. Microfinance in the Times of Crisis: The Effects of Competition, Rising Indebtedness, and Economic Crisis on Repayment Behavior. World Development 31(12): 2085-2114.
 
 Zeller, M., Sharma, M., Ahmed, A. U., \& Rashid, S. (2001). Group-based financial institutions for the rural poor in Bangladesh: An institutional and household-level analysis. Research Report No. 120, IFPRI, Washington, DC.

Amin et al. (2003) argue that in Bangladesh households that join microfinance have better access to insurance and smoothing devices.
 find that in Bangladesh, households that have access to microfinance are less likely to sale livestock in response to health shocks.\footnote{In the context of India however, Munshi and Rosenzweig (2009)  argue that within caste loans are more important for self-insurance, rather than micro-credit.   } In the context of Indonesia, while
 \noindent $\square$ \textbf{Co-existence of MFIs and moneylenders.}  Jain and Mansuri (2003) suggest that the MFIs use  repayment schemes to ensure that borrowers are forced to take bridge loans from the moneylenders, thereby allowing the MFIs to tap into the greater information that moneylenders have regarding the borrowers.   On the other hand
 
 \newpage

 \section{Appendix 3: Profit-maximizing MFIs}\label{As3}

The case of interest is thus when the MFIs are extremely profit-oriented. For ease of exposition we in fact focus on the case where    $\lambda=0$.      Let us denote the equilibrium by $(\widehat k(n),\widehat r(n))$. Consider the problem where the $i$-th MFI  maximizes $k_ir_i$ subject to the incentive compatibility conditions for the case when there is a consumption shock and $x=n-1$.\footnote{Note that from Proposition \ref{profit} later, we find that the equilibrium loan amount is less than $k^E$, so that we can argue as earlier that satisfying  (\ref{20}) and   (\ref{21}) ensures that all the DICs are satisfied. Further, note that at the equilibrium,  the limited liability constraint cannot bind because in that case the incentive compatibility constraints  would be violated. Hence in what follows we ignore (\ref{LL}). This follows because if (\ref{LL}) binds then the LHS of both  (\ref{20}) and (\ref{21}) are negative. Hence (\ref{20}) and (\ref{21}) are violated.} Hence  the MFI's problem can be written as:
 \footnotesize{\begin{eqnarray}\label{problem}
&& \max_{k_i,r_i}\; k_i(1+r_i),\nonumber\\
&& \textrm{subject to},\nonumber\\
 && \frac{\delta}{1-\delta}\frac{\{f(k_i+\frac{n-1}{n}\widehat k)-k_i(1+r_i)- \frac{(n-1)\widehat\widehat k}{n}(1+\widehat r)\}-\{f(\frac{n-1}{n}\widehat k)- \frac{(n-1)\widehat k)}{n}(1+\widehat r)\}}    {\{k_i(1+r_i)+\frac{n-1}{n}\widehat k(1+\widehat r)\}-\frac{(n-1)\widehat k}{n}(1+\widehat r)}
  \ge  \frac{m(1+c)}{p+(1-p)m(1+c)},\label{20}\nonumber\\\\
{\textrm and,}    && \frac{\delta}{1-\delta}\frac{\{f(k_i+\frac{n-1}{n}\widehat k)-k_i(1+r_i)- \frac{(n-1)\widehat k}{n}(1+\widehat r)\}-\{f(\frac{n-1}{n}\widehat k)- \frac{(n-1)\widehat k(1+\widehat r))}{n}\}}    {\widehat k(1+\widehat r)-\frac{(n-1)\widehat k}{n}(1+\widehat r)}   \ge  \frac{m(1+c)}{p+(1-p)m(1+c)}.\label{21}\nonumber\\\end{eqnarray}}\normalsize

We then argue that there exists an equilibrium where the aggregate loan amount is the efficient one, i.e. $\underline k$, and the interest rate is given by
\begin{equation}
\widehat r=  \frac{Z}{1+Z}\frac{f(\widehat k)-f(\frac{(n-1)\widehat k}{n})}{\widehat k}-1,\label{mission}\end{equation}
where $Z=\frac{\delta [p+(1-p)m(1+c)]}{(1-\delta)m(1+c)},$   where $\underline k$ solves $f'(k)=0$.

Sketch of proof: At the equilibrium, the slope of (\ref{20}) and (\ref{21}) in the $(k_i,r_i)$ plane  must be identical at the solution value, otherwise $k_ir_i$ can be increased.
Thus
$$\frac{dk_i}{dr_i}|_{(\ref{20})}=\frac{(Z+1)k_i}{Zf'(.)-(Z+1)(1+r_i)},$$ and
$$\frac{dk_i}{dr_i}|_{(\ref{21})}=\frac{Zk_i}{Zf'(.)-Z(1+r_i)},$$
Equating the two,  $f'(\widehat k)=0$.

Implication of a change in $n$: As $n$ increases $\widehat r$ decreases. Thus the borrower has more disposable income, and thus takes less loans from the moneylenders, thus the moneylender rate of interest increases. All the implications are thus reversed.

\end{document}
 
 \section{Introduction}


Micro-credit offered by microfinance institutions (MFIs) has emerged as one of the most significant source of credit, especially for poor households living in rural areas in developing countries who otherwise would not have access to formal sector credit. According to the State of the Micro-credit Summit Campaign Report 2013 (http://stateofthecampaign.org/the-report/executive-summary/) by 2011, there were 195 million micro-credit borrowers. In Bangladesh, based on the nationally representative Household Income and Expenditure Survey (HIES) 2010  conducted by the  of Bangladesh Bureau of Statistics, by 2010, MFI loan constituted more than 50\% of total credit for rural households, and for households with smaller landholding, such as 0.5 acre or less, it constituted approximately 54\% of total credit. \footnote{The second most important source of credit for poor household was the informal sector interest bearing loans that constituted approximately 24\% of total credit in 2010.} Furthermore, approximately 33\% of all rural households took MFI loans, and the average number of loans in 2010 stood at 1.2 loans per household.

The increased reach of MFIs to rural unbanked borrowers has accompanied with competition among MFIs. The objective of our study is to understand the effects of competition among MFIs on borrowers' MFI memberships, loan size and loan amount offered by MFIs, and loans received from informal sector and its interest rates, and borrower's debt burden in rural areas of Bangladesh. In addition, our focus is to understand the interactions between informal lenders that traditionally dominated rural sector and MFIs that newly emerged as a major source of credit, and how their interactions affected loans received from informal sector and its interest rates for rural borrowers with limited access to formal sector credit.

Competition among MFIs in micro-credit market and its consequences on MFIs and borrowers is a growing concern for policy makers and development practitioners in Bangladesh and many other developing countries (PKSF 2006, Vogelsang 2003, McIntosh et al. 2005, de Quidt et al. 2013). In Bangladesh and elsewhere, MFIs have achieved phenomenal success in expanding credit and reaching the poor (Gonzalez 2010). Encouraged by the success, recent years have witnessed a large influx of competing MFIs in the rural sector. With borrowers having access to multiple MFIs, the restriction of 'one borrower from one household' is breaking down and in 2005, the overlapping/double-dipping rate (the proportion of borrowers with multiple lending sources) stood at around 40\% (PKSF, 2006). \textcolor{red}{This contradicts \citet{mahmud2022microfinance}, who find that over time the double dipping rate goes up to 23\% in 2014.} Such a process has purportedly worsened MFI performance and has emerged as an important policy concern in Bangladesh and other matured micro-credit markets such as India (Srinivasan 2010), Bolivia (Vogelsang 2003) and elsewhere (McIntosh et al. 2005).

The rural credit market in Bangladesh has historically been served by a range of informal providers such as landlords, local moneylenders, shopkeepers, friends and relatives. During the 1970s, the government expanded its state-owned banks' branches to rural areas and until 1982, acted as the major source of credit in the rural areas. However, the financial liberalization that had started in 1980s resulted in either the denationalization of state-owned banks or closure of their rural branches (Khalily and Meyer, 1993). However, the void created by the retreat of state-owned banks has not been fulfilled by private commercial banks. Instead, microfinance institutions (MFIs) supported by national governments and donor agencies have rapidly expanded credit to a segment of (poor) rural households.


Understanding how MFI competition may affect lenders and borrowers is perhaps nowhere more important than in Bangladesh. Figure 1 depicts the growth of MFIs and their borrowers from 1999 to 2014 in Bangladesh. Based on Credit and Development Forum (CDF) that collates annual data on MFIs and microfinance(MF) MFIs, while the MFIs and MF-MFIs grew at an annual rate of 0.4\%, the number of borrowers borrowing from them increased at a rate of 16.7\%. Similarly, based on Mix Market Data\footnote {http://www.mixmarket.org/mfi/country/Bangladesh accessed in July 2014}, 82 MFIs listed there had a total of US\$3.9 billion outstanding loans disbursed among more than 16.5 million micro-credit borrowers \footnote{Borrower numbers for Grameen Bank â€“ the largest micro lender in the country, and few others were missing in the mix market data.} (Taka 18,900 per borrower). The five largest MFIs that cover most districts had an outstanding loan of US\$100 million or more, and each of the four largest MFIs had more than a million borrowers. While the high number of MFIs in aggregate suggests that microfinance is very a competitive industry at national level, at local level the actual number of MFIs vary substantially. Moreover, at village level, that matters the most for borrowers access to MFI loans \footnote{MFIs do not offer credit to borrowers from outside villages}, even in 2012, there were villages that still remained outside the operation of MFIs (PKSF Annual Report 2013).

The literature on MFI competition is relatively nascent. However, the general presumption here is that competition among MFIs weakens repayment incentives since borrowers having access to multiple lenders reduces any given MFIâ€™s ability to discipline a borrower (Hoff and Stiglitz 1998). Competition also creates problem due to limited information sharing among MFIs, absence of borrowerâ€™s credit history and costly verification. Indeed, focus group discussions with MFIsâ€™ credit officers in Bangladesh suggest that misreporting and manipulation from borrowers have been on the rise and presumably increased competition among MFIs has led to lower repayment rate. This point has also been argued in McIntosh et al. (2005). However, the loan repayment rate of MFIs in Bangladesh at aggregate level does not show any apparent downward trend in recent years marked by increased competition.  Based on two rounds of survey data conducted on MFI members and non-members collected from villages characterized by single and multiple MFIs, we find that as opposed to conventional wisdom, competition of MFIs is associated with fewer numbers of MFI memberships, smaller loan size and smaller amount of MFI loans. In addition, we also find that while the role of informal sector credit and its interest rate have reduced significantly between the survey periods, however, the difference in their roles between monopoly and competitive MFI market is not evident.

This paper seeks to understand the consquences of competition among MFIs on bor- rowers and informal moneylenders. We develop a theoretical model which studies the effects of increase in the number of the MFIs on aggregate MFI loan size, informal moneylender loan size and interest rate. We also study the effects of increase in MFI competition on the welfare of the borrowers as well as on the profit of the moneylenders. We find that the aggregate loan size of MFI decreases with increase in the number of the MFIs. The moneylender?s loan amount decreases and interest rate increases with increase in competition among MFIs. The borrower?s welfare decreases and profit of the moneylender increases as a consequence of increase in the number of MFIs. These predictions of our theoretical model are supported by empirical findings, in particular an empirical study conducted by Chowdhury et al. (2017) find similar results.

We develop a model similar to Jain and Mansuri (2003) where a moneylender pro- vides bridge-loans which are used to repay MFI loan(s). However in their model bridge- loans are required because of frequent repayment schedule of the MFIs, another widely used mechanism by the MFIs. In our framework, the bridge-loan is required because the borrower invests in a risky project, outcome of which is observed by the moneylender but not by the MFI(s), so in case of project failure the moneylender provides bridge-loan to the poor borrower which enables her to repay MFI loan(s). In case the borrower does not repay a particular MFI, that MFI terminates the contract and renegotiation is not allowed.

Due to this structure, the contracts with any MFI will be terminated with positive probability if the moneylender does not provide bridge-loans. Hence, the moneylender plays a very important role in our model. Though the MFIs do not provide contracts to the moneylender directly, but they know the moneylender?s strategy and design contracts accordingly.
We assume that the MFIs do not share information, so if a borrower defaults on a particular MFI loan, only that MFI terminates its contract. So as the number of MFIs increases, the effect of punishment for default on one MFI loan decreases, this increases the incentive for default. Further implication of which is that the maximum aggregate MFI loan amount which is incentive compatible decreases with increase in the number of MFIs. So at the optimum, we observe decrease in aggregate MFI loan size with increase in the number of MFIs.

Since the job of the moneylender, in our framework, is to provide bridge-loans above discussion further implies that moneylender loan size decreases with increase in the number of MFIs. On the other hand, we find that the moneylender interest rate increases with increase in MFI competition. Intuitively this can be explained in two ways - first, as discussed above as the number of MFIs increases, effects of threat of terminating credit contracts decreases, or in other words moneylender?s bargaining power increases with the increase in the number of MFIs which results in increase in moneylender?s interest rate. Second, we get this result due to our assumption of concave production function - since the moneylender is profit maximizer, the moneylender extracts the entire borrower surplus, whenever that is possible. This implies that the difference in moneylender interest rate, as a result of increase in the number of MFIs, is exactly equal to the difference in the proportion of return from total investment to total investment. Now due to our assumption of concave production function, the proportion of return to investment increases with decrease in investment. Hence the result.

We find that the borrower?s welfare decreases with the increase in the number of MFIs and profit of the moneylender increases in that case. This is due to the fact that the borrower?s present discounted value of lifetime utility is an increasing function of the loan size which decreases with MFI competition. On the other hand, moneylender?s profit increases because increase in profit due to the increase in interest rate is dominated by the decrease in profit due to the decrease in loan size which happens due to increase in the number of MFIs.
Finally we find that when the agents (we assume that the moneylender and the borrower have common discount factor) are very patient, there is over borrowing - this happens because the objective of the MFI is to increase borrower?s present discounted value of lifetime utility which is an increasing function of investment till it reaches the effcient amount (without any risk) and the risk of the project is borne by the

\section{Literature on MFI Competition}

\subsection{Theoretical Models and their Testable Implications}

Theoretical papers that examined the mechanics of competition among MFIs on borrowers' outcomes are relatively new. Three such papers that make testable predictions are Hoff and Stiglitz (1998), Jain and Mansuri (2003), and McIntosh and Wydick (2005).

In Hoff and Stiglitz (1998), an increase in the number of lenders adversely affects borrowers' incentives to repay. If enforcement is done through collection effort and reputation effects of borrowers, the later (reputation effects) gets weakened with increased in (the number of) lenders. This in turn increases lenders enforcement costs, and may result in higher interest rates. Hoff and Stiglitz provide testable predictions; with increased competition among MFIs (uncollateralized lenders) the interest rates and default on loans will go up. However, the first outcome cannot be directly tested since interest rate is more of an administratively set and it is fairly constant across MFIs in Bangladesh (Salim 2013). In the case of the second outcome, the default rate of MFI loans is close to zero. In the specific case of our data, the total number of default loans among the sample households stood at 10 and 4 in year 1998 and 2010, respectively, which represented less than 1\% of the total number of MFI loans. Hoff and Stiglitz (1998) predictions may be relevant if we think, given the regulatory constraints, MFIs adjust their loan size at the margin. In this case, competition may reduce loan size, which, as we will see, is supported by our empirical findings.

In McIntosh and Wydick (2005), while competition among MFIs may make the wealthier and impatient borrowers better off by offering them multiple loans, it makes the poorer and more patient borrowers worse off. The entry of a new MFI reduces the ability of the incumbent MFI to generate rents from the profitable borrowers and to cross-subsidize the poorest borrowers. Similar to Hoff and Stiglitz (1998), competition in McIntosh and Wydick (2005) reduces information sharing among lenders about borrowers which leads impatient borrowers to take multiple loans, which in turn, reduces repayment rate, and makes patient borrowers worse off. As said, repayment rates remain very high, almost perfect, in Bangladesh. However, in our empirical analysis, we will see if borrowers' heterogeneity matters or not. In particular, if a poor (land ownership less than 0.5 acre), female-headed, and older (household head is above 50 years old) borrowers are worse off in competition.

Jain and Mansuri (2003) is perhaps the first paper that gazed at the competition in rural credit market between incumbent informal lenders and an entrant MFI. However instead of competition, they focused on one specific but very important feature of micro-credit contracts, weekly/regular repayment schedule often imposed on clients, and showed how this particular feature may complement MFIs and informal lenders. In their model, the repayment schedule that the entrant MFI offers to its client requires the client to borrow from an informal lender, who (informal lender) ensures that the client does not engage in ex-post moral hazardous behavior. As such, Jain and Mansuri predict positive complementarities between informal lenders and MFIs, and the possibility of crowding-in of informal lenders and increased interest rate in the informal sector. Though Jain and Mansuri did not discuss competition among MFIs, one can conjecture that with competition, borrowers may not need to borrow from informal lenders, and it seems our data supports this conjecture. While borrowing from informal sector went up for MFI members in monopoly MFI market, MFI members in competitive MFI market borrowed less from the informal sector.

The testable implications of competition among MFIs derived from above three theoretical papers are summarized in Table 1.

Competition among MFIs can lead borrowers to borrow from more than one MFIs, which is termed as double dipping in McIntosh and Wydick (2005). In PKSF data, information on the total number of loans taken from MFIs and the total number of MFI memberships are available, and we have used them to test this prediction. Related to double dipping is the debt burden on borrowers often portrayed as â€œover-indebtednessâ€ due to MFI competition (Vogelgesang 2003). From PKSF data, we have constructed debt-ratio (total outstanding credit over total income) and an indicator if the client had to take a loan to repay another loan. The effect of competition on repayment of MFI loans is predicated to be negative. However, as is well known in the microfinance literature, repayment rates have always been very high, and in the particular case of Bangladesh, the repayment rates remain as high as 99\%, and hence this cannot be tested using the PKSF data. Similarly, in Bangladesh, interest rate is regulated and the variations in interest rate charged by MFIs are almost non-existent, and hence cannot be tested.

How competition reshapes the average loan size that MFIs offer is not discussed in any of the above three papers, and it seems that they all assume that MFIs do not adjust this important instrument. However, in PKSF data, it can be seen that the loan size varies over time and across space, and we will associate this change with a change in the market structure.

[Two kinds of dropped out borrowers are: first, leaving microfinance industry altogether and, second, dropping out from a particular MFI and joining another MFI. For the first one, we need to compare borrowers who dropped out after 1998 with those who continued. How does competition affected them? This group is already identified and we have included them in one of comparison groups. However, there are issues with the second kind of dropouts; should we look at household level or at individual borrower level? In addition, recall bias is a major issue here. Borrowers who switched MFIs cannot recall the reason and date of such switches.]

\subsection{Empirical Literature}
\textbf{On multiple loan-taking/double-dipping:} Multiple loan-taking has been associated with increased default rates in Vogelgesang (2003) and Marconi and Mosley (2005) in Bolivia during a period of economic crisis. However, Vogelgesang (2003) also found positive effect of credit market competition on repayment behavior of microfinance clients.  McIntosh et al. (2005) in Uganda found the evidence of double dipping through an increase in default rate at group level.  

\textbf{On over-indebtedness:} Vogelgesang (2003) associated increased indebtedness of microfinance clients in Bolivia with increased competition among MFIs.

\textbf{Loan size offered by MFI:} McIntosh et al. (2005) examined the effect of competition on loan size and found no effect. Again, the outcome is measured at group level, not at individual borrower level.

\textbf{On dropout from MFIs:} McIntosh et al. (2005) found that increased competition among lenders in rural credit market significantly increases the likelihood of dropout of clients from the incumbent lenders. However, lenders considered in McIntosh et al. (2005) were heterogeneous and often offered differentiated products (e.g., repayment terms were different).  

\textbf{On interest rate:} Fernando (2006) found that increased competition among MFIs led to a reduction in interest rates charged by MFIs over 2003-2006 in Cambodia. Porteus (2006) found a similar effect of competition on interest rates over the 1990s in Bolivia. However, Vogelgesang (2003) in the same context found the opposite effect â€“ with increased competition clients had to pay high interest rates in Bolivia.

\textbf{On informal lending:} Lending from MFI may work as a substitute to other lending sources and hence may reduce the borrowing from other sources including moneylenders. In fact, this was the basis for expansion of state-owned banks in many developing countries in the 1960s and 1970s, and this has been the stated objective of many MFIs including Grameen Bank.\footnote{'Freeing' poor people from the clutches of 'evil moneylenders' is one of the important justifications put forward by MFIs (Meyer, 2002).}

Available empirical evidence relies heavily on observed associations. For example, Sinha and Matin (1998), Zeller et al. (2001), McKernan et al. (2005), and Barman et al. (2009)  suggest that households in Bangladesh that borrow from MFIs often borrow from informal lenders as well. Furthermore, Mallick (2012) finds that moneylender interest rates are higher in the villages where higher percentage of households borrows from MFIs. This supports Jain and Mansuri (2003) that due to frequent, often weekly, repayment requirement, MFI lending generates demand for other lending, especially informal lending, hence works as a complement. Alternatively, MFI lending may make borrowers previously unqualified for informal lending credit-worthy and provide access to informal lending and fulfills their unmet credit demand. However, it is not clear how MFI competition may affect the demand for and supply of informal loans.

\end{document}

\subsection{The MFI offers consumptin loans}

Consider a scenario where the MFIs offer consumption loans as well.  We consider a candidate equilibrium where along the equilibrium path all MFIs initially offer a loan of $(\frac{k^C(n)}{n},0)$,   and in case the borrower reports a consumption shock in stage 2, offers an additional loan of $\frac{S-f(k^C(n))-k^C(n)\}}{n}$. The borrower is supposed to repay $\frac{k^C(n)}{n}$ in case there is no consumption shock, and $\frac{[S-\{f(k^C)-k^C\}]+k^C}{n}$ otherwise.  In order to sustain this as an equilibrium, it is necessary that the borrower has no incentive to ask for a consumption loan even if a shock has not happened, and then default on a single MFI, while repaying  the others. The DIC constraint that ensures that this does not happen is given by
{\footnotesize \begin{eqnarray}
&&f(k^C)-k^C+\frac{\delta}{1-\delta}\left[p[f(k^C)-k^C]-(1-p)[S-\{f(k^C)-k^C\}]\right]\nonumber\\
&\ge &f(k^C)-k^C+\frac{k^C+[S-\{f(k^C)-k^C\}]}{n}+\frac{\delta}{1-\delta}[p[f(\frac{(n-1)k^C}{n})-\frac{(n-1)k^C}{n}]-(1-p)[S-\{f(\frac{(n-1)k^C}{n})-\frac{(n-1)k^C}{n}].\nonumber\\\end{eqnarray}}
Re-arranging, we have that
\begin{eqnarray}
 \frac{\delta}{1-\delta}\frac{f(k^C)-f(\frac{(n-1)k^C}{n})}{\frac{k^C}{n}}\ge 1+\frac{\delta}{1-\delta}+ \frac{S-\{f(k^C)-k^C\}}{k^C}.\label{eqcons}\end{eqnarray}
Next suppose that that $\overline k(n)=k^{DIC(n,n-1)}$, so that   (\ref{11}) binds. We then have that
\begin{eqnarray}
 \frac{\delta}{1-\delta}\frac{f(\overline k)-f(\frac{(n-1)\overline k}{n})}{\frac{\overline k}{n}}= \frac{\delta}{1-\delta}+ \frac{m(1+c)}{p+(1-p)m(1+c)}.\label{30}\end{eqnarray}
Comparing  (\ref{eqcons}) and (\ref{30}), it is clear that $K^C(n)<\overline k(n)$, whenever the costs of mitigating the shock $S$ is large, i.e.
$$1+\frac{S-\{f(k^E)-k^E\}}{k^E}> \frac{m(1+c)}{p+(1-p)m(1+C)}.$$
We can mimic the preceding argument to show that borrower's welfare is higher in case consumption loans are not allowed  whenever
\begin{eqnarray}
f(\overline k)-\overline k -m(1+c)(1-p)S>f(k^C)-k^C-(1-p)S.\end{eqnarray}
Of course, as discussed earlier, there can be other factors that ensure that MFIs do not want to provide consumption loans.

  Example where borrowers worse off in case of consumption shock: ?????

\newpage

\end{document}